\documentclass[a4paper,UKenglish,cleveref, autoref, thm-restate]{lipics-v2019}

\title{A survey of BWT variants for string collections}

\author{Davide Cenzato}{Department of Environmental Sciences, Informatics and Statistics, Ca' Foscari University, Venice, Italy}{davide.cenzato@unive.it}{https://orcid.org/0000-0002-0098-3620}{}

\author{{Zs}uzsanna Lipt{\'a}k}{Department of Computer Science, University of Verona, Verona, Italy}{zsuzsanna.liptak@univr.it}{https://orcid.org/0000-0002-3233-0691}{}

\authorrunning{D.\ Cenzato and Zs.\ Lipt{\'a}k} 

\Copyright{Davide Cenzato, Zsuzsanna Lipt{\'a}k} 

\ccsdesc[100]{Theory of computation $\rightarrow$ Design and analysis of algorithms $\rightarrow$ Data structures design and analysis $\rightarrow$ Data compression; Applied computing $\rightarrow$ Life and medical sciences $\rightarrow$ Bioinformatics} 

\keywords{Burrows-Wheeler-Transform, extended BWT, string collections, repetitiveness measures, number of runs, compression} 

\category{} 

\relatedversion{} 

\supplement{The data and scripts are available at\\ \url{https://github.com/davidecenzato/BWT-variants-for-string-collections}.}

\funding{DC is funded by the European Union (ERC, REGINDEX, 101039208). Views and opinions expressed are however those of the author(s) only and do not necessarily reflect those of the European Union or the European Research Council. Neither the European Union nor the granting authority can be held responsible for them. ZsL is partially funded by the MUR PRIN Project {\it PINC, Pangenome INformatiCs: from Theory to Applications} (Grant no.\ 2022YRB97K). }

\acknowledgements{We would like to thank Massimiliano Rossi who supplied us with some cleaned and filtered datasets.}

\nolinenumbers 

\hideLIPIcs  

\usepackage{adjustbox}
\usepackage{float}

\usepackage[utf8]{inputenc}
\usepackage[group-separator={,}]{siunitx}
\usepackage[useregional]{datetime2}
\usepackage{float}
\usepackage[section]{placeins}

\usepackage{todonotes}

\graphicspath{{Fig/}}

\newcommand{\BWT}{\ensuremath{\mathrm{BWT }}}

\renewcommand{\root}{\ensuremath{\mathrm{root}}}
\renewcommand{\exp}{\ensuremath{\mathrm{exp}}}
\newcommand{\lex}{\ensuremath{\mathrm{lex}}}
\newcommand{\clex}{\ensuremath{\mathrm{colex}}}
\newcommand{\rev}{\ensuremath{\mathrm{rev}}}
\newcommand{\runs}{\ensuremath{\mathrm{runs}}}

\newcommand\Hdist{\ensuremath{\mathrm{dist_H}}}
\newcommand\edist{\ensuremath{\mathrm{dist_{\text edit}}}}

\def\chr19{\texttt{chr19}}

\def\ours{\texttt{pfpebwt}} 
\def\pfpebwt{\texttt{pfpebwt}}
\def\rope{\texttt{ropebwt2}}

\def\eBWT{\ensuremath{\mathrm{eBWT }}}
\def\dollarebwt{\ensuremath{\textrm{dolEBWT}}}
\def\concat{\ensuremath{\textrm{concBWT}}}
\def\mdollar{\ensuremath{\textrm{mdolBWT}}}
\def\rlo{\ensuremath{\textrm{colexBWT}}}
\def\colex{\ensuremath{\textrm{colexBWT}}}
\def\optimal{\ensuremath{\textrm{optBWT}}}

\colorlet{darkred}{red!80!black}
\colorlet{darkblue}{blue!60!black}
\colorlet{darkgreen}{green!60!black}
\colorlet{darkgray}{white!50!black}
\newcommand{\blue}[1]{{\color{darkblue} #1}}
\newcommand{\green}[1]{{\color{darkgreen} #1}}
\newcommand{\red}[1]{{\color{darkred} #1}}
\newcommand{\gray}[1]{{\color{lightgray} #1}}
\newcommand{\orange}[1]{{\color{orange} #1}}

\newcommand*{\PlusPlus}{%
\kern0.3ex\raisebox{-0ex}{\sebox{0.8}{\kern-0.4ex+}}%
\kern-0ex\raisebox{0.5ex}{\sebox{0.8}{\kern-0.4ex+}}}

\usepackage{sidecap}
\usepackage{adjustbox}
\usepackage{float}
\usepackage{verbatim}
\usepackage{diagbox}
\newcolumntype{R}[1]{>{\raggedleft\arraybackslash}p{#1}}
\newcolumntype{L}[1]{>{\raggedright\arraybackslash}p{#1}}
\usepackage{siunitx}

\DeclareRobustCommand*\cal{\relax\mathcal}

\begin{document}

\maketitle

\begin{abstract}
In recent years, the focus of bioinformatics research has moved from individual sequences to collections of sequences. Given the fundamental role of the Burrows-Wheeler Transform (BWT) in string processing, a number of dedicated tools have been developed for computing the BWT of string collections. While the focus has been on improving efficiency, both in space and time, the exact definition of the BWT employed has not been at the center of attention. 
As we show in this paper, the different tools in use often compute non-equivalent BWT variants: the resulting transforms can differ from each other significantly, including the number $r$ of runs, a central parameter of the BWT. Moreover, with many tools, the transform depends on the input order of the collection. In other words, on the same dataset, the same tool may output different transforms if the dataset is given in a different order. 
We studied $18$ dedicated tools for computing the BWT of string collections and have been able to identify $6$ different BWT variants computed by these tools. We review the differences between these BWT variants, both from a theoretical and from a practical point of view, comparing them on $8$ real-life biological datasets with different characteristics. We find that the differences can be extensive, depending on the datasets, and are largest on collections of many similar short sequences. 
The parameter $r$, the number of runs of the BWT, also shows notable variation between the different BWT variants; on our datasets, it varied by a multiplicative factor of up to $4.2$. 
Source code and scripts to replicate the results and download the data used in the article are available at  \url{https://github.com/davidecenzato/BWT-variants-for-string-collections}
\end{abstract}


\section{Introduction}

The Burrows-Wheeler Transform \cite{BW94} (BWT) is a fundamental string transformation which is at the heart of many modern compressed data structures for text processing, in particular in bioinformatics~\cite{bowtie,bwa,bowtie2}. With the increasing availability of low-cost high-throughput sequencing technologies, the focus has moved from single strings to large string collections, such as the 1000 Genomes Project~\cite{1000genomes}, the 10,000 Genomes Project~\cite{10K}, the 100,000 Human Genome Project~\cite{100K}, the 1001 Arabidopsis Project \cite{arab}, and the 3,000 
Rice Genomes Project  (3K RGP) \cite{rice}. This has led to a widespread use of compressed data structures on inputs which are collections of sequences, rather than individual sequences. 

A number of tools have been developed in recent years for computing the \BWT\ of a collection (multiset) of strings. The focus has been on efficiently processing datasets of ever increasing sizes, but little attention has been paid to the actual method used to compute the BWT. This is an issue, as the BWT was originally defined for a single string, and it is not immediately clear how to define it for a collection (multiset) of strings. In fact, there exists more than one way to compute a Burrows-Wheeler-type transform of multiple strings. Even though all these methods maintain the properties necessary for building string indexes on top of the \BWT, such as reversibility and LF-property, they differ in other, important, ways. 

As we will show in this paper, different tools not only apply different algorithms to compute the \BWT\ of the input collection, but they output different transforms. Studying $18$ publicly available tools, we identified six distinct \BWT-variants which are computed by these tools. The tools included in this study are: {\tt BEETL~\cite{BauerCR13}, BCR\_LCP\_GSA}~\cite{BauerCR13}, {\tt ropebwt2}~\cite{Li14a}, {\tt nvSetBWT}~\cite{Pantaleoni14}, {\tt msbwt}~\cite{HoltM14}, {\tt Merge-BWT}~\cite{Siren16},  {\tt eGSA}~\cite{Louza2017d},  {\tt BigBWT}~\cite{BoucherGKLMM19}, {\tt bwt-lcp-parallel}~\cite{BonizzoniVPPR19}, 
{\tt eGAP}~\cite{EgidiLMT19}, {\tt gsufsort}~\cite{LouzaTGPR20}, {\tt G2BWT}~\cite{DominguezN21}, {\tt grlBWT}~\cite{DominguezN22}, \pfpebwt~\cite{BoucherCLMM21}, {\tt cais}~\cite{BoucherCLMM21}, {\tt r-pfbwt}~\cite{OlivaGB23}, {\tt CMS-BWT}~\cite{Masillo23}, and {\tt optimalBWT}~\cite{CenzatoGLR23}. As a first example, in Table~\ref{table:bwts}, we give the BWT variants as computed by these $18$ tools on a toy example of $5$ DNA-strings. 

The size of BWT-based compressed data structures such as the RLFM-index~\cite{MakinenN05} or the $r$-index~\cite{GagieNP20} is typically measured in the number of runs (maximal substrings consisting of the same letter) of the BWT, commonly denoted $r$. This parameter $r$ has become central as a measure of the storage space required by these data structures. Additionally, much recent research effort has concentrated on the construction of data structures which can not only store but query, process, and mine strings in space and time proportional to $r$~\cite{GagieNP20,BannaiGI20,Oliva0SMKGB21,CobasGN21}. Moreover, the parameter $r$ (or the related $n/r$, the average runlength of the \BWT) is also being increasingly seen as a measure of repetitiveness of the string or strings, with several recent works theoretically exploring its suitability as such a measure, as well as its relationship to other repetitiveness measures~\cite{KempaK20,Navarro21a,GiulianiILPST21,AkagiFI22}. The parameter $r$ is now also being used as a property of the dataset itself (e.g.~\cite{CobasGN21,BannaiGI20,BoucherCGHMNR21}). 

However, the number of runs varies between the different \BWT-variants, as can be seen on our toy example. This has important implications not only for the storage space required for \BWT-based compressed data structures, but also for claims about the level of repetitiveness of the dataset. With competing non-equivalent methods around, this measure is not well defined. We will explore this issue further (Section~\ref{sec:effects_on_r}) and suggest solving the problem by standarizating the definition.

\noindent
\begin{table*}[t]
\centering
\begin{tabular}{l| l| l}
variant & result on example & tools \\
\hline\hline
eBWT & {\tt CGGGATGTACGTTAAAAA} & \ours~\cite{BoucherCLMM21}, {\tt cais}~\cite{BoucherCLMM21}\\
\dollarebwt & {\tt GGAAACGG\$\$\$TTACTGT\$AAA\$} & {\tt G2BWT}~\cite{DominguezN21}, \pfpebwt~\cite{BoucherCLMM21}, {\tt cais}~\cite{BoucherCLMM21}, {\tt msbwt}~\cite{HoltM14}\\
\mdollar & {\tt GAGAAGCG\$\$\$TTATCTG\$AAA\$} & {\tt BEETL}~\cite{BauerCR13}, {\tt ropebwt2}~\cite{Li14a}, {\tt nvSetBWT}~\cite{Pantaleoni14},  
 {\tt Merge-BWT}~\cite{Siren16},\\
 && {\tt eGSA}~\cite{Louza2017d}, {\tt eGAP}~\cite{EgidiLMT19}, {\tt bwt-lcp-parallel}~\cite{BonizzoniVPPR19}, \\
 && {\tt gsufsort}~\cite{LouzaTGPR20}, 
 {\tt grlBWT}~\cite{DominguezN22}, {\tt BCR\_LCP\_GSA}~\cite{BauerCR13} \\
\concat & {\tt \$AAGAGGGC\$\#\$TTACTGT\$AAA\$} & {\tt BigBWT}~\cite{BoucherGKLMM19}, {\tt r-pfbwt}~\cite{OlivaGB23}, {\tt CMS-BWT}~\cite{Masillo23},\\ && tools for single-string BWT\\
\rlo & {\tt AAAGGCGG\$\$\$TTACTGT\$AAA\$} & \rope~\cite{Li14a}\\
\optimal & {\tt AAAGGGGC\$\$\$TTACTTG\$AAA\$} & {\tt optimalBWT}~\cite{CenzatoGLR23}
\end{tabular}
\vspace{2mm}
\caption{\label{table:bwts}The different \BWT\ variants on the multiset ${\cal M} = \{ {\tt \texttt{ATATG}, \texttt{TGA} ,\texttt{ACG} ,\texttt{ATCA}, \texttt{GGA} }\}$. The \colex\ and the \optimal\ are special cases of the \mdollar. For detailed explanations, see Section~\ref{sec:bwtvariants}. Double listing due to options of the software. 
}
\end{table*}

\subsection{Overview of methods for defining multi-string \BWT}\label{sec:methods}

The classical way of computing text indexes of more than one string is to concatenate them, adding a different end-of-string-symbol at the end of each string, and then to compute the index for the concatenated string. This is the method traditionally used for generating classical data structures such as suffix trees and suffix arrays for multiple strings, and results in the so-called {\em generalized suffix tree} resp.\ {\em generalized suffix array} (see e.g.~\cite{Gusfield1997,Ohlebusch2013}). Applied directly, this method would lead to an unacceptable increase in the size of the alphabet, from $\sigma$, often a small constant in applications, to $\sigma + k$, where $k$ is the number of elements in the collection, typically in the thousands or even tens or hundreds of thousands. One way to avoid this is to use only conceptually different end-of-string-symbols, i.e.\ to have only one dollar-sign and apply string input order to break ties. This is the method employed by \rope~\cite{Li14a}, {\tt BCR}~\cite{BauerCR13}, and many others.\footnote{It should be noted that the tools listed in Table~\ref{table:bwts} may use different algorithms from the ones explained, as is the case for {\tt BCR}; however, they are equivalent w.r.t.\ the resulting transform.} Another method to avoid increasing the alphabet size is to separate the input strings using the same end-of-string-symbol; in this case, a different end-of-string-symbol has to be added to the end of the concatenated string to ensure correctness, as e.g.\ done by {\tt BigBWT}~\cite{BoucherGKLMM19}. An equivalent solution is to concatenate the input strings without removing the end-of-line or end-of-file characters, since these act as separators; or to concatenate them without separators and use a bitvector to mark the end of each string. Many studies nowadays use string collections in experiments without turning to dedicated tools for multi-string BWT (e.g.~\cite{PuglisiZ21,BannaiGI20, KuhnleMBGLM19}); often the input strings are turned into one single sequence using one of the methods described above, and then the single-string \BWT\ is computed; it is, however, not always stated explicitly which was the method used to obtain one sequence. Underlying this is the implicit assumption that all methods are equivalent. 

In 2007, Mantaci et al.~\cite{MantaciRRS07} introduced the {\em extended Burrows-Wheeler Transform} (\eBWT), which generalizes the \BWT\ to a multiset of strings. The \eBWT, like the \BWT, is reversible, and maintains other properties of the BWT such as fast pattern matching functionality. Since then, however, the term 'extended \BWT' has come to be used as a generic term to denote the BWT of a collection of sequences. This is unfortunate, as the \eBWT\ has several properties such as independence of the input order which the other methods do not share; and it is defined using a different order relation from the classical \BWT\ (omega-order rather than lexicographic order, see Section~\ref{sec:preliminaries} for exact definitions).

Two tools compute the \eBWT\ according to the original definition, \pfpebwt\ and {\tt cais} (both~\cite{BoucherCLMM21}); all others append an end-of-string character to the input strings, explicitly or implicitly, and as a consequence, the resulting transforms differ from the one defined in~\cite{MantaciRRS07}. Moreover, the output in most cases depends on the input order of the sequences (except for those tools that compute what we term \dollarebwt, \colex, or \optimal). 
The exact nature of this dependence differs from one transform to another (see Section~\ref{sec:permutations}). 

The result is that the \BWT\ variants computed by different tools on the same dataset, or by the same tool on the same dataset but given in a different order, may vary considerably.

\subsection{Multidollar \BWT}\label{sec:mdollar}

As we will see, the \BWT-variant which we term \mdollar\ is the most general one, in the sense that all others, except for the \eBWT, can be simulated by \mdollar\ (Proposition~\ref{prop:2}). This is the variant output by most tools, and it is dependent on the input order: both the transform itself and the number of runs varies depending on the order in which the strings are concatenated. Bentley, Gibney, and Thankachan recently gave a linear-time algorithm for computing \mdollar\ with the minimum number of runs amongst all input string orders~\cite{BentleyGT20}. In order to study the variation of the parameter $r$, we first implemented a variant of this algorithm counting the minimal number of runs, which later led to a new tool, {\tt optimalBWT}, computing the \optimal\ ~\cite{CenzatoGLR23}. We use this tool in our experiments as a baseline for the number of runs of the other \BWT-variants. On our real-life biological datasets, the parameter $r$ varies by up to a multiplicative factor of $4.2$ between the different variants. It was shown in~\cite{CenzatoGLR23} that an improvement by a multiplicative factor of up to $31$ can be obtained between the input order and \optimal.


\subsection{Our contributions}

\begin{enumerate}
    \item We identify six distinct BWT variants which are computed by $18$ publicly available tools, specifically designed for string collections. We formally describe the differences between these, identifying specific intervals to which differences are restricted. 
    \item We show the influence of the input order on the output, in dependence of the BWT variant. 
    \item We describe the impact on the number $r$ of runs of the BWT and give an upper bound on the amount by which the colexicographic order (sometimes referred to as 'reverse lexicographic order') can differ from the optimal order of Bentley et al.~\cite{BentleyGT20}. 
    \item We complement our theoretical analysis with extensive experiments, comparing the BWT variants on eight real-life datasets with different characteristics. 
    \item We suggest a way of standardizing the parameter $r$, and thus to eliminate the ambiguity caused by the presence of different \BWT-variants. 
\end{enumerate}

To the best of our knowledge, this is the first systematic treatment of the different BWT variants in use for collections of strings.


\subsection{What is not covered} 

This paper deals with tools for string collections, so we did not include any tool that computes the BWT of a single string, such as libdivsufsort~\cite{libdivsufsort}, sais-lite-lcp~\cite{sais-lite-lcp}, libsais~\cite{libsais}, or bwtdisk~\cite{FerraginaGM12}.  Although in many cases, these are the tools used for collections of strings, the transform they compute depends on the method with which the string collection was turned into a single string, as explained above. Nor did we include other BWT variants for single strings such as the bijective BWT~\cite{GilScott12,KopplHHS20}, since, again, these were not designed for string collections. 

The Big-xBWT~\cite{GagieGM21} is a tool for compressing and indexing read collections, using the xBWT of Ferragina et al.~\cite{xbwt,FerraginaLMM09}. In addition  to the string collection, it requires a reference sequence as input, in contrast to the other tools. Moreover, the output is not comparable either, since its length can vary---as opposed to all other BWT variants we review, the xBWT is not a permutation of the input characters but can be shorter, due to the fact that it first maps the input to a tree and then applies the xBWT to it, a BWT-like index for labeled trees, rather than for strings. Likewise, the tool~\cite{OhlebuschSB18} for reference-free xBWT is not included in this review: even though it does not require a reference sequence, it, too, computes the xBWT and not the \BWT. 
Finally, since the method for concatenating the input strings used in~\cite{CazauxR19} (using the same separator symbol but without an additional end-of-string character) differs from all BWT variants that have been implemented by some tool, we excluded it from the current study.


\subsection{Overview} 

We give the necessary definitions in Section~\ref{sec:preliminaries}. In Section~\ref{sec:bwtvariants}, we present the BWT variants and analyse their differences, followed by an in-depth analysis of the relationship between the input and the output order of the strings in the collection, for the different \BWT-variants, in Section~\ref{sec:permutations}. In Section~\ref{sec:effects_on_r}, we discuss the effects on the repetitiveness measure $r$. A summary of our experimental results is presented in Section~\ref{sec:experimental}. We draw some conclusions from our study in Section~\ref{sec:conclusion}. The full tables with detailed results on all eight datasets are included in the Appendix. 

A preliminary version of this article appeared in~\cite{CenzatoL22}. 
Source code and scripts to replicate the results and download the data used in the article are available at  \url{https://github.com/davidecenzato/BWT-variants-for-string-collections}. 


\section{Preliminaries}\label{sec:preliminaries}

Let $\Sigma$ be a finite ordered alphabet of size $\sigma$. We use the notation $T=T[1..n]$ for a string $T$ of length $n$ over $\Sigma$, $T[i]$ for the $i$th character, and $T[i..j]$ for the substring $T[i]\cdots T[j]$ of $T$, where $i\leq j$; the length of $T$ is denoted $|T|$, and the empty string is denoted $\varepsilon$. For a string $T$ over $\Sigma$ and an integer $m>0$, we write $T^m$ for the $m$-fold concatenation of $T$. A string $T$ is called {\em primitive} if $T=U^m$ implies $T=U$ and $m=1$. Every string $T$ can be written uniquely as $T=U^m$, where $U$ is primitive; in this case, we refer to $U$ as $\root(T)$ and to $m$ as $\exp(T)$. In other words, for every string $T$ it holds that $T = \root(T)^{\exp(T)}$. Often, an end-of-string character (usually denoted {\tt \$}) is appended to the end of $T$; this character is not an element of $\Sigma$ and is assumed to be smaller than all characters from $\Sigma$. Note that appending a {\tt \$} makes any string primitive. 

String $S$ is a {\em conjugate} of string $T$ if $S= T[i..n]T[1..i-1]$ for some $i\in \{1,\ldots, n\}$ (also called the {\em $i$th  rotation}  of $T$). It is easy to see that a string of length $n$ has $n$ distinct conjugates if and only if it is primitive. A {\em run} in string $T$ is a maximal substring consisting of the same character; we denote by $\runs(T)$ the number of runs of $T$. For example, $\runs({\tt CAAGGGA}) = 4$. 

For two strings $S,T$, the {\em (unit-cost) edit distance}, or {\em Levenshtein distance}, $\edist(S,T)$ is defined as the minimum number of operations necessary to transform $S$ into $T$, where an operation can be deletion or insertion of a character, or substitution of a character by another. The {\em Hamming distance} $\Hdist(S,T)$, defined only if $|S|=|T|$, is the number of positions $i$ such that $S[i] \neq T[i]$. 

The {\em lexicographic order} on $\Sigma^*$ is defined as follows: $S<_{\lex} T$ if $S$ is a proper prefix of $T$, or if there exists an index $j$ s.t.\ $S[j]<T[j]$ and for all $i<j$, $S[i]=T[i]$. The {\em colexicographic order}, or {\em colex-order} (referred to as {\em reverse lexicographic order} or {\em rlo} in~\cite{Li14a, CoxBJR12}), is defined as follows: $S<_{\clex}T$ if $S^{\rev} <_{\lex} T^{\rev}$, where $X^{\rev}=X[n]X[n-1]\cdots X[1]$ denotes the reverse of the string $X=X[1..n]$.

For a string $T=T[1..n]$ over $\Sigma$, the {\em Burrows-Wheeler Transform}~\cite{BW94}, $\BWT(T)$, is a permutation of the characters of $T$, given by concatenating the last characters of the lexicographically sorted conjugates of $T$. In Fig.~\ref{fig:bwt-ex} we give three examples: $\BWT({\tt CAGAGA}) = {\tt GGCAAA}$, $\BWT({\tt CACACA}) = {\tt CCCAAA}$, and $\BWT({\tt CAGAGA\$}) = {\tt AGGC\$AA}$.

\begin{table}[htbp]
\centering
\begin{tabular}{c@{\hspace{.2cm}}c@{\hspace{.8cm}}c@{\hspace{.2cm}}c@{\hspace{.8cm}}c@{\hspace{.2cm}}c}
${\tt CAGAGA}$ & \BWT & ${\tt CACACA}$ & \BWT & ${\tt CAGAGA\$}$ & \BWT\\
\hline
{\tt ACAGAG} & {\tt G} & {\tt ACACAC} & {\tt C} & {\tt \$CAGAGA} & {\tt A}\\
{\tt AGACAG} & {\tt G} & {\tt ACACAC} & {\tt C} & {\tt A\$CAGAG} & {\tt G}\\
{\tt AGAGAC} & {\tt C} & {\tt ACACAC} & {\tt C} & {\tt AGA\$CAG} & {\tt G}\\
{\tt CAGAGA} & {\tt A} & {\tt CACACA} & {\tt A} & {\tt AGAGA\$C} & {\tt C}\\
{\tt GACAGA} & {\tt A} & {\tt CACACA} & {\tt A} & {\tt CAGAGA\$} & {\tt \$}\\
{\tt GAGACA} & {\tt A} & {\tt CACACA} & {\tt A} & {\tt GA\$CAGA} & {\tt A}\\
&&&& {\tt GAGA\$CA} & {\tt A}\\
\end{tabular}
\caption{BWT of the strings {\tt CAGAGA, CACACA} and {\tt CAGAGA\$}. \label{fig:bwt-ex}}
\end{table}

It follows from the definition of the BWT that two strings $S,T$ are conjugates if and only if $\BWT(S) = \BWT(T)$. 
Indeed, the BWT is reversible up to conjugates: if a string $L$ is the BWT of some string $T$, then a string $S$ can be computed in linear time such that $L=\BWT(S)$, and thus $S$ is a conjugate of $T$. 
To make the BWT uniquely reversible, one can add an index to it, marking the lexicographic rank of the conjugate in input. For example, $\BWT({\tt CAGAGA}) = {\tt GGCAAA}$, and the index $4$ specifies that the input was the 4th conjugate in lexicographic order. Alternatively, one adds a {\tt \$} to the end of $T$, which makes the input unique: $\BWT({\tt CAGAGA\$}) = {\tt AGGC\$AA}$, and ${\tt CAGAGA\$}$ is the only string ending in ${\tt \$}$ with this BWT. Note that $\BWT$ with and without end-of-string symbol can be quite different, as the example shows. 

An important parameter of the $\BWT$ of string $T$ is the number of runs $r(T) = \runs(\BWT(T))$. It is well-known that on repetitive inputs, the BWT tends to produce long runs of the same character, making it amenable to compression via runlength-encoding (RLE). In our example, $r({\tt CAGAGA})=3$, while the original string has 6 runs.  This property, referred to as the {\em clustering effect} of the BWT, is taken advantage of by compressed data structures such as the RLFM-index~\cite{MakinenN05} or the $r$-index~\cite{GagieNP20}. 

Next we define the {\em omega-order}~\cite{MantaciRRS07} on $\Sigma^*$: $S \prec_\omega T$ if $\root(S) = \root(T)$ and $ \exp(S) < \exp(T)$, or if $S^\omega <_\lex T^\omega$ (implying $\root(S) \neq \root(T)$), where $T^{\omega}$ denotes the infinite string obtained by concatenating $T$ infinitely many times. The omega-order relation coincides with the lexicographic order if neither of the two strings is a proper prefix of the other. The two orders can differ otherwise, e.g.\  ${\tt GT} <_\lex {\tt GTC}$ but ${\tt GTC} \prec_\omega {\tt GT}$. 

For a multiset of strings ${\cal M}=\{T_1, \ldots, T_k\}$, the {\em extended Burrows-Wheeler Transform}, $\eBWT({\cal M})$~\cite{MantaciRRS07}, is a permutation of the characters of the strings in ${\cal M}$, given by concatenating the last characters of the conjugates of each $T_i$, for $i=1,\ldots,k$, listed in omega-order. For example, the omega-sorted conjugates of ${\cal M} = \{{\tt GTC, GT}\}$ are: {\tt CGT,} {\tt GTC,} {\tt GT,} {\tt TCG,} {\tt TG}, hence, $\eBWT({\cal M}) = {\tt TCTGG}$, see Figure~\ref{tab:ebwt-ex}. Again, adding the indices of the input conjugates, in this case $2,3$, makes the \eBWT\ uniquely reversible. For more on the \eBWT, see~\cite{MantaciRRS07}. 

\begin{table}
\centering
\begin{tabular}{l@{\hspace{.2cm}}c@{\hspace{.8cm}}l@{\hspace{.2cm}}c}
$\{{\tt GTC,GT}\}$ & \eBWT & $\{{\tt GTC\$,GT\$}\}$ & \eBWT\\
\hline
{\tt CGT} & {\tt T} & {\tt \$GT} & {\tt T}\\
{\tt GTC} & {\tt C} & {\tt \$GTC} & {\tt C}\\
{\tt GT} & {\tt T} & {\tt C\$GT} & {\tt T}\\
{\tt TCG} & {\tt G} & {\tt GT\$} & {\tt \$}\\
{\tt TG} & {\tt G} & {\tt GTC\$} & {\tt \$}\\
&& {\tt T\$G} & {\tt G}\\
&& {\tt TC\$G} & {\tt G}\\
\end{tabular}
\caption{The \eBWT\ of the set $\{{\tt GTC,GT}\}$, and of the set $\{{\tt GTC\$,GT\$}\}$ (see Section~\ref{sec:bwtvariants}).}\label{tab:ebwt-ex}
\end{table}


\section{BWT variants for string collections}\label{sec:bwtvariants}

We identified six distinct transforms, listed below, which were computed by the tools given in Table~\ref{table:bwts}. Let ${\cal M} = \{T_1,\ldots, T_k\}$ be a multiset of strings, with total length $N_{\cal M}=\sum_{i=1}^k |T_i|$. Since several of the data structures depend on the order in which the strings are listed, we implicitly regard ${\cal M}$ as a list $[T_1,\ldots, T_k]$, and write $\rho({\cal M})$ for a specific input order $\rho$ in which the strings are presented.

\begin{enumerate}
    \item {\bf extended BWT:} $\eBWT({\cal M})$ of~\cite{MantaciRRS07} (see Section~\ref{sec:preliminaries})
    \item {\bf dollar-eBWT:} $\dollarebwt({\cal M}) = \eBWT(\{T_i\$ \mid T_i \in {\cal M}\})$
    \item {\bf multidollar BWT:} $\mdollar({\cal M}) =$\\ $ \BWT(T_1\$_1T_2\$_2\cdots T_k\$_k)$, where dollars are assumed to be smaller than characters from $\Sigma$ and $\$_1<\$_2<\ldots < \$_k$ 
    \item {\bf concatenated BWT:} $\concat({\cal M}) =$\\ $ \BWT(T_1\$T_2\$\cdots T_k\$\#)$, where $\#<\$$
    \item {\bf colexicographic BWT:} $\colex({\cal M}) = \mdollar(\gamma({\cal M}))$,  where $\gamma$ is the colexicographic ('reverse lexicographic', rlo) order of the strings in ${\cal M}$.
    \item {\bf optimal BWT:} $\optimal({\cal M}) = \mdollar(opt({\cal M}))$,  where $opt({\cal M})$ is the order given by the algorithm of Bentley et al.~\cite{BentleyGT20}, which minimizes the number of runs.
\end{enumerate}

Because all BWT variants except the \eBWT\ use additional end-of-string symbols as string separators, we refer to these by the collective term {\em separator-based BWT variants}. In Table~\ref{tab:overview} we show the six transforms on our running example of $5$ DNA-strings, and give first properties of these transforms. For ease of exposition and comparison, we replaced all separator-symbols by the same dollar-sign {\tt \$}, even where, conceptually or concretely, different dollar-signs are assumed to terminate the individual strings. This is the case for \mdollar\ and its special cases, \colex\ and \optimal. Moreover, the \concat\ contains one additional character, the final end-of-string symbol, here denoted by {\tt \#}, which is smaller than all other characters; thus, the additional rotation starting with {\tt \#} is the smallest and results in an additional dollar in the first position of the transform. To facilitate the comparison with the other transforms, we remove this first symbol from \concat\ and replace the {\tt \#} by {\tt \$}.

It is important to point out that the programs listed in Table~\ref{table:bwts} do not necessarily use the definitions given here; however, in each case, the resulting transform is the one claimed, up to renaming or removing separator characters, see Sections~\ref{sec:separators} and~\ref{sec:interesting}. 

\begin{table*}[t]
\begin{center}
\begin{tabular}{l|l|l|c}
\BWT\ variant & example & order of shared suffixes & independent \\
& & & of input \\ 
& & & order? \\
\hline \hline
{\em non sep.-based} &&&\\
    $\eBWT({\cal M})$ & {\tt \red{C}\green{G}\red{G}\green{G}\green{A}T\red{G}TA\blue{C}\orange{G}\blue{T}\orange{T}\green{A}AAA\green{A}} & omega-order of strings & yes \\
    \hline
    {\em separator-based} &&&\\
    $\dollarebwt({\cal M})$ & {\tt \green{GGAAA}\red{CGG}\$\$\$TTA\blue{CT}\orange{GT}\$AAA\$} & lexicographic order of strings & yes    \\
    $\mdollar({\cal M})$ & {\tt \green{GAGAA}\red{GCG}\$\$\$TTA\blue{TC}\orange{TG}\$AAA\$} & input order of strings & no\\
    $\concat({\cal M})$ & {\tt \green{AAGAG}\red{GGC}\$\$\$TTA\blue{CT}\orange{GT}\$AAA\$} & lexicographic order of & \\
    &&  \quad    subsequent strings in input & no  \\
    $\rlo({\cal M})$ & {\tt \green{AAAGG}\red{CGG}\$\$\$TTA\blue{CT}\orange{GT}\$AAA\$} & colexicographic order of strings & yes \\
    $\optimal({\cal M})$ &  {\tt \green{AAAGG}\red{GGC}\$\$\$TTA\blue{CT}\orange{TG}\$AAA\$} & order given by Bentley et al.'s  & \\
    & & \quad algorithm~\cite{BentleyGT20} & yes
\end{tabular}
\end{center}
\vspace{2mm}
\caption{Overview of some properties of the six \BWT\ variants considered in this paper. The colors in the example BWTs correspond to interesting intervals in separator-based variants, while the same characters are highlighted in the \eBWT\ for showing their positions, see Section~\ref{sec:interesting}.\label{tab:overview}}
\end{table*}


\subsection{The effect of adding separator symbols}\label{sec:separators} 

The first obvious difference between the \eBWT\ and the separator-based variants is their length: $\eBWT({\cal M})$ has length $N_{\cal M}$, while all other variants have length $N_{\cal M}+k$, since they contain an additional character (the separator) for each input string. 

In all separator-based transforms, the $k$-length prefix consists of a permutation of the last characters of the input strings. This is because the rotations starting with the dollars are the first $k$ lexicographically. On the other hand, in the \eBWT, these $k$ characters occur interspersed with the rest of the transform; namely, in the positions corresponding to the omega-ranks of the input strings $T_i$ (see Table~\ref{tab:overview}). 

In gerenal, adding a {\tt \$} to the end of the strings introduces a distinction, not present in the \eBWT, between suffixes and other substrings: since the separators are smaller than all other characters, occurrences of a substring as suffix will be listed en bloc before all other occurrences of the same substring, while in the \eBWT, these occurrences are listed interspersed with the other occurrences of the same substring. 

    \begin{example} Let ${\cal M} = \{{\tt AACGAC, TCAC}\}$ and $U={\tt AC}$. $U$ occurs both as a suffix and as an internal factor; the characters preceding it are {\tt A} (internal substring) and {\tt C,G} (suffix), and we have \eBWT$({\cal M}) =$ {\tt C\red{GAC}ATAACC}, \dollarebwt{$({\cal M}) =$} {\tt CC\$\red{GCA}AATAC\$}.
\end{example}

Finally, it should be noted that adding end-of-string symbols to the input strings changes the definition of the order applied. As observed above, the omega-order coincides with the lexicographic order on all pairs of strings $S,T$ where neither is a proper prefix of the other; but with end-of-strings characters, no input string can be a proper prefix of another. Thus, on rotations of the $T_i\$$'s, the omega-order equals the lexicographic order. As an example, consider the multiset ${\cal M} = \{${\tt GTC\$, GT\$}$\}$ from Section~\ref{sec:preliminaries}: we have the following omega-order among the rotations: {\tt \$GT, \$GTC, C\$GT, GT\$, GTC\$, T\$G, TC\$G} (see Table~\ref{tab:ebwt-ex}), which coincides with the lexicographic order. Similarly, adding {\em different} dollars {\tt \$}$_1$, {\tt \$}$_2$, \ldots, {\tt \$}$_k$ and applying the omega-order results again in the lexicographic order between the rotations, with different dollar symbols considered as distinct characters. This implies the following:

\begin{lemma}\label{obs:1}
Let ${\cal M} = \{T_1,T_2,\ldots,T_k\}$ be a string collection. Then 

\begin{enumerate}
\item $\dollarebwt({\cal M}) = \mdollar(\lex({\cal M}))$, where $\lex({\cal M})$ denotes the lexicographic order of the strings in ${\cal M}$; 
\item $\mdollar({\cal M}) = \eBWT(\{ T_i\$_i\mid i=1,\ldots,k\})$, up to re\-naming of dollars. 
\end{enumerate}

\end{lemma}

\begin{proof} 
For {\em 1.}, note that the order of the strings $T_i\$$ is the lexicographic order, since $T_i\$ <_\omega T_j\$$ if and only if $T_i\$ <_\lex T_j\$$ if and only if $T_i <_\lex T_j$. 
  For {\em 2.}, consider that no two rotations can be one prefix of the other, resulting in the lexicographic order between rotations, with the dollar-signs breaking ties. The only difference between the two transforms is now that for $i>1$, $\$_i$ in the right transform is replaced by $\$_{i-1}$ in the left transform, and $\$_1$ is replaced by $\$_k$. 
\end{proof}

Regarding the differences among the separator-based BWT variants, we will show that all differences occur in certain well-defined intervals of the BWT, and that the differences themselves depend only on a specific permutation of $\{1,\ldots,k\}$, given by the combination of the input order, the lexicographic order of the input strings, and the BWT variant applied. In Tables~\ref{tab:3bwts-table} and~\ref{tab:3bwts-table2}, we give the full \BWT\ matrices for all BWT variants.

\begin{table*}[h]
\centering
\begin{adjustbox}{max width=140mm}

\begin{tabular}{ rc|l| } 
 \hline
 \multicolumn{1}{|r|}{index} & \eBWT & \text{rotation} \\ 
 \hline
 \multicolumn{1}{|r|}{(4,4)} & \red{\texttt{C}} & \texttt{AATC}   \\ 
  \hline
 \multicolumn{1}{|r|}{(3,1)} & \green{\texttt{G}} & \texttt{ACG}   \\ 
  \hline
 \multicolumn{1}{|r|}{(5,3)} & \red{\texttt{G}} & \texttt{AGG}   \\ 
  \hline
 \multicolumn{1}{|r|}{(1,1)} & \green{\texttt{G}} & \texttt{ATATG}   \\
  \hline
 \multicolumn{1}{|r|}{(4,1)} & \green{\texttt{A}} & \texttt{ATCA}   \\ 
 \hline
 \multicolumn{1}{|r|}{(1,3)} & \texttt{T} & \texttt{ATGAT}   \\ 
 \hline
 \multicolumn{1}{|r|}{(2,3)} & \red{\texttt{G}} & \texttt{ATG}   \\ 
 \hline
 \multicolumn{1}{|r|}{(4,3)} & \texttt{T} & \texttt{CAAT}  \\ 
 \hline
 \multicolumn{1}{|r|}{(3,2)} & \texttt{A} & \texttt{CGA}   \\ 
 \hline
 \multicolumn{1}{|r|}{(3,3)} & \blue{\texttt{C}} & \texttt{GAC}   \\ 
 \hline
 \multicolumn{1}{|r|}{(5,2)} & \orange{\texttt{G}} & \texttt{GAG}   \\ 
  \hline
 \multicolumn{1}{|r|}{(1,5)} & \blue{\texttt{T}} & \texttt{GATAT}   \\ 
  \hline
 \multicolumn{1}{|r|}{(2,2)} & \orange{\texttt{T}} & \texttt{GAT}   \\ 
  \hline
 \multicolumn{1}{|r|}{(5,1)} & \green{\texttt{A}} & \texttt{GGA}   \\ 
 \hline
 \multicolumn{1}{|r|}{(1,2)} & \texttt{A} & \texttt{TATGA}   \\ 
 \hline
 \multicolumn{1}{|r|}{(4,2)} & \texttt{A} & \texttt{TCAA}   \\ 
 \hline
 \multicolumn{1}{|r|}{(1,4)} & \texttt{A} & \texttt{TGATA}   \\ 
 \hline
 \multicolumn{1}{|r|}{(2,1)} & \green{\texttt{A}} & \texttt{TGA}   \\ 
 \hline
 
\vspace{15.8mm}
\end{tabular}
\quad 
\begin{tabular}{ |r|c|l| } 
 \hline
 \text{index} & dolE & \text{rotation} \\ 
 \hline
 (3,4) & \green{\texttt{G}} & \texttt{\$ACG}   \\ 
 \hline
 (1,6) & \green{\texttt{G}} & \texttt{\$ATATG}   \\ 
 \hline
 (4,5) & \green{\texttt{A}} & \texttt{\$ATCA}   \\ 
 \hline
 (5,4) & \green{\texttt{A}} & \texttt{\$GGA}  \\ 
 \hline
 (2,4) & \green{\texttt{A}} & \texttt{\$TGA}  \\ 
 \hline
 (4,4) & \red{\texttt{C}} & \texttt{A\$ATC}   \\ 
 \hline
 (5,3) & \red{\texttt{G}} & \texttt{A\$GG}   \\ 
  \hline
 (2,3) & \red{\texttt{G}} & \texttt{A\$TG}   \\ 
 \hline
 (3,1) & \texttt{\$} & \texttt{ACG\$}   \\ 
 \hline
 (1,1) & \texttt{\$} & \texttt{ATATG\$}   \\ 
 \hline
 (4,1) & \texttt{\$} & \texttt{ATCA\$}   \\ 
 \hline
 (1,3) & \texttt{T} & \texttt{ATG\$AT}   \\ 
  \hline
 (4,3) & \texttt{T} & \texttt{CA\$AT}   \\ 
  \hline
 (3,2) & \texttt{A} & \texttt{CG\$A}   \\ 
 \hline
 (3,3) & \blue{\texttt{C}} & \texttt{G\$AC}   \\
 \hline
 (1,5) & \blue{\texttt{T}} & \texttt{G\$ATAT}   \\
 \hline
 (5,2) & \orange{\texttt{G}} & \texttt{GA\$G}   \\
 \hline
 (2,2) & \orange{\texttt{T}} & \texttt{GA\$T}   \\
 \hline
 (5,1) & \texttt{\$} & \texttt{GGA\$}   \\ 
 \hline
 (1,2) & \texttt{A} & \texttt{TATG\$A}   \\ 
 \hline
 (4,2) & \texttt{A} & \texttt{TCA\$A}   \\ 
 \hline
 (1,4) & \texttt{A} & \texttt{TG\$ATA}   \\ 
 \hline
 (2,1) & \texttt{\$} & \texttt{TGA\$}   \\ 
 \hline
\end{tabular}
\quad
\begin{tabular}{ |r|c|l| } 
 \hline
 \text{index} & conc & \text{rotation} \\ 
 \hline
 23 & \green{\texttt{A}} & \texttt{\$\#\gray{ATATG\$TGA\$ACG\$ATCA\$GGA}}  \\ 
  \hline
 10 & \green{\texttt{A}} & \texttt{\$ACG\$ATCA\$GGA\$\#ATATG\$\gray{TGA}}   \\ 
  \hline
 14 & \green{\texttt{G}} & \texttt{\$ATCA\$GGA\$\#\gray{ATATG\$TGA\$ACG}}   \\ 
  \hline
 19 & \green{\texttt{A}} & \texttt{\$GGA\$\#\gray{ATATG\$TGA\$ACG\$ATCA}}   \\
 \hline
 6 & \green{\texttt{G}} & \texttt{\$TGA\$ACG\$ATCA\$GGA\$\#\gray{ATATG}}   \\
 \hline
 22 & \red{\texttt{G}} & \texttt{A\$\#\gray{ATATG\$TGA\$ACG\$ATCA\$GG}}   \\
 \hline
 9 & \red{\texttt{G}} & \texttt{A\$ACG\$ATCA\$GGA\$\#\gray{ATATG\$TG}}   \\ 
 \hline
 18 & \red{\texttt{C}} & \texttt{A\$GGA\$\#\gray{ATATG\$TGA\$ACG\$ATC}}   \\
 \hline
 11 & \texttt{\$} & \texttt{ACG\$ATCA\$GGA\$\#\gray{ATATG\$TGA\$}} \\ 
  \hline
 1 & \texttt{\$} & \texttt{ATATG\$TGA\$ACG\$ATCA\$GGA\$\#}   \\ 
 \hline
 15 & \texttt{\$} & \texttt{ATCA\$GGA\$\#\gray{ATATG\$TGA\$ACG\$}} \\ 
 \hline
 3 & \texttt{T} & \texttt{ATG\$TGA\$ACG\$ATCA\$GGA\$\#\gray{AT}}  \\ 
  \hline
 17 & \texttt{T} & \texttt{CA\$GGA\$\#\gray{ATATG\$TGA\$ACG\$AT}}   \\ 
 \hline
 12 & \texttt{A} & \texttt{CG\$ATCA\$GGA\$\#\gray{ATATG\$TGA\$A}}  \\ 
 \hline
 13 & \blue{\texttt{C}} & \texttt{G\$ATCA\$GGA\$\#\gray{ATATG\$TGA\$AC}}   \\ 
 \hline
 5 & \blue{\texttt{T}} & \texttt{G\$TGA\$ACG\$ATCA\$GGA\$\#\gray{ATAT}}   \\
 \hline
 21 & \orange{\texttt{G}} & \texttt{GA\$\#\gray{ATATG\$TGA\$ACG\$ATCA\$G}} \\ 
 \hline
 8 & \orange{\texttt{T}} & \texttt{GA\$ACG\$ATCA\$GGA\$\#\gray{ATATG\$T}}   \\ 
 \hline
 20 & \texttt{\$} & \texttt{GGA\$\#\gray{ATATG\$TGA\$ACG\$ATCA\$}} \\ 
 \hline
 2 & \texttt{A} & \texttt{TATG\$TGA\$ACG\$ATCA\$GGA\$\#\gray{A}}   \\ 
 \hline
 16 & \texttt{A} & \texttt{TCA\$GGA\$\#\gray{ATATG\$TGA\$ACG\$A}}  \\ 
 \hline
 4 & \texttt{A} & \texttt{TG\$TGA\$ACG\$ATCA\$GGA\$\#\gray{ATA}}  \\ 
 \hline
 7 & \texttt{\$} & \texttt{TGA\$ACG\$ATCA\$GGA\$\#\gray{ATATG\$}}  \\ 
 \hline
\end{tabular}
\end{adjustbox}
\vspace{2mm}
\caption{From left to right we show the \eBWT, the \dollarebwt, and the \concat\ of the string collection ${\cal M} = \{ {\tt \texttt{ATATG},\texttt{TGA},\texttt{ACG},\texttt{ATCA},\texttt{GGA} }\}$. Indices are given with reference to the numbering $T_1 = {\tt ATATG}, T_2 = {\tt TGA}, T_3 = {\tt ACG}, T_4 = {\tt ATCA}, T_5 = {\tt GGA}$.}
\label{tab:3bwts-table}
\end{table*}

\begin{table*}[h]
\centering
\begin{adjustbox}{max width=140mm}
\begin{tabular}{ |r|c|l| } 
 \hline
 \text{index} & mdol & rotation \\ 
 \hline
 (1,6) & \green{\texttt{G}} & \texttt{\$$_1$\gray{ATATG}}   \\ 
 \hline
 (2,4) & \green{\texttt{A}} & \texttt{\$$_2$\gray{TGA}}  \\ 
 \hline
 (3,4) & \green{\texttt{G}} & \texttt{\$$_3$\gray{ACG}}   \\ 
 \hline
 (4,5) & \green{\texttt{A}} & \texttt{\$$_4$\gray{ATCA}}   \\ 
 \hline
 (5,4) & \green{\texttt{A}} & \texttt{\$$_5$\gray{GGA}}  \\ 
 \hline
 (2,3) & \red{\texttt{G}} & \texttt{A\$$_2$\gray{TG}}   \\ 
  \hline
 (4,4) & \red{\texttt{C}} & \texttt{A\$$_4$\gray{ATC}}   \\ 
 \hline
 (5,3) & \red{\texttt{G}} & \texttt{A\$$_5$\gray{GG}}   \\ 
 \hline
 (3,1) & \texttt{\$$_3$} & \texttt{ACG\$$_3$}   \\ 
 \hline
 (1,1) & \texttt{\$$_1$} & \texttt{ATATG\$$_1$}   \\ 
 \hline
 (4,1) & \texttt{\$$_4$} & \texttt{ATCA\$$_4$}   \\ 
 \hline
 (1,3) & \texttt{T} & \texttt{ATG\$$_1$\gray{AT}}   \\ 
  \hline
 (4,3) & \texttt{T} & \texttt{CA\$$_4$\gray{AT}}   \\ 
  \hline
 (3,2) & \texttt{A} & \texttt{CG\$$_3$\gray{A}}   \\ 
 \hline
 (1,5) & \blue{\texttt{T}} & \texttt{G\$$_1$\gray{ATAT}}   \\ 
 \hline
 (3,3) & \blue{\texttt{C}} & \texttt{G\$$_3$\gray{AC}}   \\ 
 \hline
 (2,2) & \orange{\texttt{T}} & \texttt{GA\$$_2$\gray{T}}   \\ 
 \hline
 (5,2) & \orange{\texttt{G}} & \texttt{GA\$$_5$\gray{G}}   \\ 
 \hline
 (5,1) & \texttt{\$$_5$} & \texttt{GGA\$$_5$}   \\ 
 \hline
 (1,2) & \texttt{A} & \texttt{TATG\$$_1$\gray{A}}   \\ 
 \hline
 (4,2) & \texttt{A} & \texttt{TCA\$$_4$\gray{A}}   \\ 
 \hline
 (1,4) & \texttt{A} & \texttt{TG\$$_1$\gray{ATA}}   \\ 
 \hline
 (2,1) & \texttt{\$$_2$} & \texttt{TGA\$$_2$}   \\ 
 \hline
\end{tabular}
\quad
\begin{tabular}{ |r|c|l| } 
 \hline
 \text{index} & \colex\ & \text{rotation} \\ 
 \hline
 (4,5) & \green{\texttt{A}} & \texttt{\$$_1$\gray{ATCA}}   \\ 
 \hline
 (5,4) & \green{\texttt{A}} & \texttt{\$$_2$\gray{GGA}}  \\
 \hline
 (2,4) & \green{\texttt{A}} & \texttt{\$$_3$\gray{TGA}}  \\ 
 \hline
 (3,4) & \green{\texttt{G}} & \texttt{\$$_4$\gray{ACG}}   \\ 
 \hline
 (1,6) & \green{\texttt{G}} & \texttt{\$$_5$\gray{ATATG}}   \\ 
 \hline
 (4,4) & \red{\texttt{C}} & \texttt{A\$$_1$\gray{ATC}}   \\ 
 \hline
 (5,3) & \red{\texttt{G}} & \texttt{A\$$_2$\gray{GG}}   \\ 
 \hline
 (2,3) & \red{\texttt{G}} & \texttt{A\$$_3$\gray{TG}}   \\ 
  \hline
 (3,1) & \texttt{\$} & \texttt{ACG\$$_4$}   \\ 
 \hline
 (1,1) & \texttt{\$} & \texttt{ATATG\$$_5$}   \\ 
 \hline
 (4,1) & \texttt{\$} & \texttt{ATCA\$$_1$}   \\ 
 \hline
 (1,3) & \texttt{T} & \texttt{ATG\$$_5$\gray{AT}}   \\ 
  \hline
 (4,3) & \texttt{T} & \texttt{CA\$$_1$\gray{AT}}   \\ 
  \hline
 (3,2) & \texttt{A} & \texttt{CG\$$_4$\gray{A}}   \\ 
 \hline
 (3,3) & \blue{\texttt{C}} & \texttt{G\$$_4$\gray{AC}}   \\ 
 \hline
 (1,5) & \blue{\texttt{T}} & \texttt{G\$$_5$\gray{ATAT}}   \\ 
 \hline
 (5,2) & \orange{\texttt{G}} & \texttt{GA\$$_2$\gray{G}}   \\ 
 \hline
 (2,2) & \orange{\texttt{T}} & \texttt{GA\$$_3$\gray{T}}   \\ 
 \hline
 (5,1) & \texttt{\$} & \texttt{GGA\$$_2$}   \\ 
 \hline
 (1,2) & \texttt{A} & \texttt{TATG\$$_5$\gray{A}}   \\ 
 \hline
 (4,2) & \texttt{A} & \texttt{TCA\$$_1$\gray{A}}   \\ 
 \hline
 (1,4) & \texttt{A} & \texttt{TG\$$_5$\gray{ATA}}   \\ 
 \hline
 (2,1) & \texttt{\$} & \texttt{TGA\$$_3$}   \\ 
 \hline
\end{tabular}
\quad 
\begin{tabular}{ |r|c|l| } 
 \hline
 \text{index} & \optimal & \text{rotation} \\ 
 \hline
 (2,4) & \green{\texttt{A}} & \texttt{\$$_1$\gray{TGA}}   \\ 
 \hline
 (5,4) & \green{\texttt{A}} & \texttt{\$$_2$\gray{GGA}}  \\
 \hline
 (4,5) & \green{\texttt{A}} & \texttt{\$$_3$\gray{ATCA}}  \\ 
 \hline
 (3,4) & \green{\texttt{G}} & \texttt{\$$_4$\gray{ACG}}   \\ 
 \hline
 (1,6) & \green{\texttt{G}} & \texttt{\$$_5$\gray{ATATG}}   \\ 
 \hline
 (2,3) & \red{\texttt{G}} & \texttt{A\$$_1$\gray{TG}}   \\ 
 \hline
 (5,3) & \red{\texttt{G}} & \texttt{A\$$_2$\gray{GG}}   \\ 
 \hline
 (4,4) & \red{\texttt{C}} & \texttt{A\$$_3$\gray{ATC}}   \\ 
  \hline
 (3,1) & \texttt{\$} & \texttt{ACG\$$_4$}   \\ 
 \hline
 (1,1) & \texttt{\$} & \texttt{ATATG\$$_5$}   \\ 
 \hline
 (4,1) & \texttt{\$} & \texttt{ATCA\$$_3$}   \\ 
 \hline
 (1,3) & \texttt{T} & \texttt{ATG\$$_5$\gray{AT}}   \\ 
  \hline
 (4,3) & \texttt{T} & \texttt{CA\$$_3$\gray{AT}}   \\ 
  \hline
 (3,2) & \texttt{A} & \texttt{CG\$$_4$\gray{A}}   \\ 
 \hline
 (3,3) & \blue{\texttt{C}} & \texttt{G\$$_4$\gray{AC}}   \\ 
 \hline
 (1,5) & \blue{\texttt{T}} & \texttt{G\$$_5$\gray{ATAT}}   \\ 
 \hline
 (2,2) & \orange{\texttt{T}} & \texttt{GA\$$_1$\gray{T}}   \\ 
 \hline
 (5,2) & \orange{\texttt{G}} & \texttt{GA\$$_2$\gray{G}}   \\ 
 \hline
 (5,1) & \texttt{\$} & \texttt{GGA\$$_2$}   \\ 
 \hline
 (1,2) & \texttt{A} & \texttt{TATG\$$_5$\gray{A}}   \\ 
 \hline
 (4,2) & \texttt{A} & \texttt{TCA\$$_3$\gray{A}}   \\ 
 \hline
 (1,4) & \texttt{A} & \texttt{TG\$$_5$\gray{ATA}}   \\ 
 \hline
 (2,1) & \texttt{\$} & \texttt{TGA\$$_1$}   \\ 
 \hline
\end{tabular}
\end{adjustbox}
\vspace{2mm}
\caption{\label{tab:3bwts-table2} From left to right we show the \mdollar, the \colex, and the optimal BWT of the string collection ${\cal M} = \{ {\tt \texttt{ATATG},\texttt{TGA},\texttt{ACG},\texttt{ATCA},\texttt{GGA} }\}$. Indices are given with reference to the numbering $T_1 = {\tt ATATG}, T_2 = {\tt TGA}, T_3 = {\tt ACG}, T_4 = {\tt ATCA}, T_5 = {\tt GGA}$. Note that we give the rotations according to Lemma~\ref{obs:1}.}
\end{table*}


\subsection{Interesting intervals}\label{sec:interesting}

Let us call a string $U$ a {\em shared suffix} w.r.t.\ multiset ${\cal M}$ if it is the suffix of at least two strings in ${\cal M}$. Let $b$ be the lexicographic rank of the smallest rotation beginning with $U\$$ and $e$ the lexicographic rank of the largest rotation beginning with $U\$$, among all rotations of strings $T\$$, where $T \in {\cal M}$. (One can think of $[b,e]$ as the suffix-array interval of $U\$$.) We call $[b,e]$ an {\em interesting interval} if there exist $i\neq j$ s.t.\ $U$ is a suffix of both $T_i$ and $T_j$, and the preceding characters in $T_i$ and $T_j$ are different, i.e., the two occurrences of $U$ as suffix of $T_i$ and $T_j$ constitute a left-maximal repeat. (Put in different terms, interesting intervals correspond to internal nodes in the suffix tree of the reverse string, within the subtree of $\$$.)  Clearly, $[1,k]$ is an interesting interval unless all strings end with the same character. Note that interesting intervals differ both from the {\em SAP-intervals} of~\cite{CoxBJR12} and from the {\em tuples} of~\cite{BentleyGT20} (called {\em maximal row ranges} in~\cite{manzini16}): the former are the intervals corresponding to {\em all} shared suffixes $U$, even if not left-maximal, while the latter include also suffixes $U$ that are not shared.

\begin{lemma}
Any two distinct interesting intervals are disjoint. 
\end{lemma}

\begin{proof}
Follows from the fact that no two distinct substrings ending in $\$$ can be one prefix of the other.
\end{proof}

We can now narrow down the differences between any two separator-based BWTs of the same multiset. The next proposition states that these can only occur in interesting intervals (part 1). This implies that the dollar-symbols appear in the same positions in all separator-based variants except for one very specific case (part 2). Moreover, we get an upper bound on the Hamming distance between two separator-based BWTs (part 3).

\begin{proposition}\label{prop:int_intervals}
Let $L_1$ and $L_2$ be two separator-based BWTs of the same multiset ${\cal M}$. 

\begin{enumerate} 
\item If $L_1[i] \neq L_2[i]$ then  $i \in [b,e]$ for some interesting interval $[b,e]$. 
\item Let ${\cal I}_1$ resp.\ ${\cal I}_2$ be the positions of the dollars in $L_1$ resp.\ $L_2$. If ${\cal I}_1\neq {\cal I}_2$ then there exist $i\neq j$ such that $T_i$ is a proper suffix of $T_j$. 
\item $\displaystyle{\Hdist(L_1,L_2) \leq \sum_{[b,e]  \text{ interesting interval}} (e - b + 1)}$. 
\end{enumerate}
\end{proposition}

\begin{proof} {\em 1.} Let $L_1[i] = {\tt x}$ and $L_2[i] = {\tt y}$. 
Since all separator-based \BWT\ variants use the lexicographical order of the rotations, this means that there exists a substring $U$ which is preceded by {\tt x} in one string $T_j$ and by {\tt y} in another $T_{j'}$, the first occurrence has rank $i$ in one \BWT\ and the other has rank $i$ in the other \BWT\ variant. This implies that the two occurrences are followed by two dollars, and either the two dollars are different, or they are the same dollar, and the subsequent substrings are different. Therefore, $U$ defines an interesting interval. Parts {\em 2.} and {\em 3.} follow from {\em 1.}
\end{proof}

Proposition~\ref{prop:int_intervals} implies that the variation of the different transforms can be explained based solely on what rule is used to break ties for shared suffixes. We will see next how the different BWT variants determine this tie-breaking rule.


\section{Permutations induced by separator-based BWT variants}\label{sec:permutations}

Let us now restrict ourselves to ${\cal M}$ being a set, i.e., no string occurs more than once. (Later we will show how to deal with multisets.) 
As we showed in the previous section, the only differences between the separator-based BWT variants are given by the order in which shared suffixes are listed. It is also clear that the same order applies in each interesting interval, as well as to the $k$-length prefix of the transform. Therefore, it suffices to study the permutation $\pi$ of the $k$ dollars in this prefix. 

Since the strings are all distinct, they each have a unique lexicographic rank within the set ${\cal M}$. Thus the input order can be seen as a permutation $\rho$ of the lexicographic ranks\footnote{For those used to thinking about suffix arrays, $\rho$ can be seen as the inverse suffix array of the input if the strings are thought of as meta-characters.}; 
if the strings are input in lexicographic order, then $\rho = id$. For our toy example ${\cal M} = [{\tt ATATG, TGA ,ACG ,ATCA, GGA}]$, we have $\rho = 25134$.  

Let us now define as {\em output permutation} $\pi$ the permutation of the last characters of the input strings, as found in the $k$-length prefix of the \BWT\ variant in question. We will denote the output permutations of the \dollarebwt, \mdollar, \concat, and \rlo\ by $\pi_{dolE}, \pi_{mdol}, \pi_{conc},$ and $\pi_{colex}$, respectively. (As the permutation of \optimal\ is algorithmically defined, we do not treat it here.) 
Again, we give these permutations w.r.t.\ the lexicographic ranks of the strings. 
In our running example, we have $\pi_{dolE} = 12345$, $\pi_{mdol} = 25134$, $\pi_{colex} = 34512$, and $\pi_{conc}= 45132.$

It is easy to see that the output permutation $\pi_{mdol}$ is equal to $\rho$, since the dollar-symbols are ordered according to $\rho$. For the \dollarebwt, the rank of {\tt \$}$T_i$ equals the lexicographic rank of $T_i$ among all input strings (Lemma~\ref{obs:1}), i.e., $\pi_{dolE}=id$. Further, $\pi_{colex} = \gamma$ by definition, where $\gamma$ denotes the colexicographic order of the input strings.  The situation is more complex in the case of \concat. Since the {\tt \#} is the smallest character, the last string of the input will be the first, while for the others, the lexicographic rank {\em of the following string} decides the order. In our running example, $\pi_{conc} = 45132$. We next formalize this.  

Let $\Phi_\rho$ be the {\em linking permutation}~\cite{KucherovTV13} of $\rho$, defined by $\Phi_\rho(i) = \rho(\rho^{-1}(i) + 1)$, for $i\neq \rho(k)$, and $\Phi_\rho(\rho(k)) = \rho(1)$, the permutation that maps each element to the element in the next position and the last element to the first. Let us also define, for $j\in \{1,\ldots,k\}$ and $i\neq j$, $f_j(i)$ by $f_j(i)=i$ if $i<j$ and $i-1$ otherwise, i.e.\ $f_j(i)$ gives the rank of element $i$ in the set $\{1,\ldots,k\} \setminus \{j\}$. 
The next lemma gives the precise relationship between $\rho$ and $\pi_{conc}$.

\begin{lemma}\label{lemma:linking}
Let $\rho$ be the permutation of the input order w.r.t.\ the lexicographic order, i.e.\ the $i$th input string has lexicographic rank $\rho(i)$. Then $\pi_{conc}=\pi_{conc}(\rho)$ is given by: 
\begin{align}
    \pi_{conc}(1) = \rho(k), \text{ and for } i\neq \rho(k): \pi_{conc}^{-1}(i) = f_{\rho(1)}(\Phi_\rho(i)) + 1. 
\end{align}
\end{lemma} 

\begin{proof}
Follows straightforwardly from the tie-breaking rule of \concat. 
\end{proof}

Essentially, Lemma~\ref{lemma:linking} says that $\pi_{conc}$ is the BWT of $\rho$. (We thank Massimiliano Rossi for this observation.) This can be seen as follows. Take the string collection ${\cal M}$ in order $\rho$ and construct a new string $T^{\rho}$ concatenating the lexicographic ranks of the strings in ${\cal M}$ with a final dollar, in our example $T^{\rho} = 25134\$$; thus, $T^{\rho}$ is a string over the alphabet $\{1,2,\ldots,k\}$ with an additional dollar at the end. It follows from Lemma~\ref{lemma:linking} that the output permutation $\pi_{conc}$ is the BWT of $T^{\rho}$, from which the $\$$-sign was removed: $\BWT(25134\$) = 45\$132$, therefore, $\pi_{conc} = 45132$.

\begin{example}
The mapping $\rho \mapsto \pi_{conc}$ for $k=3$ is as follows: $123 \mapsto 312$, $132 \mapsto 231$, $312 \mapsto 231$, $213 \mapsto 321$, $231 \mapsto 132$, and $321 \mapsto 123$. Note that no $\rho$ maps to $213$. 
\end{example}

As can be seen already for $k=3$, not all permutations $\pi$ are reached by this mapping. We will call a permutation $\pi$ {\em conc-feasible} if there exists an input order $\rho$ such that $\pi_{conc}(\rho) = \pi$. For $k=4$, there are $18$ conc-feasible permutations (out of $24$), for $k=5$, $82$ (out of $120$). In Table~\ref{tab:forbidden}, we give the percentage of conc-feasible permutations $\pi$, for $k$ up to $11$. The lexicographic order is always conc-feasible, namely with $\rho = k,k-1,\ldots, 2,1$; the colex order is not always conc-feasible, as the following example shows. 

\begin{table*}[h]
\centering
\begin{tabular}{ |l|c|c|c|c|c|c|c|c|c| } 
 \hline
 $k$  & 3 & 4 & 5 & 6 & 7 & 8 & 9 & 10 & 11 \\
 \hline \hline
         & {\tt 83.33\%} & {\tt 75.0\%} & {\tt 68.33\%} & {\tt 63.89\%} & {\tt 60.12\%} & {\tt 57.29\%} & {\tt 54.8\%} & {\tt 52.81\%} & {\tt 51.0\%} \\
 \hline
\end{tabular}
\vspace{2mm}
\caption{Percentage of conc-feasible permutations w.r.t.\ \concat.\label{tab:forbidden}}
\label{tab:exp-results}
\end{table*}

\begin{example}
Let ${\cal M} = \{\tt ACA,TGA,GAA\}$, thus $\rho = 132$, $\gamma = 213$, but as we have seen, no permutation of the strings in ${\cal M}$ will yield this order for \concat. In particular, the $\colex({\cal M}) =$ {\tt AAAACGG\$AT\$\$} has $7$ runs, while all conc-feasible concatBWTs have at least $8$: {\tt AAAGACG\$AT\$\$}, {\tt AAACGAG\$AT\$\$}, {\tt AAAAGCG\$AT\$\$}, {\tt AAAGCAG\$AT\$\$}, {\tt AAACAGG\$AT\$\$}. 
\end{example}

An important consequence is that, given an input permutation $\rho$, the output permutations induced by \mdollar\ and \concat\ are always different: $\pi_{mdol} \neq \pi_{conc}$ holds always, since $\pi_{conc}(1)=\rho(k)$. This means that, in whatever order the strings are given, on most string sets the resulting transforms \mdollar\ and \concat\ will differ.

\subsection{Permutations on multisets} 

Now let ${\cal M}$ be a multiset, so the same string can be contained more than once in ${\cal M}$. Let us again map ${\cal M}$ to a string $T^\rho$ over the alphabet of the lexicographic ranks $\{1,2,\ldots,k'\}$, where $k'\leq k$, and let us define the output order $\pi$ as before, as the order in which the lexicographic ranks appear in the $k$-length prefix of the BWT variant.\footnote{Formally, $\rho$ and $\pi$ are multi-permutations.} Then, $\pi_{mdol} = T^\rho$, $\pi_{dolE}$ is the sequence of non-decreasingly sorted ranks; and $\pi_{colex}$ is non-decreasing w.r.t.\ the colexicographical order. Finally, again $\pi_{conc}$ is the BWT of $T^\rho\$$ from which the dollar-sign has been removed.  

\begin{example}
    Let ${\cal M} = \{\tt ACA,TGA,ACA,GAA,TGA,TGA\}$, thus $T^{\rho} = 131233$. Then $\pi_{mdol} = 131233$, $\pi_{dolE} = 112333$, $\pi_{colex}= 211333$, and $\BWT(131233\$)=33\$1312$, resulting in $\pi_{conc} = 331312$.  
\end{example}

Let us denote by $\BWT^*(T)$ the string $\BWT(T\$)$ with the dollar removed. 
It has been shown experimentally that more than half of binary and ternary strings of length between $10$ and $20$ do not lie in the image of the function $\BWT^*$, with the percentage of those not in the image increasing with increasing length~\cite{GiulianiILPST21}.  As already seen for permutations (Table~\ref{tab:forbidden}), the function $\BWT^*$ is not surjective; the results of~\cite{GiulianiILPST21} seem to indicate that, in fact, the majority of multi-permutations cannot be produced by \concat. 

On the other hand, clearly all multi-permutations can be produced with \mdollar, as in that case, the output permutation is the same as the input permutation. Moreover, we have seen that all separator-based BWT-variants can be simulated by the \mdollar\ transform, since it suffices to apply \mdollar\ to the output permutation of the desired variant. We summarize: 

\begin{proposition}\label{prop:2}
    Let ${\cal M}$ be given, and $\rho$ the order of the lexicographic ranks in which the strings appear in ${\cal M}$. Then 
    \begin{enumerate}
        \item $\dollarebwt({\cal M}) = \mdollar(\lambda({\cal M}))$, with $\lambda$ the lexicographic order; 
        \item $\colex({\cal M}) = \mdollar(\gamma({\cal M}))$, with $\gamma$ the colexicographic order; 
        \item $\concat({\cal M}) = \mdollar(\beta({\cal M}))$, where $\beta = \BWT^*(T^\rho)$ and 
        $T^\rho$ is the meta-string consisting of the lexicographic ranks of the input strings. 
    \end{enumerate}
\end{proposition}


\section{Effects on the parameter $r$}\label{sec:effects_on_r}

What is the effect of the different permutations $\pi$ of the strings in ${\cal M}$, induced by these \BWT\ variants, on the number of runs of the \BWT? As the following example shows, the number of runs can differ significantly between different variants. 

\begin{example}
Let ${\cal M} = \{{\tt AAAA, AGCA, GCAA, GTCA, CAAA, CGCA, TCAA,}$ ${\tt TTCA}\}$. Then \\ $\mdollar({\cal M})$ $=$  {\tt AAAAAAAAACACACACACACAC\$\$GTGTGT\$\$AC\$\$GT\$\$} has 28 runs,  while \\ $\rlo({\cal M}) =$ {\tt AAAAAAAAAAAACCCCAACCAC\$\$GGTTGT\$\$AC\$\$GT\$\$} has 18 runs. 
\end{example}

The results of Section~\ref{sec:bwtvariants} give us a method to measure the degree to which the BWT variants can differ. 

\begin{lemma}\label{lemma:maxruns}
Let $[b,e]$ be an interesting interval, and $(n_1,\ldots,n_{\sigma})$ the Parikh vector of $L[b..e]$, i.e.\ $n_i$ is the number of occurrences of the $i$th character. Let {\tt a} be such that $n_{\tt a} = \max_{i} n_i$, and $N_{\tt a} = (e-b+1)-n_{\tt a}$, the sum of the other character multiplicities. Then the maximum number of runs in interval $[b,e]$ is $e - b + 1$ if $n_{\tt a}-1 \leq N_{\tt a}$, and $2N_{\tt a} +1$ otherwise. 
\end{lemma}

\begin{proof}
Place the $n_{\tt a}$ {\tt a}-characters in a row, creating $n_{\tt a} +1$ gaps, namely one between each adjacent {\tt a}, and one each at the beginning and at the end. Now place all {\tt b}-characters, each in a different gap; since $n_{\tt a}$ is maximum, there are enough gaps. Then place all {\tt c}'s, first filling gaps that are still empty, if any, then into gaps without {\tt c}, etc. We never have to place two identical characters in the same gap. If the total number of non-{\tt a}-characters is at least than $n_{\tt a} - 1$, then we can fill every gap, thus separating all {\tt a}'s, and creating a run for every character of $I$. If we have fewer than $n_{\tt a}-1$ characters, then we are still creating two runs with each non-{\tt a}-character, but we cannot separate all {\tt a}'s. 
\end{proof}

We will use this lemma to measure the variability of a dataset: 

\begin{definition}\label{def:var}
Let ${\cal M}$ be a multiset. For an interesting interval $[b,e]$, let $var([b,e])$ be the upper bound on the number of runs in $[b,e]$ from Lemma~\ref{lemma:maxruns}. Then the {\em variability} of ${\cal M}$ is 
\[var({\cal M}) = \frac{\sum_{[b,e] \text{ interesting interval}} var([b,e])}{\sum_{[b,e] \text{ interesting interval}} (e-b+1)}.\]
\end{definition}

The \colex\ has been shown experimentally to yield a low number of runs of the \BWT\ \cite{Li14a,CoxBJR12}. Even though it does not always minimize $r$ (one can easily create small examples where other permutations yield a lower number of runs), we can bound its distance from the optimum. 

\begin{proposition}\label{prop:runs}
Let $L$ be the \rlo\ of multiset ${\cal M}$, and let $r_{\text{OPT}}$ denote the minimum number of runs of any separator-based BWT of ${\cal M}$. Then $\runs(L) \leq r_{\text{OPT}}+ 2\cdot c_{\cal M}$, where $c_{\cal M}$ is the  number of interesting intervals.
\end{proposition}

\begin{proof} 
Let $I=[b_I,e_I]$ be an interesting interval containing $d$ distinct characters, and let $U$ be the shared suffix defining $I$. Since the strings are listed according to the colex order, all strings in which $U$ is preceded by the same character will appear in one block, and therefore, $L$ has exactly $d$ runs in the interval $I$. Let $L_{b_I-1} = {\tt x}$ and $L_{e_I+1} = {\tt y}$. If ${\tt x}$ occurs in $I$ and it is not the first run of $I$ (i.e., $L_{b_I}\neq {\tt x}$), then listing first the strings where $U$ is preceded by ${\tt x}$ would reduce the number of runs by $1$; similarly, listing those where ${\tt y}$ precedes $U$ as last of the group would reduce the number of runs by $1$. By Prop.~\ref{prop:int_intervals}, this is the only possibility for varying the number of runs. 
\end{proof}

The algorithm of Bentley et al.~\cite{BentleyGT20} for the optimal order for \mdollar\ is based on the idea of 
starting from the colex-order and then adjusting, where possible, the order of the runs within interesting intervals in order to minimize character changes at the borders, i.e.\ such that the first and the last run of each interesting interval is identical to the run preceding and following that interesting interval. This is equivalent to sorting groups of sequences sharing the same left-maximal suffix. This sorting can be done on each interesting interval independently without affecting the other interesting intervals. In Table~\ref{tab:3bwts-table2}, we show the result on our toy example, where it reduces the number of runs by $2$ w.r.t.\ colex order. 
In the next section, we compare the number of runs of the non-separator based BWT variants to the optimum. 


\section{Experimental results}\label{sec:experimental}

We computed the five BWT variants \eBWT, \dollarebwt, \mdollar, \concat, and \colex, on eight different genomic datasets. We used the tool {\tt optimalBWT} to compute the minimum number of runs (i.e., that of \optimal) and used this as a baseline for comparison with the $r$ parameter of the other \BWT-variants. For \mdollar\ and \concat, we used the default input order in which the dataset was downloaded. 
The eight datasets have different characteristics: Four of the datasets contain short reads: SARS-CoV-2 short~\cite{STARR20201295}, Simons Diversity reads~\cite{Mallick2016}, 16S rRNA short~\cite{Winand2019}, Influenza A reads~\cite{Influenza2015}, and four contain long sequences: SARS-CoV-2 long~\cite{Greaney22}, 16S rRNA long~\cite{16S2018}, Candida auris reads~\cite{candida2019}, one of which, SARS-CoV-2 genomes, whole viral genomes~\cite{BoucherCLMM21}. The main features of the datasets, including the number of sequences, sequence length, and the mean runlength of the optimal BWT are reported in Table~\ref{tab:datasets}. 
We include the details of the experimental setup in the Appendix.

On each of the datasets, we computed the pairwise Hamming distance between separator-based BWTs. To compare them to the \eBWT, we computed the pairwise edit distance on a small subset of the sequences (for obvious computational reasons), computing also the Hamming distance on the small set, for comparison. We generated the following statistics on each of the data sets: the number of interesting intervals, the fraction of positions within interesting intervals (total length of interesting intervals divided by total length of the dataset), and the dataset's variability (Def.~\ref{def:var}). 
In Table~\ref{tab:SARS_short_compact} and~\ref{tab:SARS_genomes_compact}, we include a compact version of these results for the two datasets with the highest and the lowest variation between the BWT variants, the SARS-CoV-2 short sequences and the SARS-CoV-2 genomes, respectively. The full experimental results for all eight datasets are contained in the Appendix. 

In Table~\ref{tab:results_summary} we give a brief summary of the results, reporting, for each dataset, the fraction of positions in interesting intervals, the dataset's variability, the average pairwise Hamming distance between separator-based BWT variants, and the maximum and minimum value, among the different BWT variants, of the average runlength ($n/r$) of the BWT. 

\newcommand{\OPT}{\textit{opt}}
The experiments showed a high variation in the number of runs in particular on datasets of short sequences. The highest difference was between \colex\ and \concat, by a multiplicative factor of over $4.2$, on the SARS-CoV-2 short dataset. In Figure~\ref{fig:bars} we plot the average runlength $n/r$ for the four short sequence datasets, and the percentage increase of the number of runs w.r.t.\ $r_{OPT}$. 
The variation is less pronounced on the one dataset which is less repetitive, namely Simons Diversity reads.  Recall that the \mdollar\ and \concat\ vary depending on the input permutation. On most long sequence datasets, on the other hand, the differences were quite small (see Appendix). Recall also that the \mdollar\ and \concat\ vary depending on the input permutation. To better understand how far the \colex\ is from the optimum w.r.t.\ the  number of runs, we plot in Figure~\ref{fig:bluebars} the number of runs of \colex\ w.r.t.\ to $r_{\OPT}$, on all eight datasets. The strongest increase is on short sequences, where the variation among all BWT variants is high, as well; on the long sequence datasets, with the exception of SARS-CoV-2 long sequences, the \colex\ is very close to the optimum; however, note that on those datasets, all BWTs are close to the optimum. 

The average number of runs and the average pairwise Hamming distance strongly depend on the length of the sequences in the input collection. If the collection has a lot of short sequences which are very similar, then the differences between the BWTs both w.r.t.\ the number of runs, and as measured by the Hamming distance, can be large. This is because there are a lot of maximal shared suffixes and so, many positions are in interesting intervals. To better understand this relationship, we plotted, in Figure~\ref{fig:scatter}, the average Hamming distance against the two parameters variability and fraction of positions in interesting intervals. We see that the two datasets with highest average Hamming distance, SARS-CoV-2 short dataset and the Simons Diversity reads, have at least one of the two values very close to $1$, while for those datasets where both values are very low, the BWT variants do not differ very much.

Note that the input order used by the \mdollar\ and the \concat\ is the order in which the input sequences appear when the dataset is downloaded. Our study shows that only a few input permutations can minimize the number of runs of the resulting BWT, namely those orders that group the characters inside the interesting intervals in at most $\sigma$ runs, such as the order of Bentley et al.\ and the colexicographic order. However, since there are $k!$ possible input permutations, selecting an arbitrary input order will likely result in a BWT whose number of runs is much larger than the optimal one, especially on datasets with high variability.


\begin{table*}[h]
\centering
\begin{adjustbox}{max width=140mm}
\begin{tabular}{ |l|r|r|r|r|r|r| } 
 \hline
 \text{dataset} & \text{no. seq} & \textrm{total length} & \textrm{avg} & \textrm{min} & \textrm{max} & \textrm{$n/r$ (opt)} \\ 
 \hline\hline
 SARS-CoV-2 short & 500,000 & 25,000,000 & 50 & 50 & 50 & 35.125   \\ 
 \hline
 Simons Diversity reads & 500,000 & 50,000,000 & 100 & 100 & 100 & 8.133  \\ 
 \hline
 16S rRNA short & 500,000 & 75,929,833 & 152 & 69 & 301 & 44.873   \\ 
 \hline
 Influenza A reads & 500,000 & 115,692,842 & 231 & 60 & 251 & 50.275   \\ 
 \hline
 SARS-CoV-2 long & 50,000 & 53,726,351 & 1,075 & 265 &  3,355 & 74.498  \\ 
 \hline
 16S rRNA long & 16,741 & 25,142,323 & 1,502 & 1,430 & 1,549 & 47.140   \\ 
 \hline
 Candida auris reads & 50,000 & 124,150,880 & 2,483 & 214 &  8,791 & 1.732   \\ 
 \hline
 SARS-CoV-2 genomes & 2,000 & 59,610,692 & 29,805 & 22,871 & 29,920 & 523.240   \\ 
 \hline

\end{tabular}
\end{adjustbox}
\vspace{2mm}
\caption{\label{tab:datasets} Summary of the most important parameters of the eight datasets. From left to right we report the dataset name, the number of sequences, the total length, the average, minimum and maximum sequence length, and the average runlength $n/r$ of the optimum BWT according to Bentley et al.~\cite{BentleyGT20}.}
\end{table*}

\begin{table*}[h]
\centering
\begin{adjustbox}{max width=140mm}
\begin{tabular}{ |l|r|r|r|r|r| } 
 \hline
 dataset & \textrm{ratio pos.s} & \textrm{varia-} & avg.\ Hamming d.\ & max $n/r$ & min $n/r$ \\ 
 & in intr.int.s & bility & \multicolumn{1}{l|}{ betw.\ \$-sep.\ BWTs} & (avg.\ runlength) & (avg.\ runlength)\\
 \hline\hline
 SARS-CoV-2 short & 0.792 & 0.210 & $0.11754$ & 31.524 & 7.494   \\ 
 \hline
Simons Diversity reads & 0.107 & 0.976 & $0.07195$ & 7.873 & 5.299   \\ 
 \hline
 16S rRNA short & 0.741 & 0.058 & $0.02982$ & 44.253 & 18.836  \\ 
 \hline
 Influenza A reads & 0.103 & 0.363 & $0.02609$ & 49.172 & 23.100 \\ 
 \hline
 SARS-CoV-2 long & 0.175 & 0.037 & $0.00464$  & 73.204 & 57.568 \\ 
 \hline
 16S rRNA long & 0.047 & 0.104 & $0.00289$ & 46.879 & 45.015 \\ 
 \hline
 Candida auris reads & 0.007 & 0.497 & $0.00246$ & 1.732 & 1.726 \\ 
 \hline
 SARS-CoV-2 genomes & 0.001 & 0.148 & $0.00012$ & 521.610 & 499.549 \\ 
 \hline
\end{tabular}
\end{adjustbox}
\vspace{2mm}
\caption{\label{tab:results_summary} Summary of the results on the eight datasets. From left to right we report dataset names followed by the ratio of positions in interesting intervals, the variability of the dataset (see Def.~\ref{def:var}), the average normalized Hamming distance between any two separator-based BWT variants. In the last two columns we report the maximum and minimum average runlength ($n/r$) taken over all five BWT variants.}
\end{table*}

\begin{table*}[h]
\begin{center}
\raggedright
\textbf{SARS-CoV-2 short (500,000 short sequences)}\\
\captionsetup{width=\linewidth}
\begin{adjustbox}{max width=140mm}
\setlength{\tabcolsep}{5pt}
\renewcommand{\arraystretch}{1.6}
\begin{tabular}{|r||rrrr|}
  \hline
  \multirow{2}{*}{{\diagbox[width=4.7cm, height=1.18cm]{\rlap{\enspace\raisebox{0ex}{ \it \hspace{-3.0mm} norm.\ Hamming d.\  }}}{\raisebox{-0ex}{\it \hspace{-3.0mm} Hamming d.\ }}}} & \multicolumn{4}{l|}{\it Hamming distance on the big dataset} \\
  \cline{2-5}
  & \multicolumn{1}{R{1.6cm}|}{\dollarebwt} & \multicolumn{1}{R{1.6cm}|}{\mdollar} & \multicolumn{1}{R{1.6cm}|}{\concat} & \multicolumn{1}{R{1.6cm}|}{\rlo} \\
  \hline \hline
 \multicolumn{1}{|l||}{\dollarebwt}& \multicolumn{1}{l|}{\hspace{5.0mm} 0} & \multicolumn{1}{r|}{3,014,183} & \multicolumn{1}{r|}{2,926,602} & \multicolumn{1}{r|}{2,912,860} \\
 \hline
 \multicolumn{1}{|l||}{\mdollar} & \multicolumn{1}{r|}{0.11820} & \multicolumn{1}{l|}{\hspace{5.0mm} 0}  & \multicolumn{1}{r|}{3,013,908} & \multicolumn{1}{r|}{3,102,887} \\ 
 \hline
 \multicolumn{1}{|l||}{\concat} & \multicolumn{1}{r|}{0.11477} & \multicolumn{1}{r|}{0.11819} & \multicolumn{1}{l|}{\hspace{5.0mm} 0} & \multicolumn{1}{r|}{3,013,634} \\ 
 \hline
 \multicolumn{1}{|l||}{\rlo} & \multicolumn{1}{r|}{0.11423} & \multicolumn{1}{r|}{0.12168} & \multicolumn{1}{r|}{0.11818} & \multicolumn{1}{l|}{\hspace{5.0mm} 0} \\ 
 \hline
\end{tabular}
\quad
\begin{tabular}{|l||r|}
\hline
\multicolumn{2}{|l|}{\it dataset properties} \\
\hline \hline
no.\ sequences & 500,000 \\
\hline
average length & 50 \\
\hline
total length & 25,000,000 \\
\hline
no.\ of interesting intervals & 116,598 \\
\hline
total length intr.int.s & 20,187,840 \\
\hline
fraction pos.s in intr.int.s  & 0.792 \\
\hline
variability & 0.210 \\
\hline
\end{tabular}
\end{adjustbox}
\newline
\vspace*{0.2 cm}
\newline
\begin{adjustbox}{max width=140mm}
\setlength{\tabcolsep}{5pt}
\renewcommand{\arraystretch}{1.6}
\begin{tabular}{|r||r|r|r|r|r|}
  \hline
  \multirow{2}{*}{{\diagbox[width=4.0cm, height=1.18cm]{\rlap{\enspace\raisebox{0ex}{ \it \hspace{-3.0mm} norm.\ edit d.\  }}}{\raisebox{+0ex}{\it edit d.\ }} }} & \multicolumn{5}{l|}{\it edit distance on a subset of 5,000 sequences} \\
  \cline{2-6}
  & \multicolumn{1}{R{1.6cm}|}{\eBWT} &\multicolumn{1}{R{1.6cm}|}{\dollarebwt} & \multicolumn{1}{R{1.6cm}|}{\mdollar} & \multicolumn{1}{R{1.6cm}|}{\concat} & \multicolumn{1}{R{1.6cm}|}{\rlo} \\
  \hline \hline
 \multicolumn{1}{|l||}{\eBWT} & \multicolumn{1}{l|}{\hspace{5.0mm} 0} & \multicolumn{1}{r|}{28,702} & \multicolumn{1}{r|}{43,903} & \multicolumn{1}{r|}{43,828} & \multicolumn{1}{r|}{46,936} \\ 
 \hline
 \multicolumn{1}{|l||}{\dollarebwt}& \multicolumn{1}{r|}{0.11256} & \multicolumn{1}{l|}{\hspace{5.0mm} 0} & \multicolumn{1}{r|}{17,000} & \multicolumn{1}{r|}{16,921} & \multicolumn{1}{r|}{20,104}  \\ 
 \hline
 \multicolumn{1}{|l||}{\mdollar} & \multicolumn{1}{r|}{0.17217} & \multicolumn{1}{r|}{0.06667} & \multicolumn{1}{l|}{\hspace{5.0mm} 0}  & \multicolumn{1}{r|}{16,130} & \multicolumn{1}{r|}{20,812} \\ \cline{1-6}
 \multicolumn{1}{|l||}{\concat} & \multicolumn{1}{r|}{0.17187} & \multicolumn{1}{r|}{0.06636} & \multicolumn{1}{r|}{0.06325} & \multicolumn{1}{l|}{\hspace{5.0mm} 0} & \multicolumn{1}{r|}{20,830}  \\ 
 \hline
 \multicolumn{1}{|l||}{\rlo} & \multicolumn{1}{r|}{0.18406} & \multicolumn{1}{r|}{0.07884} & \multicolumn{1}{r|}{0.08162} & \multicolumn{1}{r|}{0.08169} & \multicolumn{1}{l|}{\hspace{5.0mm} 0} \\ 
 \hline
\end{tabular}
\quad
\begin{tabular}{|l||r|r|}
\hline
\multicolumn{3}{|l|}{\it no.\ runs big dataset} \\
\hline \hline
& $r$ & $n/r$ \\
\hline
\eBWT & 1,902,148 & $13.143$ \\
\hline
\dollarebwt & 1,868,581 & $13.647$ \\
\hline
\mdollar & 3,113,818 & $8.189$ \\
\hline
\concat & 3,402,513 & $7.494$ \\
\hline
\rlo & 808,906 & $31.524$ \\
\hline
\optimal & 725,979 & $35.125$ \\
\hline
\end{tabular}
\end{adjustbox}
\vspace{2mm}
\caption{\label{tab:SARS_short_compact} Results for the SARS-CoV-2 short dataset. Top left: absolute and normalized pairwise Hamming distance between separator-based BWT variants. Top right: summary of the dataset properties. Bottom left: absolute and normalized pairwise edit distance between all BWT variants on a subset of the input collection. Bottom right: number of runs and average runlength ($n/r$) taken over all BWT variants.
}
\end{center}
\end{table*}

\begin{table*}[t!]
\begin{center}
\raggedright
\textbf{SARS-CoV-2 genomes (2,000 long sequences)}

\captionsetup{width=\linewidth}
\begin{adjustbox}{max width=140mm}
\setlength{\tabcolsep}{5pt}
\renewcommand{\arraystretch}{1.6}
\begin{tabular}{|r||rrrr|}
  \hline
  \multirow{2}{*}{{\diagbox[width=4.7cm, height=1.18cm]{\rlap{\enspace\raisebox{0ex}{ \it \hspace{-3.0mm} norm.\ Hamming d.\  }}}{\raisebox{-0ex}{\it \hspace{-3.0mm} Hamming d.\ }}}} & \multicolumn{4}{l|}{\it Hamming distance on the big dataset} \\
  \cline{2-5}
  & \multicolumn{1}{R{1.6cm}|}{\dollarebwt} & \multicolumn{1}{R{1.6cm}|}{\mdollar} & \multicolumn{1}{R{1.6cm}|}{\concat} & \multicolumn{1}{R{1.6cm}|}{\rlo} \\
  \hline \hline
 \multicolumn{1}{|l||}{\dollarebwt}& \multicolumn{1}{l|}{\hspace{5.0mm} 0} & \multicolumn{1}{r|}{7,958} & \multicolumn{1}{r|}{7,900} & \multicolumn{1}{r|}{7,263} \\
 \hline
 \multicolumn{1}{|l||}{\mdollar} & \multicolumn{1}{r|}{0.00013} & \multicolumn{1}{l|}{\hspace{5.0mm} 0}  & \multicolumn{1}{r|}{7,958} & \multicolumn{1}{r|}{7,957} \\ 
 \hline
 \multicolumn{1}{|l||}{\concat} & \multicolumn{1}{r|}{0.00013} & \multicolumn{1}{r|}{0.00013} & \multicolumn{1}{l|}{\hspace{5.0mm} 0} & \multicolumn{1}{r|}{7,990} \\ 
 \hline
 \multicolumn{1}{|l||}{\rlo} & \multicolumn{1}{r|}{0.00012} & \multicolumn{1}{r|}{0.00013} & \multicolumn{1}{r|}{0.00013} & \multicolumn{1}{l|}{\hspace{5.0mm} 0} \\ 
 \hline
\end{tabular}
\quad
\begin{tabular}{|l||r|}
\hline
\multicolumn{2}{|l|}{\it dataset properties} \\
\hline \hline
no.\ sequences & 2,000 \\
\hline
total length & 59,612,692 \\
\hline
average length & 29,085 \\
\hline
no.\ interesting intervals & 1863 \\
\hline
total length intr.int.s & 80,486 \\
\hline
fraction pos.s in intr.int.s  & 0.001 \\
\hline
variability & 0.148 \\
\hline
\end{tabular}
\end{adjustbox}
\newline
\vspace*{0.2 cm}
\newline
\begin{adjustbox}{max width=140mm}
\setlength{\tabcolsep}{5pt}
\renewcommand{\arraystretch}{1.6}
\begin{tabular}{|r||r|r|r|r|r|}
  \hline
  \multirow{2}{*}{{\diagbox[width=4.0cm, height=1.18cm]{\rlap{\enspace\raisebox{0ex}{ \it \hspace{-3.0mm} norm.\ edit d.\  }}}{\raisebox{+0ex}{\it edit d.\ }} }} & \multicolumn{5}{l|}{\it edit distance on a subset of 50 sequences} \\
  \cline{2-6}
  & \multicolumn{1}{R{1.6cm}|}{\eBWT} &\multicolumn{1}{R{1.6cm}|}{\dollarebwt} & \multicolumn{1}{R{1.6cm}|}{\mdollar} & \multicolumn{1}{R{1.6cm}|}{\concat} & \multicolumn{1}{R{1.6cm}|}{\rlo} \\
  \hline \hline
 \multicolumn{1}{|l||}{\eBWT} & \multicolumn{1}{l|}{\hspace{5.0mm} 0} & \multicolumn{1}{r|}{786} & \multicolumn{1}{r|}{795} & \multicolumn{1}{r|}{801} & \multicolumn{1}{r|}{791} \\ 
 \hline
 \multicolumn{1}{|l||}{\dollarebwt}& \multicolumn{1}{r|}{0.00053} & \multicolumn{1}{l|}{\hspace{5.0mm} 0} & \multicolumn{1}{r|}{98} & \multicolumn{1}{r|}{107} & \multicolumn{1}{r|}{86}  \\ 
 \hline
 \multicolumn{1}{|l||}{\mdollar} & \multicolumn{1}{r|}{0.00053} & \multicolumn{1}{r|}{0.00007} & \multicolumn{1}{l|}{\hspace{5.0mm} 0}  & \multicolumn{1}{r|}{105} & \multicolumn{1}{r|}{112} \\ \cline{1-6}
 \multicolumn{1}{|l||}{\concat} & \multicolumn{1}{r|}{0.00054} & \multicolumn{1}{r|}{0.00007} & \multicolumn{1}{r|}{0.00007} & \multicolumn{1}{l|}{\hspace{5.0mm} 0} & \multicolumn{1}{r|}{114}  \\ 
 \hline
 \multicolumn{1}{|l||}{\rlo} & \multicolumn{1}{r|}{0.00053} & \multicolumn{1}{r|}{0.00006} & \multicolumn{1}{r|}{0.00008} & \multicolumn{1}{r|}{0.00008} & \multicolumn{1}{l|}{\hspace{5.0mm} 0} \\ 
 \hline
\end{tabular}
\quad
\begin{tabular}{|l||r|r|}
\hline
\multicolumn{3}{|l|}{\it no.\ runs big dataset} \\
\hline \hline
& $r$ & $n/r$ \\
\hline
\eBWT & 117,628 & $506.773$ \\
\hline
\dollarebwt & 117,410 & $507.731$ \\
\hline
\mdollar & 118,870 & $501.495$ \\
\hline
\concat & 119,334 & $499.549$ \\
\hline
\rlo & 114,287 & $521.605$ \\
\hline
\optimal & 113,930 & $523.240$ \\
\hline
\end{tabular}
\end{adjustbox}
\vspace{2mm}
\caption{\label{tab:SARS_genomes_compact} Results for the SARS-CoV-2 genomes dataset. Top left: absolute and normalized pairwise Hamming distance between separator-based BWT variants. Top right: summary of the dataset properties. 
Bottom left: absolute and normalized pairwise edit distance between all BWT variants on a subset of the input collection. Bottom right: number of runs and average runlength ($n/r$) taken over all BWT variants.
}
\end{center}
\end{table*}

\begin{figure*}[htbp]
    \centering
    \subfloat{
 	\centering
 		\includegraphics[width=0.49\textwidth]{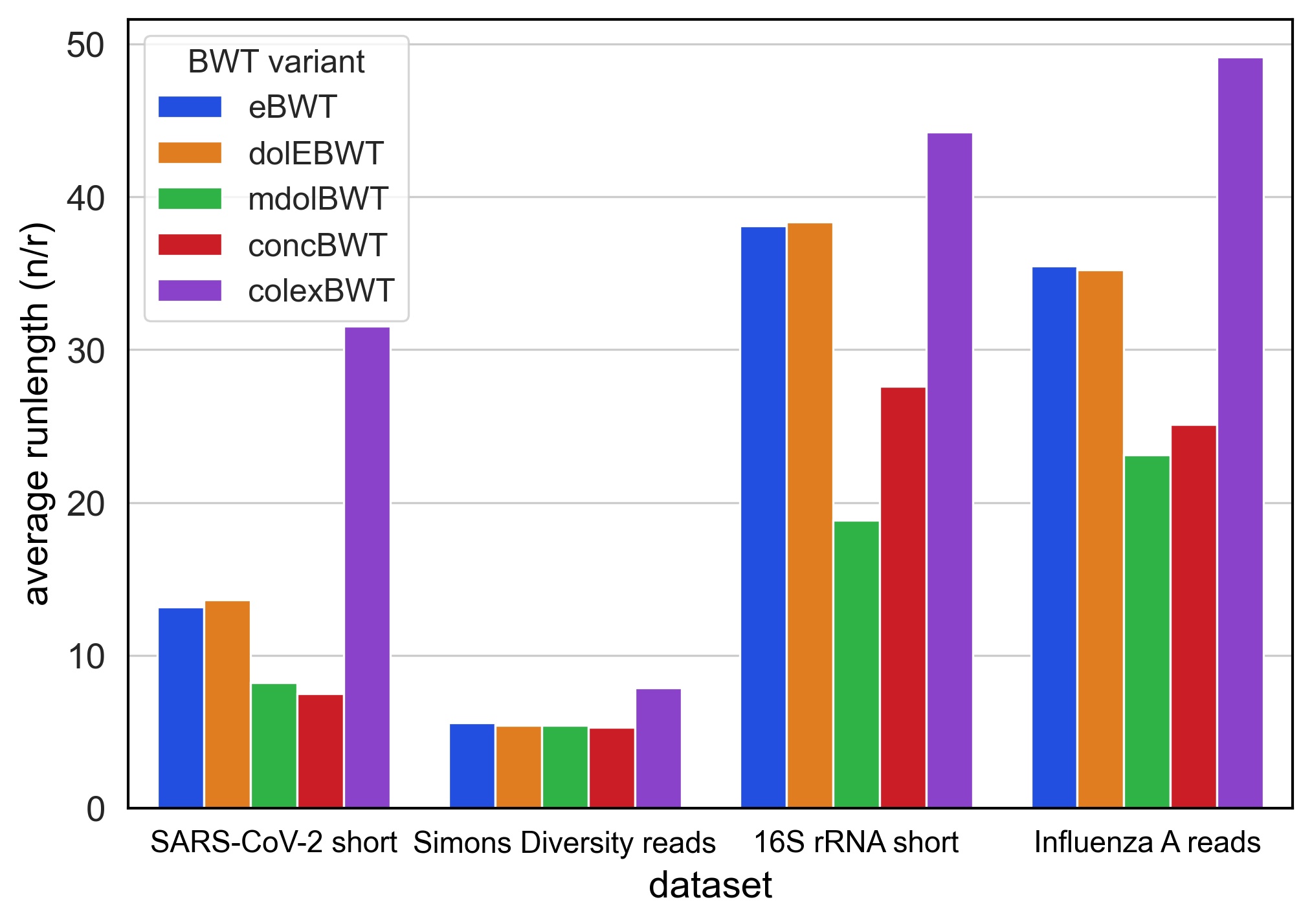}
    }~
   \subfloat{
     \centering
     		\includegraphics[width=0.49\textwidth]{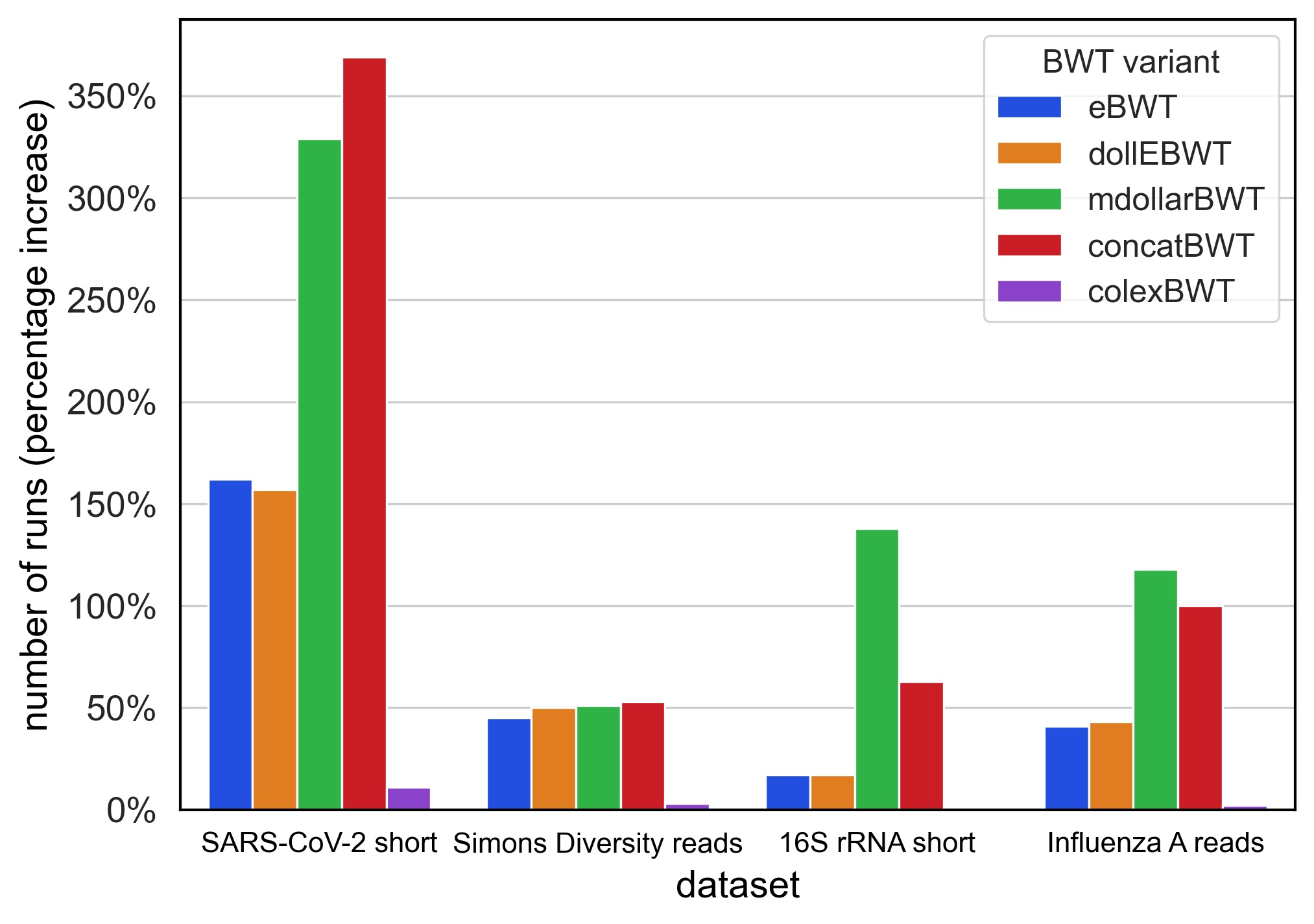}
    }
    \caption{Results regarding $r$ on short sequence datasets, of all BWT variants. Left: average runlength ($n/r$). Right: number of runs (percentage increase with respect to optimal BWT).}
    \label{fig:bars}
\end{figure*}

\begin{figure*}[htbp]
    \centering
    \includegraphics[width=0.65\textwidth]{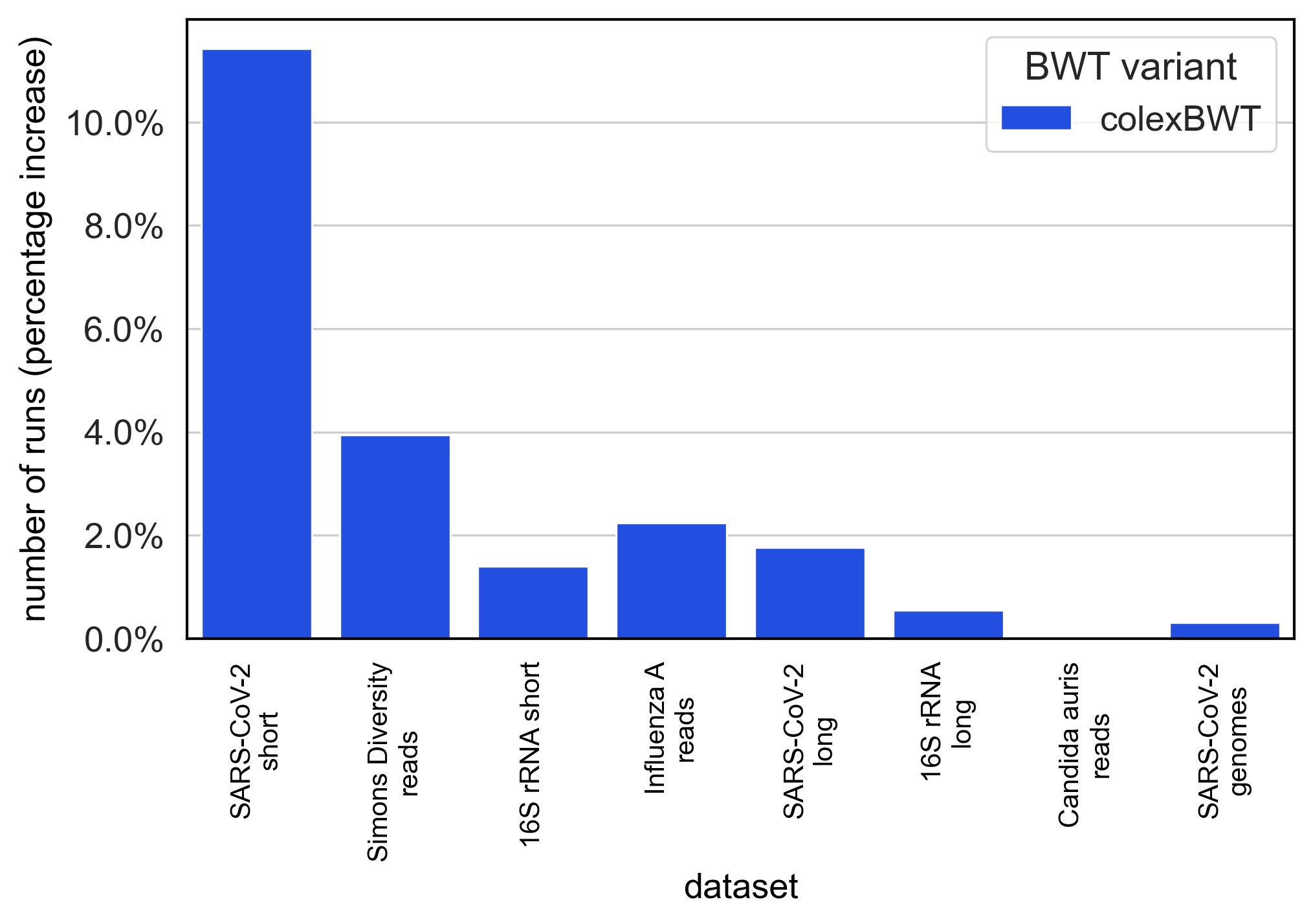}
    \caption{Number of runs of the \colex\ with respect to optimal BWT (percentage increase) on all eight datasets.}
    \label{fig:bluebars}
\end{figure*}

\begin{figure*}[htbp]
    \centering
    \includegraphics[width=0.65\textwidth]{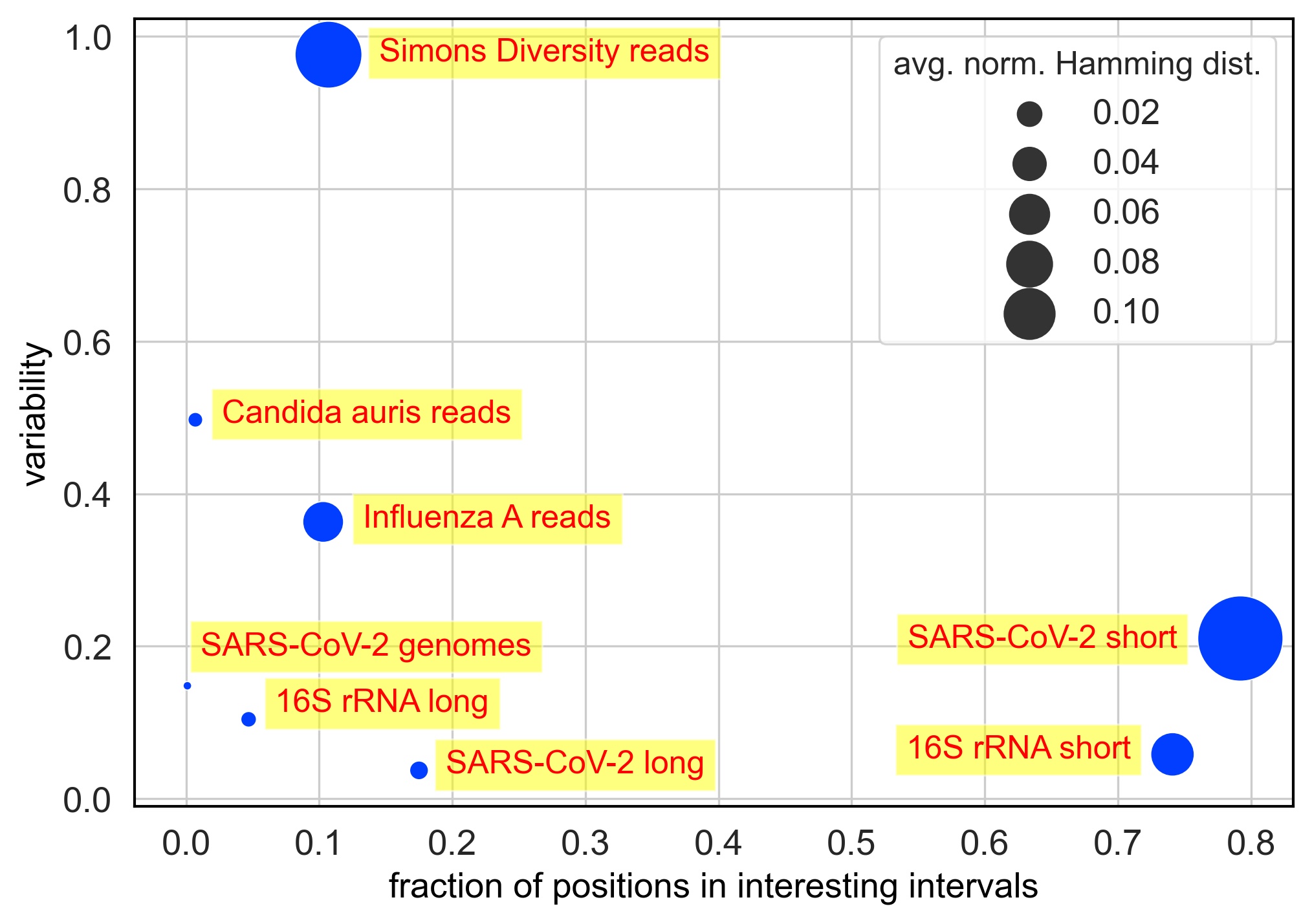}
    \caption{Average normalized Hamming distance variations with respect to variability and fraction of positions in interesting intervals on all datasets.}
    \label{fig:scatter}
\end{figure*}


\section{Conclusion}\label{sec:conclusion}

We presented the first study of the different variants of the Burrows-Wheeler Transform for string collections. 
We found that the transforms computed by different tools differ not insignificantly, as measured by the pairwise Hamming distance: up to 12\% between different BWT variants on the same dataset in our experiments. We showed that most current tools implement BWT variants that are input order dependent, so that the same tool can produce different outputs if the input set is permuted. These differences extend also to the number of runs $r$, a parameter that is central in the analysis of BWT-based data structures, and which is increasingly being used as a measure of the repetitiveness of the dataset itself. 

With string collections replacing individual sequences as the prime object of research and analysis, and thus becoming the standard input for text indexing algorithms, we believe that it is all the more important for users and researchers to be aware that not all methods are equivalent, and to understand the precise nature of the BWT variant produced by a particular tool.  

We suggest further to standardize the definition of the parameter $r$ for string collections, using  either the colexicographic order (implemented by the tool \rope~\cite{Li14a}) or the optimal order of Bentley et al.~\cite{BentleyGT20} (implemented by the tool {\tt optimalBWT}~\cite{CenzatoGLR23}). In this paper, we found that 
the number of runs can vary by up to a factor of $4.2$ on real-life biological datasets, while in~\cite{CenzatoGLR23}, a factor of $31$ was shown on other biological data. 
Not only does this heavily impact the space requirements of BWT-based data structures, but it also means that using the average runlength $n/r$ as a repetitiveness measure of a dataset is ambiguous, unless the research community agrees on the BWT variant being used for the definition of this parameter.

\bibliography{reference}

\begin{thebibliography}{10}

\bibitem{AkagiFI22}
Tooru Akagi, Mitsuru Funakoshi, and Shunsuke Inenaga.
\newblock Sensitivity of string compressors and repetitiveness measures.
\newblock {\em CoRR}, abs/2107.08615, 2021.

\bibitem{BannaiGI20}
Hideo Bannai, Travis Gagie, and Tomohiro I.
\newblock Refining the \emph{r}-index.
\newblock {\em Theor. Comput. Sci.}, 812:96--108, 2020.

\bibitem{BauerCR13}
Markus~J. Bauer, Anthony~J. Cox, and Giovanna Rosone.
\newblock Lightweight algorithms for constructing and inverting the {BWT} of
  string collections.
\newblock {\em Theor. Comput. Sci.}, 483:134--148, 2013.

\bibitem{BentleyGT20}
Jason~W. Bentley, Daniel Gibney, and Sharma~V. Thankachan.
\newblock On the complexity of {BWT}-runs minimization via alphabet reordering.
\newblock In {\em Proc.\ of 28th Annual European Symposium on Algorithms ({ESA}
  2020)}, volume 173 of {\em LIPIcs}, pages 15:1--15:13, 2020.

\bibitem{BonizzoniVPPR19}
Paola Bonizzoni, Gianluca~Della Vedova, Yuri Pirola, Marco Previtali, and
  Raffaella Rizzi.
\newblock Multithread multistring {B}urrows-{W}heeler {T}ransform and {L}ongest
  {C}ommon {P}refix array.
\newblock {\em J. Comput. Biol.}, 26(9):948--961, 2019.

\bibitem{BoucherCLMM21}
Christina Boucher, Davide Cenzato, {\relax Zs}uzsanna Lipt{\'{a}}k,
  Massimiliano Rossi, and Marinella Sciortino.
\newblock Computing the original e{BWT} faster, simpler, and with less memory.
\newblock In {\em Proc.\ of 28th International Symposium on String Processing
  and Information Retrieval ({SPIRE} 2021)}, volume 12944 of {\em LNCS}, pages
  129--142, 2021.

\bibitem{BoucherCGHMNR21}
Christina Boucher, Ondrej Cvacho, Travis Gagie, Jan Holub, Giovanni Manzini,
  Gonzalo Navarro, and Massimiliano Rossi.
\newblock {PFP} compressed suffix trees.
\newblock In {\em Proc.\ of 23rd Symposium on Algorithm Engineering and
  Experiments ({ALENEX} 2021)}, pages 60--72. {SIAM}, 2021.

\bibitem{BoucherGKLMM19}
Christina Boucher, Travis Gagie, Alan Kuhnle, Ben Langmead, Giovanni Manzini,
  and Taher Mun.
\newblock Prefix-free parsing for building big {BWT}s.
\newblock {\em Algorithms Mol. Biol.}, 14(1):13:1--13:15, 2019.

\bibitem{BW94}
Michael Burrows and David~J. Wheeler.
\newblock A block sorting lossless data compression algorithm.
\newblock Technical Report 124, Digital Equipment Corporation, 1994.

\bibitem{CazauxR19}
Bastien Cazaux and Eric Rivals.
\newblock Linking {BWT} and {XBW} via {Aho-Corasick} automaton: Applications to
  run-length encoding.
\newblock In {\em Proc.\ of 30th Annual Symposium on Combinatorial Pattern
  Matching ({CPM} 2019)}, volume 128 of {\em LIPIcs}, pages 24:1--24:20, 2019.

\bibitem{CenzatoGLR23}
Davide Cenzato, Veronica Guerrini, {\relax Zs}uzsanna Lipt{\'{a}}k, and
  Giovanna Rosone.
\newblock Computing the optimal {BWT} of very large string collections.
\newblock In Ali Bilgin, Michael~W. Marcellin, Joan Serra{-}Sagrist{\`{a}}, and
  James~A. Storer, editors, {\em Data Compression Conference, {DCC} 2023,
  Snowbird, UT, USA, March 21-24, 2023}, pages 71--80. {IEEE}, 2023.

\bibitem{CenzatoL22}
Davide Cenzato and {\relax Zs}uzsanna Lipt{\'{a}}k.
\newblock A theoretical and experimental analysis of {BWT} variants for string
  collections.
\newblock In Hideo Bannai and Jan Holub, editors, {\em 33rd Annual Symposium on
  Combinatorial Pattern Matching, {CPM} 2022, June 27-29, 2022, Prague, Czech
  Republic}, volume 223 of {\em LIPIcs}, pages 25:1--25:18. Schloss Dagstuhl -
  Leibniz-Zentrum f{\"{u}}r Informatik, 2022.

\bibitem{CobasGN21}
Dustin Cobas, Travis Gagie, and Gonzalo Navarro.
\newblock A fast and small subsampled $r$-index.
\newblock In {\em Proc.\ of 32nd Annual Symposium on Combinatorial Pattern
  Matching ({CPM} 2021)}, volume 191 of {\em LIPIcs}, pages 13:1--13:16.
  Schloss Dagstuhl - Leibniz-Zentrum f{\"{u}}r Informatik, 2021.

\bibitem{CoxBJR12}
Anthony~J. Cox, Markus~J. Bauer, Tobias Jakobi, and Giovanna Rosone.
\newblock Large-scale compression of genomic sequence databases with the
  {B}urrows-{W}heeler transform.
\newblock {\em Bioinform.}, 28(11):1415--1419, 2012.

\bibitem{DominguezN21}
Diego D{\'{\i}}az{-}Dom{\'{\i}}nguez and Gonzalo Navarro.
\newblock Efficient construction of the extended {BWT} from grammar-compressed
  {DNA} sequencing reads.
\newblock {\em CoRR}, abs/2102.03961, 2021.
\newblock URL: \url{https://arxiv.org/abs/2102.03961}.

\bibitem{DominguezN22}
Diego D{\'{\i}}az{-}Dom{\'{\i}}nguez and Gonzalo Navarro.
\newblock Efficient construction of the {BWT} for repetitive text using string
  compression.
\newblock In {\em Proc.\ of 33rd Annual Symposium on Combinatorial Pattern
  Matching ({CPM} 2022)}, volume 223 of {\em LIPIcs}, pages 29:1--29:18, 2022.

\bibitem{16S2018}
Robert~C Edgar.
\newblock {Updating the 97\% identity threshold for 16{S} ribosomal {RNA}
  {OTU}s}.
\newblock {\em Bioinf.}, 34(14):2371--2375, 2018.

\bibitem{EgidiLMT19}
Lavinia Egidi, Felipe~A. Louza, Giovanni Manzini, and Guilherme~P. Telles.
\newblock External memory {BWT} and {LCP} computation for sequence collections
  with applications.
\newblock {\em Algorithms Mol. Biol.}, 14(1):6:1--6:15, 2019.

\bibitem{FerraginaGM12}
Paolo Ferragina, Travis Gagie, and Giovanni Manzini.
\newblock Lightweight data indexing and compression in external memory.
\newblock {\em Algorithmica}, 63(3):707--730, 2012.

\bibitem{xbwt}
Paolo Ferragina, Fabrizio Luccio, Giovanni Manzini, and S.~Muthukrishnan.
\newblock Structuring labeled trees for optimal succinctness, and beyond.
\newblock In {\em Proc.\ of 46th IEEE Symposium on Foundations of Computer
  Science ({FOCS} 2005)}, pages 184--193, 2005.

\bibitem{FerraginaLMM09}
Paolo Ferragina, Fabrizio Luccio, Giovanni Manzini, and S.~Muthukrishnan.
\newblock Compressing and indexing labeled trees, with applications.
\newblock {\em J. {ACM}}, 57(1):4:1--4:33, 2009.

\bibitem{sais-lite-lcp}
Johannes Fischer and Florian Kurpicz.
\newblock sais-lite-lcp.
\newblock \url{https://github.com/kurpicz/sais-lite-lcp}.
\newblock Accessed: 2022-02-05.

\bibitem{GagieGM21}
Travis Gagie, Garance Gourdel, and Giovanni Manzini.
\newblock Compressing and indexing aligned readsets.
\newblock In {\em Proc.\ of 21st International Workshop on Algorithms in
  Bioinformatics ({WABI} 2021)}, volume 201 of {\em LIPIcs}, pages 13:1--13:21,
  2021.

\bibitem{GagieNP20}
Travis Gagie, Gonzalo Navarro, and Nicola Prezza.
\newblock Fully functional suffix trees and optimal text searching in
  {BWT}-runs bounded space.
\newblock {\em Journal of the {ACM}}, 67(1):2:1--2:54, 2020.

\bibitem{GilScott12}
Joseph~Yossi Gil and David~Allen Scott.
\newblock A bijective string sorting transform.
\newblock {\em CoRR}, abs/1201.3077, 2012.

\bibitem{GiulianiILPST21}
Sara Giuliani, Shunsuke Inenaga, Zsuzsanna Lipt{\'{a}}k, Nicola Prezza,
  Marinella Sciortino, and Anna Toffanello.
\newblock Novel results on the number of runs of the
  {B}urrows-{W}heeler-{T}ransform.
\newblock In {\em Proc.\ of 47th International Conference on Current Trends in
  Theory and Practice of Computer Science ({SOFSEM} 2021)}, volume 12607 of
  {\em LNCS}, pages 249--262, 2021.

\bibitem{Greaney22}
Allison~J. Greaney et~al.
\newblock A {SARS}-{C}o{V}-2 variant elicits an antibody response with a
  shifted immunodominance hierarchy.
\newblock {\em PLOS Pathogens}, 18:1--27, 02 2022.

\bibitem{libsais}
Ilya Grebnov.
\newblock libsais.
\newblock \url{https://github.com/IlyaGrebnov/libsais}.
\newblock Accessed: 2022-02-05.

\bibitem{Gusfield1997}
Dan Gusfield.
\newblock {\em Algorithms on Strings, Trees, and Sequences - Computer Science
  and Computational Biology}.
\newblock Cambridge University Press, 1997.

\bibitem{HoltM14}
James Holt and Leonard McMillan.
\newblock Merging of multi-string {BWT}s with applications.
\newblock {\em Bioinform.}, 30(24):3524--3531, 2014.

\bibitem{KempaK20}
Dominik Kempa and Tomasz Kociumaka.
\newblock Resolution of the {B}urrows-{W}heeler {T}ransform conjecture.
\newblock In {\em Proc.\ of 61st {IEEE} Annual Symposium on Foundations of
  Computer Science ({FOCS} 2020)}, pages 1002--1013, 2020.

\bibitem{KopplHHS20}
Dominik K{\"{o}}ppl, Daiki Hashimoto, Diptarama Hendrian, and Ayumi Shinohara.
\newblock In-place {Bijective Burrows-Wheeler Transforms}.
\newblock In {\em Proc.\ of 31st Annual Symposium on Combinatorial Pattern
  Matching ({CPM} 2020)}, volume 161 of {\em LIPIcs}, pages 21:1--21:15, 2020.

\bibitem{KucherovTV13}
Gregory Kucherov, Lilla T{\'{o}}thm{\'{e}}r{\'{e}}sz, and St{\'{e}}phane
  Vialette.
\newblock On the combinatorics of suffix arrays.
\newblock {\em Inf Process Lett}, 113(22-24):915--920, 2013.

\bibitem{KuhnleMBGLM19}
Alan Kuhnle, Taher Mun, Christina Boucher, Travis Gagie, Ben Langmead, and
  Giovanni Manzini.
\newblock Efficient construction of a complete index for pan-genomics read
  alignment.
\newblock In {\em Proc.\ of 23rd Annual Conference in Computational Molecular
  Biology ({RECOMB} 2019)}, volume 11467 of {\em LNCS}, pages 158--173, 2019.

\bibitem{bowtie2}
Ben Langmead and Steven~L Salzberg.
\newblock Fast gapped-read alignment with {B}owtie 2.
\newblock {\em Nature Methods}, 9(4):357--359, 2012.

\bibitem{bowtie}
Ben Langmead, Cole Trapnell, Mihai Pop, and Steven~L Salzberg.
\newblock Ultrafast and memory-efficient alignment of short {DNA} sequences to
  the human genome.
\newblock {\em Genome Biology}, 10:R25, 2009.

\bibitem{Li14a}
Heng Li.
\newblock Fast construction of {FM}-index for long sequence reads.
\newblock {\em Bioinform.}, 30(22):3274--3275, 2014.

\bibitem{bwa}
Heng Li and Richard Durbin.
\newblock Fast and accurate long-read alignment with {Burrows-Wheeler}
  transform.
\newblock {\em Bioinformatics}, 26(5):589--595, 2010.

\bibitem{LiuLL14}
Chi{-}Man Liu, Ruibang Luo, and Tak~Wah Lam.
\newblock {GPU}-accelerated {BWT} construction for large collection of short
  reads.
\newblock {\em CoRR}, abs/1401.7457, 2014.

\bibitem{LouzaGT17b}
Felipe~A. Louza, Simon Gog, and Guilherme~P. Telles.
\newblock Inducing enhanced suffix arrays for string collections.
\newblock {\em Theor. Comput. Sci.}, 678:22--39, 2017.

\bibitem{LouzaTGPR20}
Felipe~A. Louza, Guilherme~P. Telles, Simon Gog, Nicola Prezza, and Giovanna
  Rosone.
\newblock gsufsort: constructing suffix arrays, {LCP} arrays and {BWT}s for
  string collections.
\newblock {\em Algorithms Mol. Biol.}, 15(1):18, 2020.

\bibitem{Louza2017d}
Felipe~A. Louza, Guilherme~P. Telles, Steve Hoffmann, and Cristina~Dutra
  de~Aguiar~Ciferri.
\newblock Generalized enhanced suffix array construction in external memory.
\newblock {\em Algorithms Mol. Biol.}, 12(1):26:1--26:16, 2017.

\bibitem{MakinenN05}
Veli M{\"{a}}kinen and Gonzalo Navarro.
\newblock Succinct suffix arrays based on run-length encoding.
\newblock {\em Nordic Journal of Computing}, 12(1):40--66, 2005.

\bibitem{Mallick2016}
Swapan Mallick et~al.
\newblock {T}he {S}imons {G}enome {D}iversity {P}roject: 300 genomes from 142
  diverse populations.
\newblock {\em Nature}, 538(7624):201--206, 2016.

\bibitem{MantaciRRS07}
Sabrina Mantaci, Antonio Restivo, Giovanna Rosone, and Marinella Sciortino.
\newblock An extension of the {B}urrows-{W}heeler {T}ransform.
\newblock {\em Theor. Comput. Sci.}, 387(3):298--312, 2007.

\bibitem{manzini16}
Giovanni Manzini.
\newblock {XBWT} tricks.
\newblock In {\em Proc.\ of 23rd International Symposium on String Processing
  and Information Retrieval ({SPIRE} 2016)}, volume 9954 of {\em LNCS}, pages
  80--92, 2016.

\bibitem{Masillo23}
Francesco Masillo.
\newblock Matching statistics speed up {BWT} construction.
\newblock In {\em Proc.\ of 31st Annual European Symposium on Algorithms ({ESA}
  2023)}, volume 274 of {\em LIPIcs}, pages 83:1--83:15, 2023.

\bibitem{libdivsufsort}
Yuta Mori.
\newblock libdivsufsort.
\newblock \url{https://github.com/y-256/libdivsufsort}.
\newblock Accessed: 2022-02-05.

\bibitem{Navarro21a}
Gonzalo Navarro.
\newblock Indexing highly repetitive string collections, part {I:}
  repetitiveness measures.
\newblock {\em {ACM} Comput. Surv.}, 54(2):29:1--29:31, 2021.

\bibitem{10K}
Genome 10K~Community of~Scientists.
\newblock A proposal to obtain whole-genome sequence for 10,000 vertebrate
  species.
\newblock {\em J Hered.}, 100:659-674, 2009.

\bibitem{Ohlebusch2013}
Enno Ohlebusch.
\newblock {\em Bioinformatics Algorithms: Sequence Analysis, Genome
  Rearrangements, and Phylogenetic Reconstruction}.
\newblock Oldenbusch Verlag, 2013.

\bibitem{OhlebuschSB18}
Enno Ohlebusch, Stefan Stau{\ss}, and Uwe Baier.
\newblock Trickier {XBWT} tricks.
\newblock In {\em Proc.\ of 25th International Symposium in String Processing
  and Information Retrieval ({SPIRE} 2018)}, volume 11147 of {\em LNCS}, pages
  325--333, 2018.

\bibitem{OlivaGB23}
Marco Oliva, Travis Gagie, and Christina Boucher.
\newblock Recursive prefix-free parsing for building big {BWT}s.
\newblock In {\em Proc.\ of 33rd Data Compression Conference ({DCC} 2023)},
  pages 62--70, 2023.

\bibitem{Oliva0SMKGB21}
Marco Oliva, Massimiliano Rossi, Jouni Sir{\'{e}}n, Giovanni Manzini, Tamer
  Kahveci, Travis Gagie, and Christina Boucher.
\newblock Efficiently merging r-indexes.
\newblock In {\em Proc.\ of 31st Data Compression Conference ({DCC} 2021)},
  pages 203--212, 2021.

\bibitem{Pantaleoni14}
Jacopo Pantaleoni.
\newblock {BWT} of large string sets.
\newblock {\em CoRR}, abs/1410.0562, 2014.

\bibitem{PuglisiZ21}
Simon~J. Puglisi and Bella Zhukova.
\newblock Document retrieval hacks.
\newblock In {\em Proc.\ of 19th International Symposium on Experimental
  Algorithms ({SEA} 2021)}, volume 190 of {\em LIPIcs}, pages 12:1--12:12,
  2021.

\bibitem{Siren16}
Jouni Sir{\'{e}}n.
\newblock {B}urrows-{W}heeler {T}ransform for terabases.
\newblock In {\em Proc.\ of 26th Data Compression Conference ({DCC} 2016)},
  pages 211--220, 2016.

\bibitem{STARR20201295}
Tyler~N. Starr et~al.
\newblock Deep mutational scanning of {SARS}-{C}o{V}-2 receptor binding domain
  reveals constraints on folding and {ACE}2 binding.
\newblock {\em Cell}, 182(5):1295--1310.e20, 2020.

\bibitem{rice}
C.~Sun et~al.
\newblock {RPAN: rice pan-genome browser for 3000 rice genomes}.
\newblock {\em Nucleic Acids Res}, 45(2):597--605, 2017.

\bibitem{1000genomes}
{The 1000 Genomes Project Consortium}.
\newblock A global reference for human genetic variation.
\newblock {\em Nature}, 526:68--74, 2015.

\bibitem{arab}
{The 1001 Genomes Consortium}.
\newblock {Epigenomic Diversity in a Global Collection of Arabidopsis thaliana
  Accessions}.
\newblock {\em Cell}, 166(2):492--505, 2016.

\bibitem{100K}
C.~Turnbull et~al.
\newblock The 100,000 genomes project: bringing whole genome sequencing to the
  {NHS}.
\newblock {\em Br Med J}, 361, 2018.

\bibitem{Influenza2015}
Silvie Van~den Hoecke, Judith Verhelst, Marnik Vuylsteke, and Xavier Saelens.
\newblock Analysis of the genetic diversity of influenza {A} viruses using
  next-generation {DNA} sequencing.
\newblock {\em BMC Genomics}, 16(1):79, 2015.

\bibitem{Winand2019}
Raf Winand et~al.
\newblock Targeting the 16s r{RNA} gene for bacterial identification in complex
  mixed samples: Comparative evaluation of second ({I}llumina) and third
  ({O}xford nanopore technologies) generation sequencing technologies.
\newblock {\em Int. J. of Mol. Sci.}, 21(1):298, 2019.

\bibitem{candida2019}
Michael~H. Woodworth et~al.
\newblock {Sentinel case of {C}andida auris in the Western United States
  Following Prolonged Occult Colonization in a Returned Traveler from {I}ndia}.
\newblock {\em Microb Drug Resist}, 25(5):677--680, 2019.

\end{thebibliography}

\newpage

\appendix


\section{Experimental setup}
All datasets are stored in FASTA format. 

We used three tools for computing the five \BWT\ variants; \pfpebwt, {\tt ropebwt2} and {\tt Big-BWT}. In order to make the BWTs comparable we did some adaptations to both tools and inputs. We modified {\tt ropebwt2} to make it work with the same character order as the other tools, i.e.\ $\tt \$<A<C<G<N<T$. Then we used {\tt ropebwt2} for computing both the \mdollar\ and the \rlo\ using the {\tt -R} and {\tt -R -s} flags respectively. We used \ours\ for constructing both the \eBWT\ and the \dollarebwt\ variants. In order to compute the \dollarebwt, we modified the input files, appending an end-of-string character at the end of each sequence. Finally, for computing the \concat, we removed the headers from the FASTA files, arranging the sequences in newline separated files, and ran {\tt Big-BWT} without additional flags on these newline separated files. 

\section{Further information on the tools}
We tested all 18 tools extensively, and determined which data structure they compute, using both our tests and the algorithm descriptions in the respective papers. In this section, we include further information about some of these tools.

\begin{itemize}
    \item {\tt pfpebwt} is a tool computing the \eBWT\ of string collections (\url{https://github.com/davidecenzato/PFP-eBWT.git}). It takes in input a fasta file and gives in output the \eBWT\ in either plain ASCII text or RLE (run-length-encoded) format. We used (a) no flags for long sequences, and (b) the flags {\tt -w 10 -p 10 -n 3 -{}-reads}  for short sequences. We included it in two different rows of Table~\ref{table:bwts} because by default {\tt pfpebwt} computes the \eBWT, but it can compute the \dollarebwt\ if the sequences have explicit end-of-string characters (not in multi-thread mode).
   \item {\tt cais} is a tool implementing the {\tt SAIS\_for\_eBWT} algorithm \cite{BoucherCLMM21}, which computes both the eBWT and the dolEBWT (\url{https://github.com/davidecenzato/cais.git}) depending on the input flag. It takes in input a fasta file, a fastq file, or a plain text file and gives in output the final transform in plain ASCII text. The {\tt -c} and {\tt -a} flags enable to output the conjugate array along with the resulting BWT.
    \item {\tt G2BWT} is a tool computing the \dollarebwt\ of short sequence collections (\url{https://bitbucket.org/DiegoDiazDominguez/lms_grammar/src/bwt_imp2}). It takes in input newline separated files. Even though it is not stated explicitly, this tool computes the \dollarebwt\ because, when it constructs the grammar, it uses dollars for separating adjacent strings. Thus, also the string rotations will contain dollars. We tested it using the default settings.
    
    \item {\tt msbwt} is a tool implementing the Holt and McMillan~\cite{HoltM14} merge-based BWT construction algorithm  (\url{https://github.com/holtjma/msbwt.git}). It takes in input a list of one or several fastq files. Even if this tool uses the {\tt BCR} approach \cite{BauerCR13} for computing the BWTs to merge, it actually computes the \dollarebwt . This is because it features a preprocessing where it sorts the input strings lexicographically. Thus, the resulting \mdollar\ corresponds to the \dollarebwt. 
    
   \item {\tt BEETL} is a suite containing several tools, including a tool computing the \mdollar\ of string collections using an implementation of the BCR and BCR-ext algorithms \cite{BauerCR13} (\url{https://github.com/BEETL/BEETL.git}). This tool requires that all input sequences have to have the same length. We tested this tool using {\tt --output-format ASCII} and {\tt --concatenate-output} flags. This tool also computes the a BWT variant similar to the \colex\ by using the {\tt --sap-ordering} flag (BCR-ext mode only).
    
    \item {\tt BCR\_LCP\_GSA} is a tool computing the \mdollar\ of sting collections in semi-external memory (\url{https://github.com/giovannarosone/BCR_LCP_GSA}). It implements an algorithm similar to {\tt BCR} contained in the BEETL tool, but it can process a string collection containing sequences of different lengths. It takes in input a fasta file, a fastq file, or a gz-compressed fastq file. It computes the \mdollar\ following the method of Bauer et al., described in \cite{BauerCR13}. We set the {\tt 'dataTypeLengthSequences'} variable in {\tt Parameters.h} to 1.
    
    \item {\tt ropebwt2} is a tool computing the FM-index and the \mdollar\ of string collections   (\url{https://github.com/lh3/ropebwt2.git}), using an approach similar to {\tt BCR}. It takes in input a fasta file, a fastq file, or a gz compressed fastq file. We listed it in two different rows of Table~\ref{table:bwts} because it computes the \mdollar\ or the \colex, depending on the flags. We used the {\tt -R} and the {\tt -R -s} flags, respectively, to obtain the two transforms. In addition, we modified {\tt main.c} in order to change the order of the characters to {\tt \$ < A < C < G < N < T}. 
    
    \item {\tt merge-BWT} computes the \mdollar\ of a string collection by merging the BWTs of subcollections of the input  (\url{https://github.com/jltsiren/bwt-merge.git}). It takes in input a list of one or several \mdollar s. The order of the dollars will depend on the order in which the input BWTs are listed. We tested it using {\tt -i plain\_sorted and -o plain\_sorted} flags. We computed the BWTs of the subcollections using {\tt ropebwt2}. 
    
    \item {\tt nvSetBWT} is a tool included in {\tt nvbio} suite (\url{https://github.com/NVlabs/nvbio.git}). It takes in input either a fastq or a newline separated file. We tested it using the {\tt -R} flag for skipping the reverse strand. However, even if the algorithmic descriptions in~\cite{Pantaleoni14, LiuLL14} seem to describe the \mdollar, the output of the current version (version 1.1) does not correspond to a possible BWT because the Parikh vector is different from that of the input. 
    
    \item {\tt eGSA} computes the generalized enhanced suffix array and the \mdollar\ of a string collection (\url{https://github.com/felipelouza/egsa.git}). It takes in input a text file, a fasta file, or a fastq file. It uses the gSACA-K algorithm for computing the suffix array of subcollections of the input and then merges all suffix arrays. Thus it  computes the \mdollar. We tested it with the {\tt -b} flag.  
    
    \item {\tt eGAP} computes the \mdollar, and optionally the LCP-array (longest common prefix array) and DA (document array) of a string collection (\url{https://github.com/felipelouza/egap.git}). It takes in input a newline separated file, a fasta file, or a fastq file. We tested it with default settings.
    
    \item {\tt bwt-lcp-parallel} computes the \mdollar\ and the LCP-array of a collection of short sequences  (\url{https://github.com/AlgoLab/bwt-lcp-parallel.git}). It takes in input fasta files and does not support the {\tt N} character. We tested it using standard settings.
    
    \item {\tt gsufsort} computes the SA, LCP and \mdollar\ of a string collection (\url{https://github.com/felipelouza/gsufsort.git}), using the gSACA-K algorithm of~\cite{LouzaGT17b}. It takes in input a newline separated file, a fasta file, or a fastq file. We tested it using {\tt -{}-fasta} and {\tt -{}-bwt} flags. 

    \item {\tt grlBWT} is a tool computing the \mdollar\ of string collections using an induced suffix sorting based algorithm that keeps the intermediate data structures in compressed form (\url{https://github.com/ddiazdom/grlBWT}). It takes in input a concatenated string collection and gives in output the \mdollar\ in run-length compressed form. We tested it with the default parameters and used new-line separated files as input.
    
    \item {\tt BigBWT} computes the \concat, and optionally the suffix array, of a highly repetitive text or string collection (\url{https://github.com/alshai/Big-BWT.git}) using the Prefix-free parsing (PFP) algorithm. It takes in input a newline separated file or a fasta file. This tool with the {\tt -f} flag is used internally in the $r$-index implementation (\url{https://github.com/alshai/r-index}), producing the BWT of the strings concatenated without dollars, thus, the end-of-string symbols have to be added explicitly. On the other hand, the tool without the {\tt -f} flag will compute the BWT of the fasta files without skipping the fasta headers. We used standard parameters and as input newline separated files, the output then is the \concat.

    \item {\tt r-pfbwt} is a tool which computes the run-length encoded \concat\ by using a similar algorithm than {\tt BigBWT} (\url{https://github.com/marco-oliva/r-pfbwt}). However, unlike {\tt BigBWT}, {\tt r-pfbwt} employs an improved version of the PFP algorithm, which allows the process of even larger datasets through a recursive pre-processing of the input. We tested it using the {\tt --bwt-only} flag and computed the PFP data structures using the {\tt pfp++} software (\url{https://github.com/marco-oliva/pfp.git}).

    \item {\tt CMS-BWT} is a tool computing the \concat\ by using the matching statistics to speed up the BWT computation and reduce the memory footprint for large and repetitive datasets (\url{https://github.com/fmasillo/CMS-BWT.git}). Unlike the other software it requires two input files, one containing a string collection and another containing a reference sequence. It takes in input fasta files and outputs the resulting BWT in plain format or run-length encoding. We tested it using the default parameters.

    \item {\tt optimalBWT} is a tool computing the optimal BWT of Bentley et al., it features two different construction algorithms, a variant of SAIS of Nong et al.\ which works in internal memory and a variant of BCR working in semi-external memory (\url{https://github.com/davidecenzato/optimalBWT.git}). It takes in input either a fasta or fastq file and outputs the resulting BWT in plain ascii text. We tested it using both {\tt -a sais} and {\tt -a bcr} flags.
\end{itemize}

\section{Results on individual datasets}

\begin{table*}[t!]
\raggedright
\textbf{SARS-CoV-2 short (500,000 short sequences)}

\vspace{4mm}

\captionsetup{width=\linewidth}
\begin{adjustbox}{max width=140mm}
\setlength{\tabcolsep}{5pt}
\renewcommand{\arraystretch}{1.6}
\begin{tabular}{|r||rrrr|}
  \hline
  \multirow{2}{*}{{\diagbox[width=4.7cm, height=1.18cm]{\rlap{\enspace\raisebox{0ex}{ \it \hspace{-3.0mm} norm.\ Hamming d.\  }}}{\raisebox{-0ex}{\it \hspace{-3.0mm} Hamming d.\ }}}} & \multicolumn{4}{l|}{\it Hamming distance on the big dataset} \\
  \cline{2-5}
  & \multicolumn{1}{R{1.6cm}|}{\dollarebwt} & \multicolumn{1}{R{1.6cm}|}{\mdollar} & \multicolumn{1}{R{1.6cm}|}{\concat} & \multicolumn{1}{R{1.6cm}|}{\rlo} \\
  \hline \hline
 \multicolumn{1}{|l||}{\dollarebwt}& \multicolumn{1}{l|}{\hspace{5.0mm} 0} & \multicolumn{1}{r|}{3,014,183} & \multicolumn{1}{r|}{2,926,602} & \multicolumn{1}{r|}{2,912,860} \\
 \hline
 \multicolumn{1}{|l||}{\mdollar} & \multicolumn{1}{r|}{0.11820} & \multicolumn{1}{l|}{\hspace{5.0mm} 0}  & \multicolumn{1}{r|}{3,013,908} & \multicolumn{1}{r|}{3,102,887} \\ 
 \hline
 \multicolumn{1}{|l||}{\concat} & \multicolumn{1}{r|}{0.11477} & \multicolumn{1}{r|}{0.11819} & \multicolumn{1}{l|}{\hspace{5.0mm} 0} & \multicolumn{1}{r|}{3,013,634} \\ 
 \hline
 \multicolumn{1}{|l||}{\rlo} & \multicolumn{1}{r|}{0.11423} & \multicolumn{1}{r|}{0.12168} & \multicolumn{1}{r|}{0.11818} & \multicolumn{1}{l|}{\hspace{5.0mm} 0} \\ 
 \hline
\end{tabular}
\quad
\begin{tabular}{|l||r|}
\hline
\multicolumn{2}{|l|}{\it dataset properties} \\
\hline \hline
no.\ sequences & 500,000 \\
\hline
average length & 50 \\
\hline
total length & 25,000,000 \\
\hline
no.\ of interesting intervals & 116,598 \\
\hline
total length intr.int.s & 20,187,840 \\
\hline
fraction pos.s in intr.int.s  & 0.792 \\
\hline
variability & 0.210 \\
\hline
\end{tabular}
\end{adjustbox}
\newline
\vspace*{0.4 cm}
\newline
\begin{adjustbox}{max width=140mm}
\setlength{\tabcolsep}{5pt}
\renewcommand{\arraystretch}{1.6}
\hspace{14cm}
\begin{tabular}{|l||r|r|}
\hline
\multicolumn{3}{|l|}{\it no.\ runs big dataset} \\
\hline \hline
& $r$ & $n/r$ \\
\hline
\eBWT & 1,902,148 & $13.143$ \\
\hline
\dollarebwt & 1,868,581 & $13.647$ \\
\hline
\mdollar & 3,113,818 & $8.189$ \\
\hline
\concat & 3,402,513 & $7.494$ \\
\hline
\rlo & 808,906 & $31.524$ \\
\hline
\optimal & 725,979 & $35.125$ \\
\hline
\end{tabular}
\end{adjustbox}
\newline
\vspace*{0.4 cm}
\newline
\begin{adjustbox}{max width=140mm}
\setlength{\tabcolsep}{5pt}
\renewcommand{\arraystretch}{1.6}
\begin{tabular}{|r||rrrr|}
  \hline
  \multirow{2}{*}{{\diagbox[width=4.7cm, height=1.18cm]{\rlap{\enspace\raisebox{0ex}{ \it \hspace{-3.0mm} norm.\ Hamming d.\  }}}{\raisebox{-0ex}{\it \hspace{-3.0mm} Hamming d.\ }}}} & \multicolumn{4}{l|}{\it Hamming distance on a subset of 5,000 sequences} \\
  \cline{2-5}
  & \multicolumn{1}{R{1.6cm}|}{\dollarebwt} & \multicolumn{1}{R{1.6cm}|}{\mdollar} & \multicolumn{1}{R{1.6cm}|}{\concat} & \multicolumn{1}{R{1.6cm}|}{\rlo} \\
  \hline \hline
 \multicolumn{1}{|l||}{\dollarebwt}& \multicolumn{1}{l|}{\hspace{5.0mm} 0} & \multicolumn{1}{r|}{21,362} & \multicolumn{1}{r|}{21,196} & \multicolumn{1}{r|}{20,626} \\
 \hline
 \multicolumn{1}{|l||}{\mdollar} & \multicolumn{1}{r|}{0.08377} & \multicolumn{1}{l|}{\hspace{5.0mm} 0}  & \multicolumn{1}{r|}{21,376} & \multicolumn{1}{r|}{21,256} \\ 
 \hline
 \multicolumn{1}{|l||}{\concat} & \multicolumn{1}{r|}{0.08312} & \multicolumn{1}{r|}{0.08383} & \multicolumn{1}{l|}{\hspace{5.0mm} 0} & \multicolumn{1}{r|}{21,259} \\ 
 \hline
 \multicolumn{1}{|l||}{\rlo} & \multicolumn{1}{r|}{0.08089} & \multicolumn{1}{r|}{0.08336} & \multicolumn{1}{r|}{0.08337} & \multicolumn{1}{l|}{\hspace{5.0mm} 0} \\ 
 \hline
\end{tabular}
\quad
\begin{tabular}{|l||r|}
\hline
\multicolumn{2}{|l|}{\it small dataset properties} \\
\hline \hline
no.\ of sequences & 5,000 \\
\hline
total length & 250,000 \\
\hline
average length & 100 \\
\hline
no.\ of interesting intervals & 2,476 \\
\hline
total length intr.int.s & 180,038 \\
\hline
fraction pos.s in intr.int.s  & 0.706 \\
\hline
variability & 0.173 \\
\hline
\end{tabular}
\end{adjustbox}
\newline
\vspace*{0.4 cm}
\newline
\begin{adjustbox}{max width=140mm}
\setlength{\tabcolsep}{5pt}
\renewcommand{\arraystretch}{1.6}
\begin{tabular}{|r||r|r|r|r|r|}
  \hline
  \multirow{2}{*}{{\diagbox[width=4.0cm, height=1.18cm]{\rlap{\enspace\raisebox{0ex}{ \it \hspace{-3.0mm} norm.\ edit d.\  }}}{\raisebox{+0ex}{\it edit d.\ }} }} & \multicolumn{5}{l|}{\it edit distance on a subset of 5,000 sequences} \\
  \cline{2-6}
  & \multicolumn{1}{R{1.6cm}|}{\eBWT} &\multicolumn{1}{R{1.6cm}|}{\dollarebwt} & \multicolumn{1}{R{1.6cm}|}{\mdollar} & \multicolumn{1}{R{1.6cm}|}{\concat} & \multicolumn{1}{R{1.6cm}|}{\rlo} \\
  \hline \hline
 \multicolumn{1}{|l||}{\eBWT} & \multicolumn{1}{l|}{\hspace{5.0mm} 0} & \multicolumn{1}{r|}{28,702} & \multicolumn{1}{r|}{43,903} & \multicolumn{1}{r|}{43,828} & \multicolumn{1}{r|}{46,936} \\ 
 \hline
 \multicolumn{1}{|l||}{\dollarebwt}& \multicolumn{1}{r|}{0.11256} & \multicolumn{1}{l|}{\hspace{5.0mm} 0} & \multicolumn{1}{r|}{17,000} & \multicolumn{1}{r|}{16,921} & \multicolumn{1}{r|}{20,104}  \\ 
 \hline
 \multicolumn{1}{|l||}{\mdollar} & \multicolumn{1}{r|}{0.17217} & \multicolumn{1}{r|}{0.06667} & \multicolumn{1}{l|}{\hspace{5.0mm} 0}  & \multicolumn{1}{r|}{16,130} & \multicolumn{1}{r|}{20,812} \\ \cline{1-6}
 \multicolumn{1}{|l||}{\concat} & \multicolumn{1}{r|}{0.17187} & \multicolumn{1}{r|}{0.06636} & \multicolumn{1}{r|}{0.06325} & \multicolumn{1}{l|}{\hspace{5.0mm} 0} & \multicolumn{1}{r|}{20,830}  \\ 
 \hline
 \multicolumn{1}{|l||}{\rlo} & \multicolumn{1}{r|}{0.18406} & \multicolumn{1}{r|}{0.07884} & \multicolumn{1}{r|}{0.08162} & \multicolumn{1}{r|}{0.08169} & \multicolumn{1}{l|}{\hspace{5.0mm} 0} \\ 
 \hline
\end{tabular}
\quad
\begin{tabular}{|l||r|r|}
\hline
\multicolumn{3}{|l|}{\it no.\ runs small dataset} \\
\hline \hline
& $r$ & $n/r$ \\
\hline
\eBWT & 52,979 & $4.719$ \\
\hline
\dollarebwt & 50,803 & $5.019$ \\
\hline
\mdollar & 54,766 & $4.656$ \\
\hline
\concat & 54,698 & $4.662$ \\
\hline
\colex & 37,320 & $6.833$ \\
\hline
\optimal & 35,904 & $7.102$ \\
\hline
\end{tabular}
\end{adjustbox}
\vspace{2mm}
\caption{\label{tab:sars_short_complete} Results for the SARS-CoV-2 short dataset. First row left: absolute and normalized pairwise Hamming distance between separator-based BWT variants. First row right: summary of the dataset properties. Second row: number of runs and average runlength ($n/r$) of all BWT variants. Third row left: absolute and normalized pairwise Hamming distance between separator-based BWT variants on a subset of the input collection. Third row right: summary of the dataset properties of a subset of the input collection. Fourth row left: absolute and normalized pairwise edit distance between all BWT variants on a subset of the input collection. Fourth row right: number of runs and average runlength ($n/r$) of all BWT variants on a subset of the input collection.}
\end{table*}

\begin{table*}[t!]
\raggedright
\textbf{Simons Diversity reads (500,000 short sequences)}

\vspace{4mm}

\captionsetup{width=\linewidth}
\begin{adjustbox}{max width=140mm}
\setlength{\tabcolsep}{5pt}
\renewcommand{\arraystretch}{1.6}
\begin{tabular}{|r||rrrr|}
  \hline
  \multirow{2}{*}{{\diagbox[width=4.7cm, height=1.18cm]{\rlap{\enspace\raisebox{0ex}{ \it \hspace{-3.0mm} norm.\ Hamming d.\  }}}{\raisebox{-0ex}{\it \hspace{-3.0mm} Hamming d.\ }}}} & \multicolumn{4}{l|}{\it Hamming distance on the big dataset} \\
  \cline{2-5}
  & \multicolumn{1}{R{1.6cm}|}{\dollarebwt} & \multicolumn{1}{R{1.6cm}|}{\mdollar} & \multicolumn{1}{R{1.6cm}|}{\concat} & \multicolumn{1}{R{1.6cm}|}{\rlo} \\
  \hline \hline
 \multicolumn{1}{|l||}{\dollarebwt}& \multicolumn{1}{l|}{\hspace{5.0mm} 0} & \multicolumn{1}{r|}{3,624,283} & \multicolumn{1}{r|}{3,602,362} & \multicolumn{1}{r|}{3,594,438} \\
 \hline
 \multicolumn{1}{|l||}{\mdollar} & \multicolumn{1}{r|}{0.07249} & \multicolumn{1}{l|}{\hspace{5.0mm} 0}  & \multicolumn{1}{r|}{3,628,799} & \multicolumn{1}{r|}{3,623,154} \\ 
 \hline
 \multicolumn{1}{|l||}{\concat} & \multicolumn{1}{r|}{0.07133} & \multicolumn{1}{r|}{0.07186} & \multicolumn{1}{l|}{\hspace{5.0mm} 0} & \multicolumn{1}{r|}{3,617,679} \\ 
 \hline
 \multicolumn{1}{|l||}{\rlo} & \multicolumn{1}{r|}{0.07189} & \multicolumn{1}{r|}{0.07246} & \multicolumn{1}{r|}{0.07168} & \multicolumn{1}{l|}{\hspace{5.0mm} 0} \\ 
 \hline
\end{tabular}
\quad
\begin{tabular}{|l||r|}
\hline
\multicolumn{2}{|l|}{\it dataset properties} \\
\hline \hline
no.\ of sequences & 500,000 \\
\hline
total length & 50,000,000 \\
\hline
average length & 100 \\
\hline
no.\ of interesting intervals & 316,013 \\
\hline
total length intr.int.s & 5,387,549 \\
\hline
fraction pos.s in intr.int.s  & 0.107 \\
\hline
variability & 0.976 \\
\hline
\end{tabular}
\end{adjustbox}
\newline
\vspace*{0.4 cm}
\newline
\begin{adjustbox}{max width=140mm}
\setlength{\tabcolsep}{5pt}
\renewcommand{\arraystretch}{1.6}
\hspace{14.0cm}
\begin{tabular}{|l||r|r|}
\hline
\multicolumn{3}{|l|}{\it no.\ runs big dataset} \\
\hline \hline
& $r$ & $n/r$ \\
\hline
\eBWT & 8,974,105 & $5.572$ \\
\hline
\dollarebwt & 9,337,122 & $5.409$ \\
\hline
\mdollar & 9,362,564 & $5.394$ \\
\hline
\concat & 9,530,334 & $5.299$ \\
\hline
\rlo & 6,414,356 & $7.873$ \\
\hline
\optimal & 6,209,567 & $8.133$ \\
\hline
\end{tabular}
\end{adjustbox}
\newline
\vspace*{0.4 cm}
\newline
\begin{adjustbox}{max width=140mm}
\setlength{\tabcolsep}{5pt}
\renewcommand{\arraystretch}{1.6}
\begin{tabular}{|r||rrrr|}
  \hline
  \multirow{2}{*}{{\diagbox[width=4.7cm, height=1.18cm]{\rlap{\enspace\raisebox{0ex}{ \it \hspace{-3.0mm} norm.\ Hamming d.\  }}}{\raisebox{-0ex}{\it \hspace{-3.0mm} Hamming d.\ }}}} & \multicolumn{4}{l|}{\it Hamming distance on a subset of 5,000 sequences} \\
  \cline{2-5}
  & \multicolumn{1}{R{1.6cm}|}{\dollarebwt} & \multicolumn{1}{R{1.6cm}|}{\mdollar} & \multicolumn{1}{R{1.6cm}|}{\concat} & \multicolumn{1}{R{1.6cm}|}{\rlo} \\
  \hline \hline
 \multicolumn{1}{|l||}{\dollarebwt}& \multicolumn{1}{l|}{\hspace{5.0mm} 0} & \multicolumn{1}{r|}{23,742} & \multicolumn{1}{r|}{23,461} & \multicolumn{1}{r|}{23,535} \\
 \hline
 \multicolumn{1}{|l||}{\mdollar} & \multicolumn{1}{r|}{0.04748} & \multicolumn{1}{l|}{\hspace{5.0mm} 0}  & \multicolumn{1}{r|}{23,785} & \multicolumn{1}{r|}{23,722} \\ 
 \hline
 \multicolumn{1}{|l||}{\concat} & \multicolumn{1}{r|}{0.04646} & \multicolumn{1}{r|}{0.04710} & \multicolumn{1}{l|}{\hspace{5.0mm} 0} & \multicolumn{1}{r|}{23,660} \\ 
 \hline
 \multicolumn{1}{|l||}{\rlo} & \multicolumn{1}{r|}{0.04707} & \multicolumn{1}{r|}{0.04744} & \multicolumn{1}{r|}{0.04685} & \multicolumn{1}{l|}{\hspace{5.0mm} 0} \\ 
 \hline
\end{tabular}
\quad
\begin{tabular}{|l||r|}
\hline
\multicolumn{2}{|l|}{\it small dataset properties} \\
\hline \hline
no.\ of sequences & 5,000 \\
\hline
total length & 500,000 \\
\hline
average length & 100 \\
\hline
no.\ of interesting intervals & 3,111 \\
\hline
total length intr.int.s & 35,404 \\
\hline
fraction pos.s in intr.int.s  & 0.070 \\
\hline
variability & 0.989 \\
\hline
\end{tabular}
\end{adjustbox}
\newline
\vspace*{0.4 cm}
\newline
\begin{adjustbox}{max width=140mm}
\setlength{\tabcolsep}{5pt}
\renewcommand{\arraystretch}{1.6}
\begin{tabular}{|r||r|r|r|r|r|}
  \hline
  \multirow{2}{*}{{\diagbox[width=4.0cm, height=1.18cm]{\rlap{\enspace\raisebox{0ex}{ \it \hspace{-3.0mm} norm.\ edit d.\  }}}{\raisebox{+0ex}{\it edit d.\ }} }} & \multicolumn{5}{l|}{\it edit distance on a subset of 5,000 sequences} \\
  \cline{2-6}
  & \multicolumn{1}{R{1.6cm}|}{\eBWT} &\multicolumn{1}{R{1.6cm}|}{\dollarebwt} & \multicolumn{1}{R{1.6cm}|}{\mdollar} & \multicolumn{1}{R{1.6cm}|}{\concat} & \multicolumn{1}{R{1.6cm}|}{\rlo} \\
  \hline \hline
 \multicolumn{1}{|l||}{\eBWT} & \multicolumn{1}{l|}{\hspace{5.0mm} 0} & \multicolumn{1}{r|}{72,898} & \multicolumn{1}{r|}{72,878} & \multicolumn{1}{r|}{72,918} & \multicolumn{1}{r|}{74,026} \\ 
 \hline
 \multicolumn{1}{|l||}{\dollarebwt}& \multicolumn{1}{r|}{0.14435} & \multicolumn{1}{l|}{\hspace{5.0mm} 0} & \multicolumn{1}{r|}{17,820} & \multicolumn{1}{r|}{17,560} & \multicolumn{1}{r|}{22,481}  \\ 
 \hline
 \multicolumn{1}{|l||}{\mdollar} & \multicolumn{1}{r|}{0.14431} & \multicolumn{1}{r|}{0.03529} & \multicolumn{1}{l|}{\hspace{5.0mm} 0}  & \multicolumn{1}{r|}{17,726} & \multicolumn{1}{r|}{22,586} \\ \cline{1-6}
 \multicolumn{1}{|l||}{\concat} & \multicolumn{1}{r|}{0.14439} & \multicolumn{1}{r|}{0.03477} & \multicolumn{1}{r|}{0.03510} & \multicolumn{1}{l|}{\hspace{5.0mm} 0} & \multicolumn{1}{r|}{22,595}  \\ 
 \hline
 \multicolumn{1}{|l||}{\rlo} & \multicolumn{1}{r|}{0.14659} & \multicolumn{1}{r|}{0.04452} & \multicolumn{1}{r|}{0.04472} & \multicolumn{1}{r|}{0.04474} & \multicolumn{1}{l|}{\hspace{5.0mm} 0} \\ 
 \hline
\end{tabular}
\quad
\begin{tabular}{|l||r|r|}
\hline
\multicolumn{3}{|l|}{\it no.\ runs small dataset} \\
\hline \hline
& $r$ & $n/r$ \\
\hline
\eBWT & 77,646 & $6.439$ \\
\hline
\dollarebwt & 81,758 & $6.177$ \\
\hline
\mdollar & 81,883 & $6.167$ \\
\hline
\concat & 82,779 & $6.101$ \\
\hline
\colex & 64,229 & $7.862$ \\
\hline
\optimal & 62,117 & $8.130$ \\
\hline
\end{tabular}
\end{adjustbox}
\vspace{2mm}
\caption{\label{tab:simon_diversity_complete} Results for the Simons Diversity reads dataset. First row left: absolute and normalized pairwise Hamming distance between separator-based BWT variants. First row right: summary of the dataset properties. Second row: number of runs and average runlength ($n/r$) of all BWT variants. Third row left: absolute and normalized pairwise Hamming distance between separator-based BWT variants on a subset of the input collection. Third row right: summary of the dataset properties of a subset of the input collection. Fourth row left: absolute and normalized pairwise edit distance between all BWT variants on a subset of the input collection. Fourth row right: number of runs and average runlength ($n/r$) of all BWT variants on a subset of the input collection.}
\end{table*}

\begin{table*}[t!]
\raggedright
\textbf{16S rRNA short (500,000 short sequences)}

\vspace{4mm}

\captionsetup{width=\linewidth}
\begin{adjustbox}{max width=140mm}
\setlength{\tabcolsep}{5pt}
\renewcommand{\arraystretch}{1.6}
\begin{tabular}{|r||rrrr|}
  \hline
  \multirow{2}{*}{{\diagbox[width=4.7cm, height=1.18cm]{\rlap{\enspace\raisebox{0ex}{ \it \hspace{-3.0mm} norm.\ Hamming d.\  }}}{\raisebox{-0ex}{\it \hspace{-3.0mm} Hamming d.\ }}}} & \multicolumn{4}{l|}{\it Hamming distance on the big dataset} \\
  \cline{2-5}
  & \multicolumn{1}{R{1.6cm}|}{\dollarebwt} & \multicolumn{1}{R{1.6cm}|}{\mdollar} & \multicolumn{1}{R{1.6cm}|}{\concat} & \multicolumn{1}{R{1.6cm}|}{\rlo} \\
  \hline \hline
 \multicolumn{1}{|l||}{\dollarebwt}& \multicolumn{1}{l|}{\hspace{5.0mm} 0} & \multicolumn{1}{r|}{2,202,008} & \multicolumn{1}{r|}{2,540,310} & \multicolumn{1}{r|}{1,748,072} \\
 \hline
 \multicolumn{1}{|l||}{\mdollar} & \multicolumn{1}{r|}{0.02881} & \multicolumn{1}{l|}{\hspace{5.0mm} 0}  & \multicolumn{1}{r|}{2,201,003} & \multicolumn{1}{r|}{2,202,717} \\ 
 \hline
 \multicolumn{1}{|l||}{\concat} & \multicolumn{1}{r|}{0.03324} & \multicolumn{1}{r|}{0.02880} & \multicolumn{1}{l|}{\hspace{5.0mm} 0} & \multicolumn{1}{r|}{2,784,600} \\ 
 \hline
 \multicolumn{1}{|l||}{\rlo} & \multicolumn{1}{r|}{0.02287} & \multicolumn{1}{r|}{0.02882} & \multicolumn{1}{r|}{0.03643} & \multicolumn{1}{l|}{\hspace{5.0mm} 0} \\ 
 \hline
\end{tabular}
\quad
\begin{tabular}{|l||r|}
\hline
\multicolumn{2}{|l|}{\it dataset properties} \\
\hline \hline
no.\ of sequences & 500,000 \\
\hline
total length & 75,929,833 \\
\hline
average length & 152 \\
\hline
no.\ of interesting intervals & 54,366 \\
\hline
total length intr.int.s & 56,708,529 \\
\hline
fraction pos.s in intr.int.s  & 0.742 \\
\hline
variability & 0.058 \\
\hline
\end{tabular}
\end{adjustbox}
\newline
\vspace*{0.4 cm}
\newline
\begin{adjustbox}{max width=140mm}
\setlength{\tabcolsep}{5pt}
\renewcommand{\arraystretch}{1.6}
\hspace{14cm}
\begin{tabular}{|l||r|r|}
\hline
\multicolumn{3}{|l|}{\it no.\ runs big dataset} \\
\hline \hline
& $r$ & $n/r$ \\
\hline
\eBWT & 1,992,130 & $38.115$ \\
\hline
\dollarebwt & 1,992,211 & $38.364$ \\
\hline
\mdollar & 4,057,541 & $18.836$ \\
\hline
\concat & 2,767,797 & $27.614$ \\
\hline
\rlo & 1,727,127 & $44.253$ \\
\hline
\optimal & 1,703,234  & $44.873$ \\
\hline
\end{tabular}
\end{adjustbox}
\newline
\vspace*{0.4 cm}
\newline
\begin{adjustbox}{max width=140mm}
\setlength{\tabcolsep}{5pt}
\renewcommand{\arraystretch}{1.6}
\begin{tabular}{|r||rrrr|}
  \hline
  \multirow{2}{*}{{\diagbox[width=4.7cm, height=1.18cm]{\rlap{\enspace\raisebox{0ex}{ \it \hspace{-3.0mm} norm.\ Hamming d.\  }}}{\raisebox{-0ex}{\it \hspace{-3.0mm} Hamming d.\ }}}} & \multicolumn{4}{l|}{\it Hamming distance on a subset of 5,000 sequences} \\
  \cline{2-5}
  & \multicolumn{1}{R{1.6cm}|}{\dollarebwt} & \multicolumn{1}{R{1.6cm}|}{\mdollar} & \multicolumn{1}{R{1.6cm}|}{\concat} & \multicolumn{1}{R{1.6cm}|}{\rlo} \\
  \hline \hline
 \multicolumn{1}{|l||}{\dollarebwt}& \multicolumn{1}{l|}{\hspace{5.0mm} 0} & \multicolumn{1}{r|}{20,159} & \multicolumn{1}{r|}{23,229} & \multicolumn{1}{r|}{15,835} \\
 \hline
 \multicolumn{1}{|l||}{\mdollar} & \multicolumn{1}{r|}{0.02635} & \multicolumn{1}{l|}{\hspace{5.0mm} 0}  & \multicolumn{1}{r|}{20,024} & \multicolumn{1}{r|}{20,092} \\ 
 \hline
 \multicolumn{1}{|l||}{\concat} & \multicolumn{1}{r|}{0.03036} & \multicolumn{1}{r|}{0.02617} & \multicolumn{1}{l|}{\hspace{5.0mm} 0} & \multicolumn{1}{r|}{25,464} \\ 
 \hline
 \multicolumn{1}{|l||}{\rlo} & \multicolumn{1}{r|}{0.02070} & \multicolumn{1}{r|}{0.02626} & \multicolumn{1}{r|}{0.03329} & \multicolumn{1}{l|}{\hspace{5.0mm} 0} \\ 
 \hline
\end{tabular}
\quad
\begin{tabular}{|l||r|}
\hline
\multicolumn{2}{|l|}{\it small dataset properties} \\
\hline \hline
no.\ of sequences & 5,000 \\
\hline
total length & 765,037 \\
\hline
average length & 152 \\
\hline
no.\ of interesting intervals & 1,376 \\
\hline
total length intr.int.s & 139,041 \\
\hline
fraction pos.s in intr.int.s  & 0.182 \\
\hline
variability & 0.222 \\
\hline
\end{tabular}
\end{adjustbox}
\newline
\vspace*{0.4 cm}
\newline
\begin{adjustbox}{max width=140mm}
\setlength{\tabcolsep}{5pt}
\renewcommand{\arraystretch}{1.6}
\begin{tabular}{|r||r|r|r|r|r|}
  \hline
  \multirow{2}{*}{{\diagbox[width=4.0cm, height=1.18cm]{\rlap{\enspace\raisebox{0ex}{ \it \hspace{-3.0mm} norm.\ edit d.\  }}}{\raisebox{+0ex}{\it edit d.\ }} }} & \multicolumn{5}{l|}{\it edit distance on a subset of 5,000 sequences} \\
  \cline{2-6}
  & \multicolumn{1}{R{1.6cm}|}{\eBWT} &\multicolumn{1}{R{1.6cm}|}{\dollarebwt} & \multicolumn{1}{R{1.6cm}|}{\mdollar} & \multicolumn{1}{R{1.6cm}|}{\concat} & \multicolumn{1}{R{1.6cm}|}{\rlo} \\
  \hline \hline
 \multicolumn{1}{|l||}{\eBWT} & \multicolumn{1}{l|}{\hspace{5.0mm} 0} & \multicolumn{1}{r|}{51,683} & \multicolumn{1}{r|}{62,799} & \multicolumn{1}{r|}{63,303} & \multicolumn{1}{r|}{61,732} \\ 
 \hline
 \multicolumn{1}{|l||}{\dollarebwt}& \multicolumn{1}{r|}{0.06756} & \multicolumn{1}{l|}{\hspace{5.0mm} 0} & \multicolumn{1}{r|}{16,968} & \multicolumn{1}{r|}{20,180} & \multicolumn{1}{r|}{14,166}  \\ 
 \hline
 \multicolumn{1}{|l||}{\mdollar} & \multicolumn{1}{r|}{0.08209} & \multicolumn{1}{r|}{0.02218} & \multicolumn{1}{l|}{\hspace{5.0mm} 0}  & \multicolumn{1}{r|}{16,695} & \multicolumn{1}{r|}{19,371} \\ \cline{1-6}
 \multicolumn{1}{|l||}{\concat} & \multicolumn{1}{r|}{0.08274} & \multicolumn{1}{r|}{0.02638} & \multicolumn{1}{r|}{0.02182} & \multicolumn{1}{l|}{\hspace{5.0mm} 0} & \multicolumn{1}{r|}{21,683}  \\ 
 \hline
 \multicolumn{1}{|l||}{\colex} & \multicolumn{1}{r|}{0.08069} & \multicolumn{1}{r|}{0.01852} & \multicolumn{1}{r|}{0.02532} & \multicolumn{1}{r|}{0.02834} & \multicolumn{1}{l|}{\hspace{5.0mm} 0} \\ 
 \hline
\end{tabular}
\quad
\begin{tabular}{|l||r|r|}
\hline
\multicolumn{3}{|l|}{\it no.\ runs small dataset} \\
\hline \hline
& $r$ & $n/r$ \\
\hline
\eBWT & 35,262 & $21.554$ \\
\hline
\dollarebwt & 35,293 & $21.677$ \\
\hline
\mdollar & 50,581 & $15.125$ \\
\hline
\concat & 38,900 & $19.667$ \\ 
\hline
\rlo & 30,568 & $25.027$ \\
\hline
\optimal & 30,007 & $25.495$ \\
\hline
\end{tabular}
\end{adjustbox}
\vspace{2mm}
\caption{\label{tab:16S_short_complete} Results for the 16S rRNA short dataset. First row left: absolute and normalized pairwise Hamming distance between separator-based BWT variants. First row right: summary of the dataset properties. Second row: number of runs and average runlength ($n/r$) of all BWT variants. Third row left: absolute and normalized pairwise Hamming distance between separator-based BWT variants on a subset of the input collection. Third row right: summary of the dataset properties of a subset of the input collection. Fourth row left: absolute and normalized pairwise edit distance between all BWT variants on a subset of the input collection. Fourth row right: number of runs and average runlength ($n/r$) of all BWT variants on a subset of the input collection.}
\end{table*}

\begin{table*}[t!]
\raggedright
\textbf{Influenza A reads (500,000 short sequences)}

\vspace{4mm}

\captionsetup{width=\linewidth}
\begin{adjustbox}{max width=140mm}
\setlength{\tabcolsep}{5pt}
\renewcommand{\arraystretch}{1.6}
\begin{tabular}{|r||rrrr|}
  \hline
  \multirow{2}{*}{{\diagbox[width=4.7cm, height=1.18cm]{\rlap{\enspace\raisebox{0ex}{ \it \hspace{-3.0mm} norm.\ Hamming d.\  }}}{\raisebox{-0ex}{\it \hspace{-3.0mm} Hamming d.\ }}}} & \multicolumn{4}{l|}{\it Hamming distance on the big dataset} \\
  \cline{2-5}
  & \multicolumn{1}{R{1.6cm}|}{\dollarebwt} & \multicolumn{1}{R{1.6cm}|}{\mdollar} & \multicolumn{1}{R{1.6cm}|}{\concat} & \multicolumn{1}{R{1.6cm}|}{\rlo} \\
  \hline \hline
 \multicolumn{1}{|l||}{\dollarebwt}& \multicolumn{1}{l|}{\hspace{5.0mm} 0} & \multicolumn{1}{r|}{3,040,590} & \multicolumn{1}{r|}{3,038,509} & \multicolumn{1}{r|}{2,938,706} \\
 \hline
 \multicolumn{1}{|l||}{\mdollar} & \multicolumn{1}{r|}{0.02617} & \multicolumn{1}{l|}{\hspace{5.0mm} 0}  & \multicolumn{1}{r|}{3,039,095} & \multicolumn{1}{r|}{3,041,816} \\ 
 \hline
 \multicolumn{1}{|l||}{\concat} & \multicolumn{1}{r|}{0.02615} & \multicolumn{1}{r|}{0.02616} & \multicolumn{1}{l|}{\hspace{5.0mm} 0} & \multicolumn{1}{r|}{3,089,670} \\ 
 \hline
 \multicolumn{1}{|l||}{\rlo} & \multicolumn{1}{r|}{0.02529} & \multicolumn{1}{r|}{0.02618} & \multicolumn{1}{r|}{0.02659} & \multicolumn{1}{l|}{\hspace{5.0mm} 0} \\ 
 \hline
\end{tabular}
\quad
\begin{tabular}{|l||r|}
\hline
\multicolumn{2}{|l|}{\it dataset properties} \\
\hline \hline
no.\ of sequences & 500,000 \\
\hline
total length & 116,192,842 \\
\hline
average length & 231 \\
\hline
no.\ of interesting intervals  & 213,735 \\
\hline
total length intr.int.s & 11,995,246 \\
\hline
fraction pos.s in intr.int.s  & 0.103 \\
\hline
variability & 0.363 \\
\hline
\end{tabular}
\end{adjustbox}
\newline
\vspace*{0.4 cm}
\newline
\begin{adjustbox}{max width=140mm}
\setlength{\tabcolsep}{5pt}
\renewcommand{\arraystretch}{1.6}
\hspace{14cm}
\begin{tabular}{|l||r|r|}
\hline
\multicolumn{3}{|l|}{\it no.\ runs big dataset} \\
\hline \hline
& $r$ & $n/r$ \\
\hline
\eBWT & 3,258,605 & $35.504$ \\
\hline
\dollarebwt & 3,298,502 & $35.226$ \\
\hline
\mdollar & 5,030,032 & $23.100$ \\
\hline
\concat & 4,629,150 & $25.100$ \\
\hline
\colex & 2,362,987 & $49.172$  \\
\hline
\optimal & 2,311,133 & $50.275$ \\
\hline
\end{tabular}
\end{adjustbox}
\newline
\vspace*{0.4 cm}
\newline
\begin{adjustbox}{max width=140mm}
\setlength{\tabcolsep}{5pt}
\renewcommand{\arraystretch}{1.6}
\begin{tabular}{|r||rrrr|}
  \hline
  \multirow{2}{*}{{\diagbox[width=4.7cm, height=1.18cm]{\rlap{\enspace\raisebox{0ex}{ \it \hspace{-3.0mm} norm.\ Hamming d.\  }}}{\raisebox{-0ex}{\it \hspace{-3.0mm} Hamming d.\ }}}} & \multicolumn{4}{l|}{\it Hamming distance on a subset of 5,000 sequences} \\
  \cline{2-5}
  & \multicolumn{1}{R{1.6cm}|}{\dollarebwt} & \multicolumn{1}{R{1.6cm}|}{\mdollar} & \multicolumn{1}{R{1.6cm}|}{\concat} & \multicolumn{1}{R{1.6cm}|}{\rlo} \\
  \hline \hline
 \multicolumn{1}{|l||}{\dollarebwt}& \multicolumn{1}{l|}{\hspace{5.0mm} 0} & \multicolumn{1}{r|}{23,456} & \multicolumn{1}{r|}{23,456} & \multicolumn{1}{r|}{22,873} \\
 \hline
 \multicolumn{1}{|l||}{\mdollar} & \multicolumn{1}{r|}{0.02018} & \multicolumn{1}{l|}{\hspace{5.0mm} 0}  & \multicolumn{1}{r|}{23,509} & \multicolumn{1}{r|}{23,407} \\ 
 \hline
 \multicolumn{1}{|l||}{\concat} & \multicolumn{1}{r|}{0.02018} & \multicolumn{1}{r|}{0.02023} & \multicolumn{1}{l|}{\hspace{5.0mm} 0} & \multicolumn{1}{r|}{24,061} \\ 
 \hline
 \multicolumn{1}{|l||}{\rlo} & \multicolumn{1}{r|}{0.01968} & \multicolumn{1}{r|}{0.02014} & \multicolumn{1}{r|}{0.02070} & \multicolumn{1}{l|}{\hspace{5.0mm} 0} \\ 
 \hline
\end{tabular}
\quad
\begin{tabular}{|l||r|}
\hline
\multicolumn{2}{|l|}{\it small dataset properties} \\
\hline \hline
no.\ of sequences & 5,000 \\
\hline
total length & 1,162,319 \\
\hline
average length & 231 \\
\hline
no.\ of interesting intervals & 3,062 \\
\hline
total length intr.int.s & 36,019 \\
\hline
fraction pos.s in intr.int.s  & 0.031 \\
\hline
variability & 0.966 \\
\hline
\end{tabular}
\end{adjustbox}
\newline
\vspace*{0.4 cm}
\newline
\begin{adjustbox}{max width=140mm}
\setlength{\tabcolsep}{5pt}
\renewcommand{\arraystretch}{1.6}
\begin{tabular}{|r||r|r|r|r|r|}
  \hline
  \multirow{2}{*}{{\diagbox[width=4.0cm, height=1.18cm]{\rlap{\enspace\raisebox{0ex}{ \it \hspace{-3.0mm} norm.\ edit d.\  }}}{\raisebox{+0ex}{\it edit d.\ }} }} & \multicolumn{5}{l|}{\it edit distance on a subset of 5,000 sequences} \\
  \cline{2-6}
  & \multicolumn{1}{R{1.6cm}|}{\eBWT} &\multicolumn{1}{R{1.6cm}|}{\dollarebwt} & \multicolumn{1}{R{1.6cm}|}{\mdollar} & \multicolumn{1}{R{1.6cm}|}{\concat} & \multicolumn{1}{R{1.6cm}|}{\rlo} \\
  \hline \hline
 \multicolumn{1}{|l||}{\eBWT} & \multicolumn{1}{l|}{\hspace{5.0mm} 0} & \multicolumn{1}{r|}{75,966} & \multicolumn{1}{r|}{75,935} & \multicolumn{1}{r|}{75,991} & \multicolumn{1}{r|}{76,437} \\ 
 \hline
 \multicolumn{1}{|l||}{\dollarebwt}& \multicolumn{1}{r|}{0.06536} & \multicolumn{1}{l|}{\hspace{5.0mm} 0} & \multicolumn{1}{r|}{18,043} & \multicolumn{1}{r|}{18,316} & \multicolumn{1}{r|}{21,869}  \\ 
 \hline
 \multicolumn{1}{|l||}{\mdollar} & \multicolumn{1}{r|}{0.06533} & \multicolumn{1}{r|}{0.01552} & \multicolumn{1}{l|}{\hspace{5.0mm} 0}  & \multicolumn{1}{r|}{17,835} & \multicolumn{1}{r|}{22,536} \\ \cline{1-6}
 \multicolumn{1}{|l||}{\concat} & \multicolumn{1}{r|}{0.06538} & \multicolumn{1}{r|}{0.01576} & \multicolumn{1}{r|}{0.01534} & \multicolumn{1}{l|}{\hspace{5.0mm} 0} & \multicolumn{1}{r|}{23,078}  \\ 
 \hline
 \multicolumn{1}{|l||}{\rlo} & \multicolumn{1}{r|}{0.06576} & \multicolumn{1}{r|}{0.01881} & \multicolumn{1}{r|}{0.01939} & \multicolumn{1}{r|}{0.01986} & \multicolumn{1}{l|}{\hspace{5.0mm} 0} \\ 
 \hline
\end{tabular}
\quad
\begin{tabular}{|l||r|r|}
\hline
\multicolumn{3}{|l|}{\it no.\ runs small dataset} \\
\hline \hline
& $r$ & $n/r$ \\
\hline
\eBWT & 81,992 & $14.115$ \\
\hline
\dollarebwt & 85,489 & $13.596$ \\
\hline
\mdollar & 89,256 & $13.022$ \\
\hline
\concat & 87,867 & $13.228$ \\ 
\hline
\rlo & 70,534 & $16.479$ \\
\hline
\optimal & 68,900 & $16.870$ \\
\hline
\end{tabular}
\end{adjustbox}
\vspace{2mm}
\caption{\label{tab:Influenza_complete} Results for the Influenza A reads dataset. First row left: absolute and normalized pairwise Hamming distance between separator-based BWT variants. First row right: summary of the dataset properties. Second row: number of runs and average runlength ($n/r$) of all BWT variants. Third row left: absolute and normalized pairwise Hamming distance between separator-based BWT variants on a subset of the input collection. Third row right: summary of the dataset properties of a subset of the input collection. Fourth row left: absolute and normalized pairwise edit distance between all BWT variants on a subset of the input collection. Fourth row right: number of runs and average runlength ($n/r$) of all BWT variants on a subset of the input collection.}
\end{table*}

\begin{table*}[t!]
\raggedright
\textbf{SARS-CoV-2 long (50,000 long sequences)}

\vspace{4mm}

\captionsetup{width=\linewidth}
\begin{adjustbox}{max width=140mm}
\setlength{\tabcolsep}{5pt}
\renewcommand{\arraystretch}{1.6}
\begin{tabular}{|r||rrrr|}
  \hline
  \multirow{2}{*}{{\diagbox[width=4.7cm, height=1.18cm]{\rlap{\enspace\raisebox{0ex}{ \it \hspace{-3.0mm} norm.\ Hamming d.\  }}}{\raisebox{-0ex}{\it \hspace{-3.0mm} Hamming d.\ }}}} & \multicolumn{4}{l|}{\it Hamming distance on the big dataset} \\
  \cline{2-5}
  & \multicolumn{1}{R{1.6cm}|}{\dollarebwt} & \multicolumn{1}{R{1.6cm}|}{\mdollar} & \multicolumn{1}{R{1.6cm}|}{\concat} & \multicolumn{1}{R{1.6cm}|}{\rlo} \\
  \hline \hline
 \multicolumn{1}{|l||}{\dollarebwt}& \multicolumn{1}{l|}{\hspace{5.0mm} 0} & \multicolumn{1}{r|}{248,189} & \multicolumn{1}{r|}{248,205} & \multicolumn{1}{r|}{255,357} \\
 \hline
 \multicolumn{1}{|l||}{\mdollar} & \multicolumn{1}{r|}{0.00462} & \multicolumn{1}{l|}{\hspace{5.0mm} 0}  & \multicolumn{1}{r|}{248,572} & \multicolumn{1}{r|}{248,631} \\ 
 \hline
 \multicolumn{1}{|l||}{\concat} & \multicolumn{1}{r|}{0.00462} & \multicolumn{1}{r|}{0.00462} & \multicolumn{1}{l|}{\hspace{5.0mm} 0} & \multicolumn{1}{r|}{248,765} \\ 
 \hline
 \multicolumn{1}{|l||}{\rlo} & \multicolumn{1}{r|}{0.00475} & \multicolumn{1}{r|}{0.00462} & \multicolumn{1}{r|}{0.00463} & \multicolumn{1}{l|}{\hspace{5.0mm} 0} \\ 
 \hline
\end{tabular}
\quad
\begin{tabular}{|l||r|}
\hline
\multicolumn{2}{|l|}{\it dataset properties} \\
\hline \hline
no.\ sequences & 50,000 \\
\hline
total length & 53,776,351 \\
\hline
average length & 1,075 \\
\hline
no.\ of interesting intervals & 31,931 \\
\hline
total length intr.int.s & 9,436,894 \\
\hline
fraction pos.s in intr.int.s  & 0.17548 \\
\hline
variability & 0.03716 \\
\hline
\end{tabular}
\end{adjustbox}
\newline
\vspace*{0.3 cm}
\newline
\begin{adjustbox}{max width=140mm}
\setlength{\tabcolsep}{5pt}
\renewcommand{\arraystretch}{1.6}
\hspace{14cm}
\begin{tabular}{|l||r|r|}
\hline
\multicolumn{3}{|l|}{\it no.\ runs big dataset} \\
\hline \hline
& $r$ & $n/r$ \\
\hline
\eBWT & 882,634 & $60.870$ \\
\hline
\dollarebwt & 879,608 & $61.137$ \\
\hline
\mdollar & 934,129 & $57.568$ \\
\hline
\concat & 934,117 & $57.569$ \\
\hline
\colex & 734,610 & $73.204$ \\  
\hline
\optimal & 721,845 & $74.498$ \\
\hline
\end{tabular}
\end{adjustbox}
\newline
\vspace*{0.4 cm}
\newline
\begin{adjustbox}{max width=140mm}
\setlength{\tabcolsep}{5pt}
\renewcommand{\arraystretch}{1.6}
\begin{tabular}{|r||rrrr|}
  \hline
  \multirow{2}{*}{{\diagbox[width=4.7cm, height=1.18cm]{\rlap{\enspace\raisebox{0ex}{ \it \hspace{-3.0mm} norm.\ Hamming d.\  }}}{\raisebox{-0ex}{\it \hspace{-3.0mm} Hamming d.\ }}}} & \multicolumn{4}{l|}{\it Hamming distance on a subset of 1,500 sequences} \\
  \cline{2-5}
  & \multicolumn{1}{R{1.6cm}|}{\dollarebwt} & \multicolumn{1}{R{1.6cm}|}{\mdollar} & \multicolumn{1}{R{1.6cm}|}{\concat} & \multicolumn{1}{R{1.6cm}|}{\rlo} \\
  \hline \hline
 \multicolumn{1}{|l||}{\dollarebwt}& \multicolumn{1}{l|}{\hspace{5.0mm} 0} & \multicolumn{1}{r|}{4,936} & \multicolumn{1}{r|}{4,939} & \multicolumn{1}{r|}{5,296} \\
 \hline
 \multicolumn{1}{|l||}{\mdollar} & \multicolumn{1}{r|}{0.00306} & \multicolumn{1}{l|}{\hspace{5.0mm} 0}  & \multicolumn{1}{r|}{4,884} & \multicolumn{1}{r|}{4,908} \\ 
 \hline
 \multicolumn{1}{|l||}{\concat} & \multicolumn{1}{r|}{0.00306} & \multicolumn{1}{r|}{0.00303} & \multicolumn{1}{l|}{\hspace{5.0mm} 0} & \multicolumn{1}{r|}{5,012} \\ 
 \hline
 \multicolumn{1}{|l||}{\rlo} & \multicolumn{1}{r|}{0.00328} & \multicolumn{1}{r|}{0.00304} & \multicolumn{1}{r|}{0.00310} & \multicolumn{1}{l|}{\hspace{5.0mm} 0} \\ 
 \hline
\end{tabular}
\quad
\begin{tabular}{|l||r|}
\hline
\multicolumn{2}{|l|}{\it small dataset properties} \\
\hline \hline
no.\ sequences & 1,500 \\
\hline
total length & 1,612,956 \\
\hline
average length & 1,075 \\
\hline
no.\ of interesting intervals & 1,046 \\
\hline
total length intr.int.s & 152,035 \\
\hline
fraction pos.s in intr.int.s  & 0.094 \\
\hline
variability & 0.047 \\
\hline
\end{tabular}
\end{adjustbox}
\newline
\vspace*{0.4 cm}
\newline
\begin{adjustbox}{max width=140mm}
\setlength{\tabcolsep}{5pt}
\renewcommand{\arraystretch}{1.6}
\begin{tabular}{|r||r|r|r|r|r|}
  \hline
  \multirow{2}{*}{{\diagbox[width=4.0cm, height=1.18cm]{\rlap{\enspace\raisebox{0ex}{ \it \hspace{-3.0mm} norm.\ edit d.\  }}}{\raisebox{+0ex}{\it edit d.\ }} }} & \multicolumn{5}{l|}{\it edit distance on a subset of 1,500 sequences} \\
  \cline{2-6}
  & \multicolumn{1}{R{1.6cm}|}{\eBWT} &\multicolumn{1}{R{1.6cm}|}{\dollarebwt} & \multicolumn{1}{R{1.6cm}|}{\mdollar} & \multicolumn{1}{R{1.6cm}|}{\concat} & \multicolumn{1}{R{1.6cm}|}{\rlo} \\
  \hline \hline
 \multicolumn{1}{|l||}{\eBWT} & \multicolumn{1}{l|}{\hspace{5.0mm} 0} & \multicolumn{1}{r|}{19,140} & \multicolumn{1}{r|}{21,809} & \multicolumn{1}{r|}{21,796} & \multicolumn{1}{r|}{22,618} \\ 
 \hline
 \multicolumn{1}{|l||}{\dollarebwt}& \multicolumn{1}{r|}{0.01186} & \multicolumn{1}{l|}{\hspace{5.0mm} 0} & \multicolumn{1}{r|}{4,345} & \multicolumn{1}{r|}{4,322} & \multicolumn{1}{r|}{5,186}  \\ 
 \hline
 \multicolumn{1}{|l||}{\mdollar} & \multicolumn{1}{r|}{0.01351} & \multicolumn{1}{r|}{0.00269} & \multicolumn{1}{l|}{\hspace{5.0mm} 0}  & \multicolumn{1}{r|}{4,110} & \multicolumn{1}{r|}{4,820} \\ \cline{1-6}
 \multicolumn{1}{|l||}{\concat} & \multicolumn{1}{r|}{0.01350} & \multicolumn{1}{r|}{0.00306} & \multicolumn{1}{r|}{0.00255} & \multicolumn{1}{l|}{\hspace{5.0mm} 0} & \multicolumn{1}{r|}{4,893}  \\ 
 \hline
 \multicolumn{1}{|l||}{\rlo} & \multicolumn{1}{r|}{0.01401} & \multicolumn{1}{r|}{0.00321} & \multicolumn{1}{r|}{0.00299} & \multicolumn{1}{r|}{0.00303} & \multicolumn{1}{l|}{\hspace{5.0mm} 0} \\ 
 \hline
\end{tabular}
\quad
\begin{tabular}{|l||r|r|}
\hline
\multicolumn{3}{|l|}{\it no.\ runs small dataset} \\
\hline \hline
& $r$ & $n/r$ \\
\hline
\eBWT & 45,262 & $35.636$ \\
\hline
\dollarebwt & 45,155 & $35.754$ \\
\hline
\mdollar & 45,572 & $35.426$ \\
\hline
\concat & 45,644 & $35.371$ \\
\hline
\rlo & 42,516 & $37.973$ \\ 
\hline
\optimal & 42,093 & $38.355$ \\
\hline
\end{tabular}
\end{adjustbox}
\vspace{2mm}
\caption{\label{tab:SARS_long_complete} Results for the SARS-CoV-2 long dataset. First row left: absolute and normalized pairwise Hamming distance between separator-based BWT variants. First row right: summary of the dataset properties. Second row: number of runs and average runlength ($n/r$) of all BWT variants. Third row left: absolute and normalized pairwise Hamming distance between separator-based BWT variants on a subset of the input collection. Third row right: summary of the dataset properties of a subset of the input collection. Fourth row left: absolute and normalized pairwise edit distance between all BWT variants on a subset of the input collection. Fourth row right: number of runs and average runlength ($n/r$) of all BWT variants on a subset of the input collection.}
\end{table*}

\begin{table*}[t!]
\raggedright
\textbf{16S rRNA long (16,741 long sequences)}

\vspace{4mm}

\captionsetup{width=\linewidth}
\begin{adjustbox}{max width=140mm}
\setlength{\tabcolsep}{5pt}
\renewcommand{\arraystretch}{1.6}
\begin{tabular}{|r||rrrr|}
  \hline
  \multirow{2}{*}{{\diagbox[width=4.7cm, height=1.18cm]{\rlap{\enspace\raisebox{0ex}{ \it \hspace{-3.0mm} norm.\ Hamming d.\  }}}{\raisebox{-0ex}{\it \hspace{-3.0mm} Hamming d.\ }}}} & \multicolumn{4}{l|}{\it Hamming distance on the big dataset} \\
  \cline{2-5}
  & \multicolumn{1}{R{1.6cm}|}{\dollarebwt} & \multicolumn{1}{R{1.6cm}|}{\mdollar} & \multicolumn{1}{R{1.6cm}|}{\concat} & \multicolumn{1}{R{1.6cm}|}{\rlo} \\
  \hline \hline
 \multicolumn{1}{|l||}{\dollarebwt}& \multicolumn{1}{l|}{\hspace{5.0mm} 0} & \multicolumn{1}{r|}{85,960} & \multicolumn{1}{r|}{42,948} &  \multicolumn{1}{r|}{67,103} \\
 \hline
 \multicolumn{1}{|l||}{\mdollar} & \multicolumn{1}{r|}{0.00342} & \multicolumn{1}{l|}{\hspace{5.0mm} 0}  & \multicolumn{1}{r|}{85,961} & \multicolumn{1}{r|}{82,890} \\ 
 \hline
 \multicolumn{1}{|l||}{\concat} & \multicolumn{1}{r|}{0.00171} & \multicolumn{1}{r|}{0.00342} & \multicolumn{1}{l|}{\hspace{5.0mm} 0} & \multicolumn{1}{r|}{71,264} \\ 
 \hline
 \multicolumn{1}{|l||}{\rlo} & \multicolumn{1}{r|}{0.00267} & \multicolumn{1}{r|}{0.00329} & \multicolumn{1}{r|}{0.00283} & \multicolumn{1}{l|}{\hspace{5.0mm} 0} \\ 
 \hline
\end{tabular}
\quad
\begin{tabular}{|l||r|}
\hline
\multicolumn{2}{|l|}{\it dataset properties} \\
\hline \hline
no.\ sequences & 16,741 \\
\hline
total length & 25,159,064 \\
\hline
average length & 1,501 \\
\hline
no.\ of interesting intervals & 9,918 \\
\hline
total length intr.int.s & 1,173,284 \\
\hline
fraction pos.s in intr.int.s  & 0.047 \\
\hline
variability & 0.104 \\
\hline
\end{tabular}
\end{adjustbox}
\newline
\vspace*{0.4 cm}
\newline
\begin{adjustbox}{max width=140mm}
\setlength{\tabcolsep}{5pt}
\renewcommand{\arraystretch}{1.6}
\hspace{14cm}
\begin{tabular}{|l||r|r|}
\hline
\multicolumn{3}{|l|}{\it no.\ runs big dataset} \\
\hline \hline
& $r$ & $n/r$ \\
\hline
\eBWT & 547,991 & $45.881$ \\
\hline
\dollarebwt & 547,793 & $45.928$ \\
\hline
\mdollar & 555,687 & $45.276$ \\
\hline
\concat & 558,902 & $45.015$  \\
\hline
\rlo & 536,682 & $46.879$ \\
\hline
\optimal & 533,712 & $47.140$ \\
\hline
\end{tabular}
\end{adjustbox}
\newline
\vspace*{0.4 cm}
\newline
\begin{adjustbox}{max width=140mm}
\setlength{\tabcolsep}{5pt}
\renewcommand{\arraystretch}{1.6}
\begin{tabular}{|r||rrrr|}
  \hline
  \multirow{2}{*}{{\diagbox[width=4.7cm, height=1.18cm]{\rlap{\enspace\raisebox{0ex}{ \it \hspace{-3.0mm} norm.\ Hamming d.\  }}}{\raisebox{-0ex}{\it \hspace{-3.0mm} Hamming d.\ }}}} & \multicolumn{4}{l|}{\it Hamming distance on a subset of 1,500 sequences} \\
  \cline{2-5}
  & \multicolumn{1}{R{1.6cm}|}{\dollarebwt} & \multicolumn{1}{R{1.6cm}|}{\mdollar} & \multicolumn{1}{R{1.6cm}|}{\concat} & \multicolumn{1}{R{1.6cm}|}{\rlo} \\
  \hline \hline
 \multicolumn{1}{|l||}{\dollarebwt}& \multicolumn{1}{l|}{\hspace{5.0mm} 0} & \multicolumn{1}{r|}{4,740} & \multicolumn{1}{r|}{3,104} & \multicolumn{1}{r|}{3,926} \\
 \hline
 \multicolumn{1}{|l||}{\mdollar} & \multicolumn{1}{r|}{0.00210} & \multicolumn{1}{l|}{\hspace{5.0mm} 0}  & \multicolumn{1}{r|}{4,716} & \multicolumn{1}{r|}{4,783} \\ 
 \hline
 \multicolumn{1}{|l||}{\concat} & \multicolumn{1}{r|}{0.00137} & \multicolumn{1}{r|}{0.00209} & \multicolumn{1}{l|}{\hspace{5.0mm} 0} & \multicolumn{1}{r|}{4,208} \\ 
 \hline
 \multicolumn{1}{|l||}{\rlo} & \multicolumn{1}{r|}{0.00174} & \multicolumn{1}{r|}{0.00212} & \multicolumn{1}{r|}{0.00186} & \multicolumn{1}{l|}{\hspace{5.0mm} 0} \\ 
 \hline
\end{tabular}
\quad
\begin{tabular}{|l||r|}
\hline
\multicolumn{2}{|l|}{\it small dataset properties} \\
\hline \hline
no.\ sequences & 1,500 \\
\hline
total length & 2,260,229 \\
\hline
average length & 1,501 \\
\hline
no.\ of interesting intervals & 946 \\
\hline
total length intr.int.s & 72,933 \\
\hline
fraction pos.s in intr.int.s  & 0.032 \\
\hline
variability & 0.104 \\
\hline
\end{tabular}
\end{adjustbox}
\newline
\vspace*{0.4 cm}
\newline
\begin{adjustbox}{max width=140mm}
\setlength{\tabcolsep}{5pt}
\renewcommand{\arraystretch}{1.6}
\begin{tabular}{|r||r|r|r|r|r|}
  \hline
  \multirow{2}{*}{{\diagbox[width=4.0cm, height=1.18cm]{\rlap{\enspace\raisebox{0ex}{ \it \hspace{-3.0mm} norm.\ edit d.\  }}}{\raisebox{+0ex}{\it edit d.\ }} }} & \multicolumn{5}{l|}{\it edit distance on a subset of 1,500 sequences} \\
  \cline{2-6}
  & \multicolumn{1}{R{1.6cm}|}{\eBWT} &\multicolumn{1}{R{1.6cm}|}{\dollarebwt} & \multicolumn{1}{R{1.6cm}|}{\mdollar} & \multicolumn{1}{R{1.6cm}|}{\concat} & \multicolumn{1}{R{1.6cm}|}{\rlo} \\
  \hline \hline
 \multicolumn{1}{|l||}{\eBWT} & \multicolumn{1}{l|}{\hspace{5.0mm} 0} & \multicolumn{1}{r|}{18,328} & \multicolumn{1}{r|}{22,194} & \multicolumn{1}{r|}{21,021} & \multicolumn{1}{r|}{21,987} \\ 
 \hline
 \multicolumn{1}{|l||}{\dollarebwt}& \multicolumn{1}{r|}{0.00811} & \multicolumn{1}{l|}{\hspace{5.0mm} 0} & \multicolumn{1}{r|}{4,410} & \multicolumn{1}{r|}{2,761} & \multicolumn{1}{r|}{3,858}  \\ 
 \hline
 \multicolumn{1}{|l||}{\mdollar} & \multicolumn{1}{r|}{0.00982} & \multicolumn{1}{r|}{0.00195} & \multicolumn{1}{l|}{\hspace{5.0mm} 0}  & \multicolumn{1}{r|}{4,323} & \multicolumn{1}{r|}{4,691} \\ \cline{1-6}
 \multicolumn{1}{|l||}{\concat} & \multicolumn{1}{r|}{0.00930} & \multicolumn{1}{r|}{0.00122} & \multicolumn{1}{r|}{0.00191} & \multicolumn{1}{l|}{\hspace{5.0mm} 0} & \multicolumn{1}{r|}{4,146}  \\ 
 \hline
 \multicolumn{1}{|l||}{\rlo} & \multicolumn{1}{r|}{0.00973} & \multicolumn{1}{r|}{0.00171} & \multicolumn{1}{r|}{0.00208} & \multicolumn{1}{r|}{0.00183} & \multicolumn{1}{l|}{\hspace{5.0mm} 0} \\ 
 \hline
\end{tabular}
\quad
\begin{tabular}{|l||r|r|}
\hline
\multicolumn{3}{|l|}{\it no.\ runs small dataset} \\
\hline \hline
& $r$ & $n/r$ \\
\hline
\eBWT & 62,077 & $36.386$ \\
\hline
\dollarebwt & 62,031 & $36.437$ \\
\hline
\mdollar & 62,712 & $36.041$ \\
\hline
\concat & 62,800 & $35.991$  \\
\hline
\rlo & 61,235 & $36.911$ \\ 
\hline
\optimal & 60,979 & $37.066$ \\
\hline
\end{tabular}
\end{adjustbox}
\vspace{2mm}
\caption{\label{tab:16S_long_complete} Results for the 16S rRNA long dataset. First row left: absolute and normalized pairwise Hamming distance between separator-based BWT variants. First row right: summary of the dataset properties. Second row: number of runs and average runlength ($n/r$) of all BWT variants. Third row left: absolute and normalized pairwise Hamming distance between separator-based BWT variants on a subset of the input collection. Third row right: summary of the dataset properties of a subset of the input collection. Fourth row left: absolute and normalized pairwise edit distance between all BWT variants on a subset of the input collection. Fourth row right: number of runs and average runlength ($n/r$) of all BWT variants on a subset of the input collection.}
\end{table*}

\begin{table*}[t!]
\raggedright
\textbf{Candida auris reads (50,000 long sequences)}

\vspace{4mm}

\captionsetup{width=\linewidth}
\begin{adjustbox}{max width=140mm}
\setlength{\tabcolsep}{5pt}
\renewcommand{\arraystretch}{1.6}
\begin{tabular}{|r||rrrr|}
  \hline
  \multirow{2}{*}{{\diagbox[width=4.7cm, height=1.18cm]{\rlap{\enspace\raisebox{0ex}{ \it \hspace{-3.0mm} norm.\ Hamming d.\  }}}{\raisebox{-0ex}{\it \hspace{-3.0mm} Hamming d.\ }}}} & \multicolumn{4}{l|}{\it Hamming distance on the big dataset} \\
  \cline{2-5}
  & \multicolumn{1}{R{1.6cm}|}{\dollarebwt} & \multicolumn{1}{R{1.6cm}|}{\mdollar} & \multicolumn{1}{R{1.6cm}|}{\concat} & \multicolumn{1}{R{1.6cm}|}{\rlo} \\
  \hline \hline
 \multicolumn{1}{|l||}{\dollarebwt}& \multicolumn{1}{l|}{\hspace{5.0mm} 0} & \multicolumn{1}{r|}{306,071} & \multicolumn{1}{r|}{306,431} & \multicolumn{1}{r|}{305,665} \\
 \hline
 \multicolumn{1}{|l||}{\mdollar} & \multicolumn{1}{r|}{0.00246} & \multicolumn{1}{l|}{\hspace{5.0mm} 0}  & \multicolumn{1}{r|}{305,649} & \multicolumn{1}{r|}{305,713} \\ 
 \hline
 \multicolumn{1}{|l||}{\concat} & \multicolumn{1}{r|}{0.00247} & \multicolumn{1}{r|}{0.00246} & \multicolumn{1}{l|}{\hspace{5.0mm} 0} & \multicolumn{1}{r|}{305,469} \\ 
 \hline
 \multicolumn{1}{|l||}{\rlo} & \multicolumn{1}{r|}{0.00246} & \multicolumn{1}{r|}{0.00246} & \multicolumn{1}{r|}{0.00246} & \multicolumn{1}{l|}{\hspace{5.0mm} 0} \\ 
 \hline
\end{tabular}
\quad
\begin{tabular}{|l||r|}
\hline
\multicolumn{2}{|l|}{\it dataset properties} \\
\hline \hline
no.\ sequences & 50,000 \\
\hline
total length & 124,200,880 \\
\hline
average length & 2,483 \\
\hline
no.\ interesting intervals & 39,076 \\
\hline
total length intr.int.s & 913,721 \\
\hline
fraction pos.s in intr.int.s  & 0.007 \\
\hline
variability & 0.497 \\
\hline
\end{tabular}
\end{adjustbox}
\newline
\vspace*{0.4 cm}
\newline
\begin{adjustbox}{max width=140mm}
\setlength{\tabcolsep}{5pt}
\renewcommand{\arraystretch}{1.6}
\hspace{14cm}
\begin{tabular}{|l||r|r|}
\hline
\multicolumn{3}{|l|}{\it no.\ runs big dataset} \\
\hline \hline
& $r$ & $n/r$ \\
\hline
\eBWT & 72,014,777 & $1.724$ \\
\hline
\dollarebwt & 71,972,783 & $1.726$ \\
\hline
\mdollar & 71,972,346 & $1.726$ \\
\hline
\concat & 71,973,221 & $1.726$ \\ 
\hline
\rlo & 71,725,274 & $1.732$ \\
\hline
\optimal & 71,704,473 & $1.732$ \\
\hline
\end{tabular}
\end{adjustbox}
\newline
\vspace*{0.4 cm}
\newline
\begin{adjustbox}{max width=140mm}
\setlength{\tabcolsep}{5pt}
\renewcommand{\arraystretch}{1.6}
\begin{tabular}{|r||rrrr|}
  \hline
  \multirow{2}{*}{{\diagbox[width=4.7cm, height=1.18cm]{\rlap{\enspace\raisebox{0ex}{ \it \hspace{-3.0mm} norm.\ Hamming d.\  }}}{\raisebox{-0ex}{\it \hspace{-3.0mm} Hamming d.\ }}}} & \multicolumn{4}{l|}{\it Hamming distance on a subset of 1,500 sequences} \\
  \cline{2-5}
  & \multicolumn{1}{R{1.6cm}|}{\dollarebwt} & \multicolumn{1}{R{1.6cm}|}{\mdollar} & \multicolumn{1}{R{1.6cm}|}{\concat} & \multicolumn{1}{R{1.6cm}|}{\rlo} \\
  \hline \hline
 \multicolumn{1}{|l||}{\dollarebwt}& \multicolumn{1}{l|}{\hspace{5.0mm} 0} & \multicolumn{1}{r|}{6,333} & \multicolumn{1}{r|}{6,393} & \multicolumn{1}{r|}{6,260} \\
 \hline
 \multicolumn{1}{|l||}{\mdollar} & \multicolumn{1}{r|}{0.00169} & \multicolumn{1}{l|}{\hspace{5.0mm} 0}  & \multicolumn{1}{r|}{6,411} & \multicolumn{1}{r|}{6,354} \\ 
 \hline
 \multicolumn{1}{|l||}{\concat} & \multicolumn{1}{r|}{0.00170} & \multicolumn{1}{r|}{0.00171} & \multicolumn{1}{l|}{\hspace{5.0mm} 0} & \multicolumn{1}{r|}{6,294} \\ 
 \hline
 \multicolumn{1}{|l||}{\rlo} & \multicolumn{1}{r|}{0.00167} & \multicolumn{1}{r|}{0.00169} & \multicolumn{1}{r|}{0.00168} & \multicolumn{1}{l|}{\hspace{5.0mm} 0} \\ 
 \hline
\end{tabular}
\quad
\begin{tabular}{|l||r|}
\hline
\multicolumn{2}{|l|}{\it small dataset properties} \\
\hline \hline
no.\ sequences & 1,500 \\
\hline
total length & 3,755,776 \\
\hline
average length & 2,503 \\
\hline
no.\ of interesting intervals & 1,189 \\
\hline
total length intr.int.s & 18,372 \\
\hline
fraction pos.s in intr.int.s  & 0.005 \\
\hline
variability & 0.530 \\
\hline
\end{tabular}
\end{adjustbox}
\newline
\vspace*{0.4 cm}
\newline
\begin{adjustbox}{max width=140mm}
\setlength{\tabcolsep}{5pt}
\renewcommand{\arraystretch}{1.6}
\begin{tabular}{|r||r|r|r|r|r|}
  \hline
  \multirow{2}{*}{{\diagbox[width=4.0cm, height=1.18cm]{\rlap{\enspace\raisebox{0ex}{ \it \hspace{-3.0mm} norm.\ edit d.\  }}}{\raisebox{+0ex}{\it edit d.\ }} }} & \multicolumn{5}{l|}{\it edit distance on a subset of 1,500 sequences} \\
  \cline{2-6}
  & \multicolumn{1}{R{1.6cm}|}{\eBWT} &\multicolumn{1}{R{1.6cm}|}{\dollarebwt} & \multicolumn{1}{R{1.6cm}|}{\mdollar} & \multicolumn{1}{R{1.6cm}|}{\concat} & \multicolumn{1}{R{1.6cm}|}{\rlo} \\
  \hline \hline
 \multicolumn{1}{|l||}{\eBWT} & \multicolumn{1}{l|}{\hspace{5.0mm} 0} & \multicolumn{1}{r|}{30,345} & \multicolumn{1}{r|}{30,552} & \multicolumn{1}{r|}{30,562} & \multicolumn{1}{r|}{30,794} \\ 
 \hline
 \multicolumn{1}{|l||}{\dollarebwt}& \multicolumn{1}{r|}{0.00808} & \multicolumn{1}{l|}{\hspace{5.0mm} 0} & \multicolumn{1}{r|}{4,835} & \multicolumn{1}{r|}{4,860} & \multicolumn{1}{r|}{6,005}  \\ 
 \hline
 \multicolumn{1}{|l||}{\mdollar} & \multicolumn{1}{r|}{0.00813} & \multicolumn{1}{r|}{0.00129} & \multicolumn{1}{l|}{\hspace{5.0mm} 0}  & \multicolumn{1}{r|}{4,824} & \multicolumn{1}{r|}{6,105} \\ \cline{1-6}
 \multicolumn{1}{|l||}{\concat} & \multicolumn{1}{r|}{0.00814} & \multicolumn{1}{r|}{0.00129} & \multicolumn{1}{r|}{0.00128} & \multicolumn{1}{l|}{\hspace{5.0mm} 0} & \multicolumn{1}{r|}{6,035}  \\ 
 \hline
 \multicolumn{1}{|l||}{\rlo} & \multicolumn{1}{r|}{0.00820} & \multicolumn{1}{r|}{0.00160} & \multicolumn{1}{r|}{0.00163} & \multicolumn{1}{r|}{0.00161} & \multicolumn{1}{l|}{\hspace{5.0mm} 0} \\ 
 \hline
\end{tabular}
\quad
\begin{tabular}{|l||r|r|}
\hline
\multicolumn{3}{|l|}{\it no.\ runs small dataset} \\
\hline \hline
& $r$ & $n/r$ \\
\hline
\eBWT & 2,635,300 & $1.425$ \\
\hline
\dollarebwt & 2,633,676 & $1.426$ \\
\hline
\mdollar & 2,633,652 & $1.426$ \\
\hline
\concat & 2,633,727 & $1.426$ \\
\hline
\rlo & 2,629,094 & $1.429$ \\
\hline
\optimal & 2,628,470 & $1,429$ \\
\hline
\end{tabular}
\end{adjustbox}
\vspace{2mm}
\caption{\label{tab:candida_complete} Results for the Candida auris reads dataset. First row left: absolute and normalized pairwise Hamming distance between separator-based BWT variants. First row right: summary of the dataset properties. Second row: number of runs and average runlength ($n/r$) of all BWT variants. Third row left: absolute and normalized pairwise Hamming distance between separator-based BWT variants on a subset of the input collection. Third row right: summary of the dataset properties of a subset of the input collection. Fourth row left: absolute and normalized pairwise edit distance between all BWT variants on a subset of the input collection. Fourth row right: number of runs and average runlength ($n/r$) of all BWT variants on a subset of the input collection.}
\end{table*}

\begin{table*}[t!]
\raggedright
\textbf{SARS-CoV-2 genomes (2,000 long sequences)}

\vspace{4mm}

\captionsetup{width=\linewidth}
\begin{adjustbox}{max width=140mm}
\setlength{\tabcolsep}{5pt}
\renewcommand{\arraystretch}{1.6}
\begin{tabular}{|r||rrrr|}
  \hline
  \multirow{2}{*}{{\diagbox[width=4.7cm, height=1.18cm]{\rlap{\enspace\raisebox{0ex}{ \it \hspace{-3.0mm} norm.\ Hamming d.\  }}}{\raisebox{-0ex}{\it \hspace{-3.0mm} Hamming d.\ }}}} & \multicolumn{4}{l|}{\it Hamming distance on the big dataset} \\
  \cline{2-5}
  & \multicolumn{1}{R{1.6cm}|}{\dollarebwt} & \multicolumn{1}{R{1.6cm}|}{\mdollar} & \multicolumn{1}{R{1.6cm}|}{\concat} & \multicolumn{1}{R{1.6cm}|}{\rlo} \\
  \hline \hline
 \multicolumn{1}{|l||}{\dollarebwt}& \multicolumn{1}{l|}{\hspace{5.0mm} 0} & \multicolumn{1}{r|}{7,958} & \multicolumn{1}{r|}{7,900} & \multicolumn{1}{r|}{7,263} \\
 \hline
 \multicolumn{1}{|l||}{\mdollar} & \multicolumn{1}{r|}{0.00013} & \multicolumn{1}{l|}{\hspace{5.0mm} 0}  & \multicolumn{1}{r|}{7,958} & \multicolumn{1}{r|}{7,957} \\ 
 \hline
 \multicolumn{1}{|l||}{\concat} & \multicolumn{1}{r|}{0.00013} & \multicolumn{1}{r|}{0.00013} & \multicolumn{1}{l|}{\hspace{5.0mm} 0} & \multicolumn{1}{r|}{7,990} \\ 
 \hline
 \multicolumn{1}{|l||}{\rlo} & \multicolumn{1}{r|}{0.00012} & \multicolumn{1}{r|}{0.00013} & \multicolumn{1}{r|}{0.00013} & \multicolumn{1}{l|}{\hspace{5.0mm} 0} \\ 
 \hline
\end{tabular}
\quad
\begin{tabular}{|l||r|}
\hline
\multicolumn{2}{|l|}{\it dataset properties} \\
\hline \hline
no.\ sequences & 2,000 \\
\hline
total length & 59,612,692 \\
\hline
average length & 29,085 \\
\hline
no.\ interesting intervals & 1863 \\
\hline
total length intr.int.s & 80,486 \\
\hline
fraction pos.s in intr.int.s  & 0.001 \\
\hline
variability & 0.148 \\
\hline
\end{tabular}
\end{adjustbox}
\newline
\vspace*{0.4 cm}
\newline
\begin{adjustbox}{max width=140mm}
\setlength{\tabcolsep}{5pt}
\renewcommand{\arraystretch}{1.6}
\hspace{14cm}
\begin{tabular}{|l||r|r|}
\hline
\multicolumn{3}{|l|}{\it no.\ runs big dataset} \\
\hline \hline
& $r$ & $n/r$ \\
\hline
\eBWT & 117,628 & $506.773$ \\
\hline
\dollarebwt & 117,410 & $507.731$ \\
\hline
\mdollar & 118,870 & $501.495$ \\
\hline
\concat & 119,334 & $499.549$ \\
\hline
\rlo & 114,287 & $521.605$ \\
\hline
\optimal & 113,930 & $523.240$ \\
\hline
\end{tabular}
\end{adjustbox}
\newline
\vspace*{0.4 cm}
\newline
\begin{adjustbox}{max width=140mm}
\setlength{\tabcolsep}{5pt}
\renewcommand{\arraystretch}{1.6}
\begin{tabular}{|r||rrrr|}
  \hline
  \multirow{2}{*}{{\diagbox[width=4.7cm, height=1.18cm]{\rlap{\enspace\raisebox{0ex}{ \it \hspace{-3.0mm} norm.\ Hamming d.\  }}}{\raisebox{-0ex}{\it \hspace{-3.0mm} Hamming d.\ }}}} & \multicolumn{4}{l|}{\it Hamming distance on a subset of 50 sequences} \\
  \cline{2-5}
  & \multicolumn{1}{R{1.6cm}|}{\dollarebwt} & \multicolumn{1}{R{1.6cm}|}{\mdollar} & \multicolumn{1}{R{1.6cm}|}{\concat} & \multicolumn{1}{R{1.6cm}|}{\rlo} \\
  \hline \hline
 \multicolumn{1}{|l||}{\dollarebwt}& \multicolumn{1}{l|}{\hspace{5.0mm} 0} & \multicolumn{1}{r|}{105} & \multicolumn{1}{r|}{119} & \multicolumn{1}{r|}{90} \\
 \hline
 \multicolumn{1}{|l||}{\mdollar} & \multicolumn{1}{r|}{0.00007} & \multicolumn{1}{l|}{\hspace{5.0mm} 0}  & \multicolumn{1}{r|}{124} & \multicolumn{1}{r|}{116} \\ 
 \hline
 \multicolumn{1}{|l||}{\concat} & \multicolumn{1}{r|}{0.00008} & \multicolumn{1}{r|}{0.00008} & \multicolumn{1}{l|}{\hspace{5.0mm} 0} & \multicolumn{1}{r|}{118} \\ 
 \hline
 \multicolumn{1}{|l||}{\rlo} & \multicolumn{1}{r|}{0.00006} & \multicolumn{1}{r|}{0.00008} & \multicolumn{1}{r|}{0.00008} & \multicolumn{1}{l|}{\hspace{5.0mm} 0} \\ 
 \hline
\end{tabular}
\quad
\begin{tabular}{|l||r|}
\hline
\multicolumn{2}{|l|}{\it small dataset properties} \\
\hline \hline
no.\ sequences & 50 \\
\hline
total length & 1,490,184 \\
\hline
average length & 29,802 \\
\hline
no.\ interesting intervals & 43 \\
\hline
total length intr.int.s & 271 \\
\hline
fraction pos.s in intr.int.s  & $1.8 \cdot 10^{-4}$ \\
\hline
variability & 0.690 \\
\hline
\end{tabular}
\end{adjustbox}
\newline
\vspace*{0.4 cm}
\newline
\begin{adjustbox}{max width=140mm}
\setlength{\tabcolsep}{5pt}
\renewcommand{\arraystretch}{1.6}
\begin{tabular}{|r||r|r|r|r|r|}
  \hline
  \multirow{2}{*}{{\diagbox[width=4.0cm, height=1.18cm]{\rlap{\enspace\raisebox{0ex}{ \it \hspace{-3.0mm} norm.\ edit d.\  }}}{\raisebox{+0ex}{\it edit d.\ }} }} & \multicolumn{5}{l|}{\it edit distance on a subset of 50 sequences} \\
  \cline{2-6}
  & \multicolumn{1}{R{1.6cm}|}{\eBWT} &\multicolumn{1}{R{1.6cm}|}{\dollarebwt} & \multicolumn{1}{R{1.6cm}|}{\mdollar} & \multicolumn{1}{R{1.6cm}|}{\concat} & \multicolumn{1}{R{1.6cm}|}{\rlo} \\
  \hline \hline
 \multicolumn{1}{|l||}{\eBWT} & \multicolumn{1}{l|}{\hspace{5.0mm} 0} & \multicolumn{1}{r|}{786} & \multicolumn{1}{r|}{795} & \multicolumn{1}{r|}{801} & \multicolumn{1}{r|}{791} \\ 
 \hline
 \multicolumn{1}{|l||}{\dollarebwt}& \multicolumn{1}{r|}{0.00053} & \multicolumn{1}{l|}{\hspace{5.0mm} 0} & \multicolumn{1}{r|}{98} & \multicolumn{1}{r|}{107} & \multicolumn{1}{r|}{86}  \\ 
 \hline
 \multicolumn{1}{|l||}{\mdollar} & \multicolumn{1}{r|}{0.00053} & \multicolumn{1}{r|}{0.00007} & \multicolumn{1}{l|}{\hspace{5.0mm} 0}  & \multicolumn{1}{r|}{105} & \multicolumn{1}{r|}{112} \\ \cline{1-6}
 \multicolumn{1}{|l||}{\concat} & \multicolumn{1}{r|}{0.00054} & \multicolumn{1}{r|}{0.00007} & \multicolumn{1}{r|}{0.00007} & \multicolumn{1}{l|}{\hspace{5.0mm} 0} & \multicolumn{1}{r|}{114}  \\ 
 \hline
 \multicolumn{1}{|l||}{\rlo} & \multicolumn{1}{r|}{0.00053} & \multicolumn{1}{r|}{0.00006} & \multicolumn{1}{r|}{0.00008} & \multicolumn{1}{r|}{0.00008} & \multicolumn{1}{l|}{\hspace{5.0mm} 0} \\ 
 \hline
\end{tabular}
\quad
\begin{tabular}{|l||r|r|}
\hline
\multicolumn{3}{|l|}{\it no.\ runs small dataset} \\
\hline \hline
& $r$ & $n/r$ \\
\hline
\eBWT & 25,258 & $58.997$ \\
\hline
\dollarebwt & 25,255 & $59.006$ \\
\hline
\mdollar & 25,274 & $58.961$ \\
\hline
\concat & 25,285 & $58.936$ \\
\hline
\rlo & 25,221 & $59.085$ \\
\hline
\optimal & 25,210 & $59.111$ \\
\hline
\end{tabular}
\end{adjustbox}
\vspace{2mm}
\caption{\label{tab:SARS_genomes_complete} Results for the SARS-CoV-2 genomes dataset. First row left: absolute and normalized pairwise Hamming distance between separator-based BWT variants. First row right: summary of the dataset properties. Second row: number of runs and average runlength ($n/r$) of all BWT variants. Third row left: absolute and normalized pairwise Hamming distance between separator-based BWT variants on a subset of the input collection. Third row right: summary of the dataset properties of a subset of the input collection. Fourth row left: absolute and normalized pairwise edit distance between all BWT variants on a subset of the input collection. Fourth row right: number of runs and average runlength ($n/r$) of all BWT variants on a subset of the input collection.}
\end{table*}

\end{document}